\documentclass[reprint,amsfonts,amssymb,amsmath,showkeys,nofootinbib,aps,pra,superscriptaddress,twocolumn,longbibliography]{revtex4-2}

\usepackage[english]{babel}
\usepackage[normalem]{ulem} 
\usepackage[utf8]{inputenc}
\usepackage[T1]{fontenc}
\usepackage{amsthm}
\usepackage{mathtools}
\usepackage{physics}
\usepackage[svgnames]{xcolor}
\usepackage{graphicx}
\usepackage{adjustbox}
\usepackage{dsfont}

\usepackage{longtable}
\usepackage{multirow}

\colorlet{MAGENTA}{magenta}

\usepackage{bbm}
\newcommand{\unity}{\mathbbmss{1}}
\newcommand{\floor}[1]{\left\lfloor#1\right\rfloor}

\newcommand{\lie}[1]{\langle #1 \rangle_{\text{Lie}}}
\newcommand{\group}[1]{\langle #1 \rangle_{\text{group}}}
\newcommand{\alg}[1]{\langle #1 \rangle_{\C}}
\newcommand{\mult}{\mathrm{mult}}
\newcommand{\su}{\mathfrak{su}}
\newcommand{\uu}{\mathfrak{u}}
\newcommand{\so}{\mathfrak{so}}

\newcommand{\usp}{\mathfrak{sp}}
\newcommand{\had}{\mathrm{H}}

\newcommand{\cent}{\mathcal{Z}}
\newcommand{\gfree}{\mathfrak{g}_{\mathrm{free}}}
\newcommand{\gstd}{\mathfrak{g}_{\mathrm{std}}}
\newcommand{\gdot}{\mathfrak{h}_{\bullet}}
\newcommand{\gdotdot}{\mathfrak{h}_{\bullet}^{\bullet}}
\newcommand{\kdot}{\mathfrak{k}_{\bullet}}

\newcommand{\g}{\mathfrak{g}}
\newcommand{\s}{\mathfrak{s}}
\newcommand{\gh}{\mathfrak{h}}
\newcommand{\gk}{\mathfrak{k}}
\newcommand{\ga}{\mathfrak{a}}
\newcommand{\bipartite}{\mathfrak{k}_{G}^{k}}
\newcommand{\bipartiteH}{\mathfrak{k}_{\tilde{G}}^{k}}
\newcommand{\bipartiteW}{\mathfrak{k}_{G}^{W}}
\newcommand{\bipartiteWW}{\mathfrak{k}_{\tilde{G}}^{W}}

\newcommand{\com}{\mathfrak{com}}
\newcommand{\C}{\mathbb{C}}
\newcommand{\CNOT}{\mathrm{CNOT}}
\newcommand{\SWAP}{\mathrm{SWAP}}
\newcommand{\zero}{\mathit{0}}
\newcommand{\Gnat}{\mathbb{G}_{\text{nat}}}
\newcommand{\Gnatr}{\mathrm{G}_{\text{nat}}}
\newcommand{\Snat}{\mathcal{S}_{\text{nat}}}
\newcommand{\Bnat}{\mathcal{B}_{\text{nat}}}
\newcommand{\Bext}{\mathcal{B}_{\text{ext}}}
\newcommand{\Binv}{\mathcal{B}_{\text{inv}}}
\newcommand{\iso}{\cong}
\newcommand{\ZTWO}{\Z_2}

\newcommand{\pp}{\wp}
\newcommand{\PP}{\mathcal{P}}
\newcommand{\zz}{\mathbf{z}}

\usepackage{mathtools}

\newcommand\mbb[1]{\mathbb{#1}}
\newcommand{\ad}{^{\dagger}}

\usepackage[colorlinks=true,citecolor=blue,linkcolor=magenta]{hyperref}

\newcommand{\poly}{\operatorname{poly}}

\usepackage{xspace}
\newcommand{\yz}{\ensuremath{\left\{\!\begin{smallmatrix}\mathrm{Y}\\\mathrm{Z}\end{smallmatrix}\!\right\}}\xspace}
\newcommand{\cddot}{{\cdot\cdot}}

\newcommand{\Ubb}{\mathbb{U}}

\newcommand{\Z}{\mathbb{Z}}

\newcommand{\BC}{\mathcal{B}}
\newcommand{\CC}{\mathcal{C}}

\newcommand{\GC}{\mathcal{G}}
\newcommand{\HC}{\mathcal{H}}

\newcommand{\OC}{\mathcal{O}}

\newcommand{\Var}{{\rm Var}}

\renewcommand{\geq}{\geqslant}
\renewcommand{\leq}{\leqslant}

\DeclareMathOperator*{\argmin}{arg\,min}

\newcommand*{\id}{\openone}

\renewcommand{\th}{\theta } 
\newcommand{\vth}{\vartheta }

\newcommand{\lm}{\lambda }

\newcommand{\sg}{\sigma }

\newcommand{\thv}{\vec{\theta}}

\newcommand{\ubadm}{Departamento de Matemática, FCEN, UBA - IMAS CONICET.}

\newcommand{\losalamos}{Theoretical Division, Los Alamos National Laboratory, Los Alamos, New Mexico 87545, USA}
\newcommand{\courant}{Courant Institute of Mathematical Sciences, New York University, New York, New York 10012, USA}
\newcommand{\duke}{Department of Electrical and Computer Engineering, Duke University, Durham, NC 27708, USA}

\def\R{\mathds{R}}

\def\n{^{(n)}}
\def\tn{^{\otimes n}}

\newcommand{\spn}{\text{span}}
\newcommand{\spnC}{\spn_{\C}}
\newcommand{\spnR}{\spn_{\R}}

\theoremstyle{definition}

\makeatletter
\newtheorem*{rep@theorem}{\rep@title}
\newcommand{\newreptheorem}[2]{%
\newenvironment{rep#1}[1]{%
 \def\rep@title{#2 \ref{##1}}%
 \begin{rep@theorem}}%
 {\end{rep@theorem}}}
\makeatother

\newtheorem{theorem}{Theorem}

\newtheorem{fact}{Fact}
\newtheorem{lemma}{Lemma}
\newtheorem{corollary}{Corollary}
\newreptheorem{corollary}{Corollary}
\newreptheorem{theorem}{Theorem}

\newtheorem{conjecture}{Conjecture}
\newtheorem{proposition}[lemma]{Proposition}
\newtheorem{example}{Example}
\newtheorem{definition}{Definition}

\newcommand{\free}{\mathrm{free}}
\newcommand{\aut}{\mathrm{aut}}
\newcommand{\orb}{\mathrm{orbit}}
\newcommand{\std}{\mathrm{std}}

\newcommand{\triv}{\text{t}}

\newcommand{\signrepr}{\text{s}}

\newcommand{\gnat}{\g_{\mathrm{nat}}}
\newcommand{\unat}{\uu_{\mathrm{nat}}}
\newcommand{\taunat}{\tau_{\mathrm{nat}}}
\newcommand{\tauhat}{\hat{\tau}_{\mathrm{nat}}}
\newcommand{\tauaut}{\tau_{\aut}}
\newcommand{\gorb}{\g_{\orb}}

\newcommand{\VV}{\mathcal{V}}
\newcommand{\E}{\mathcal{E}}
\newcommand{\Cstd}{\CC_{\std}}
\newcommand{\Cfree}{\CC_{\free}}
\newcommand{\DLA}{Lie algebra\xspace}
\newcommand{\DLAs}{Lie algebras\xspace}
\newcommand{\DLAc}{Lie-algebraic\xspace}
\newcommand{\DLAm}{Lie-algebra\xspace}

\newcommand{\Aut}{\mathrm{Aut}(G)}
\newcommand{\AUT}{\mathrm{Aut}}
\newcommand{\symone}{\mathbf{1}}

\renewcommand{\Tr}{\mathrm{Tr}}

\begin{document}

\title{\texorpdfstring{Analyzing the quantum approximate optimization algorithm:\\ ansätze, symmetries, and Lie algebras}{Analyzing the quantum approximate optimization algorithm: ansätze, symmetries, and Lie algebras}}

\author{Sujay Kazi\,\href{https://orcid.org/0000-0002-9404-8376}{\includegraphics[height=5pt]{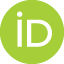}}}
\affiliation{\courant}
\affiliation{\losalamos}
\affiliation{\duke}

\author{Mart\'{i}n Larocca\,\href{https://orcid.org/0000-0002-8700-4308}{\includegraphics[height=5pt]{ORCID-iD_icon-64x64.png}}}
\email{larocca@lanl.gov}
\affiliation{\losalamos}
\affiliation{Center for Nonlinear Studies, Los Alamos National Laboratory, Los Alamos, New Mexico 87545, USA}

\author{Marco Farinati\,\href{https://orcid.org/0000-0003-4307-8505}{\includegraphics[height=5pt]{ORCID-iD_icon-64x64.png}}}
\affiliation{\ubadm}

\author{Patrick J. Coles\,\href{https://orcid.org/0000-0001-9879-8425}{\includegraphics[height=5pt]{ORCID-iD_icon-64x64.png}}}
\affiliation{Normal Computing Corporation, New York, New York, USA}
\affiliation{\losalamos}

\author{M. Cerezo\,\href{https://orcid.org/0000-0002-2757-3170}{\includegraphics[height=5pt]{ORCID-iD_icon-64x64.png}}}
\email{cerezo@lanl.gov}
\affiliation{Information Sciences, Los Alamos National Laboratory, Los Alamos, NM 87545, USA}

\author{Robert Zeier\,\href{https://orcid.org/0000-0002-2929-612X}{\includegraphics[height=5pt]{ORCID-iD_icon-64x64.png}}}
\email{r.zeier@fz-juelich.de}
\affiliation{Forschungszentrum J\"ulich GmbH, Peter Gr\"unberg Institute, Quantum Control (PGI-8), 54245 J\"ulich, Germany}

\begin{abstract}
The Quantum Approximate Optimization Algorithm (QAOA) has been proposed as a method to obtain approximate solutions for combinatorial optimization tasks. In this work, we study the underlying algebraic properties of three QAOA ans\"atze for the maximum-cut (maxcut) problem on connected graphs, while focusing on the generated Lie algebras as well as their invariant subspaces. Specifically, we analyze the standard QAOA ansatz as well as the orbit  and the multi-angle ans\"atze. We are able to fully characterize the Lie algebras of the multi-angle ansatz across arbitrary connected graphs, finding that they only fall into one of just six families. Besides the cycle and the path graphs, the Lie dimensions for every graph are exponentially large in the system size, meaning that multi-angle ans\"atze are extremely prone to exhibiting barren plateaus. Then, a similar quasi-graph-independent Lie-algebraic characterization beyond the multi-angle ansatz is impeded as the circuit exhibits additional ``hidden'' symmetries besides those naturally arising from a certain parity-superselection operator and all automorphisms of the considered graph. Disregarding the ``hidden'' symmetries, we can upper bound the dimensions of the orbit and the standard Lie algebras, and the dimensions of the associated invariant subspaces are determined via explicit character formulas. To finish, we conjecture that (for most graphs) the standard Lie algebras have only components that are either exponential or that grow, at most, polynomially with the system size. This would imply that the QAOA is either prone to barren plateaus, or classically simulable. More generally, our work provides a symmetry framework
and tools to analyze any desired variational quantum algorithm.
\end{abstract}

\date{July 11, 2025}

\maketitle

\section{Introduction\label{sec:introduction}}

Variational models, such as the Quantum Approximate Optimization Algorithm (QAOA)~\cite{farhi2014quantum,farhi2016quantum}, hold the promise to make practical use of intermediate-scale quantum computers~\cite{cerezo2020variationalreview,bharti2021noisy,endo2021hybrid,schuld2015introduction,biamonte2017quantum}. At their core, these schemes train a parametrized quantum circuit to minimize a loss function encoding the solution to a given task of interest. For example, QAOA aims at finding approximate solutions to the maximum-cut (maxcut) problem by encoding the edges of a graph 
as interaction terms in an Ising-type Hamiltonian, and training a circuit to prepare the ground states of said Hamiltonian. 

\begin{figure*}[t]
\centering
\includegraphics[width=.9\linewidth]{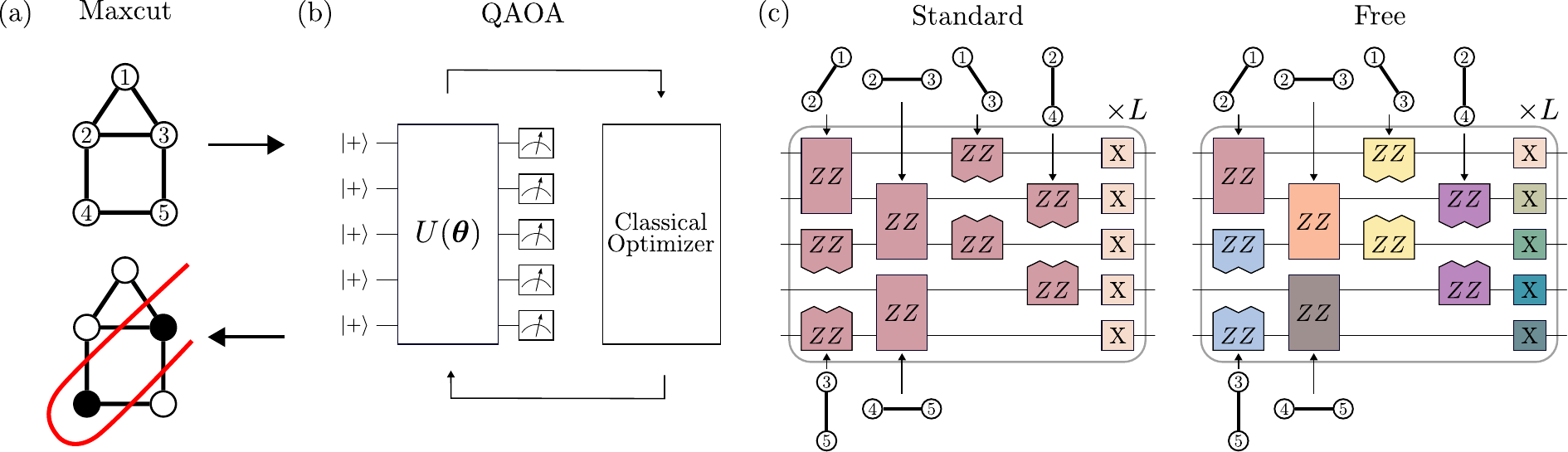}
\caption{\textbf{Maxcut, QAOA and ans\"atze.}  (a) Given a graph, the maxcut problem is to determine a partition of the vertices into two complementary sets, 
such that the number of edges between those sets is as large as possible. (b) The QAOA algorithm is a hybrid quantum-classical algorithm that can be used to approximately
solve the maxcut problem. The success of QAOA hinges on the ability to optimize the parametrized quantum circuit $U(\thv)$. Crucially, the trainability of QAOA can be linked to certain algebraic properties of $U(\thv)$.  (c) Here we depict two of the considered QAOA ansätze: the standard
and the multi-angle (or free) ansätze [see, e.g., Figs.~\ref{fig:ex:house} and \ref{fig:ex:house:orbit}].
In the image, a gate with $ZZ$ indicates a two-qubit entangling gate generated by a $Z_w Z_{\tilde{w}}$ interaction, while the $X$ gate indicates a single-qubit rotation around the $x$-axis.
Boxes with the same color share the same parameter. Hence,  each gate in the free ansatz 
is individually parametrized and generated by a single Pauli operator, which makes its \DLA tractable across all graphs. In the standard case,  all single-qubit gates share the same parameter, and similarly for all two-qubit gates. This means that the 
infinitesimal generators of the circuit
are given by a sum of Paulis, which limits the ability to treat these cases.     
\label{fig:schematic}}
\end{figure*}

One of the critical design choices that can make or break a variational model is the choice of an ansatz for the parametrized quantum circuit~\cite{cerezo2022challenges}. Indeed, it has been shown that certain circuit architectures can lead to trainability barriers such as exponentially suppressed loss function gradients~\cite{larocca2024review,mcclean2018barren,cerezo2020cost} (i.e., barren plateaus), or optimization landscapes plagued with local-minima~\cite{anschuetz2022quantum,anschuetz2021critical,bittel2021training,you2021exponentially,fontana2022nontrivial,kossmann2023,rajakumar2024trainability,bermejo2024improving}. Importantly, given that we do not yet posses large-scale quantum computers to heuristically test performance at scale, it is crucial to look for theoretical characterizations capable of predicting whether an ansatz  will run into issues or not. One such analysis is the study of the generated \DLA~\cite{zeier2011symmetry}, defined as the Lie closure of the 
infinitesimal generators of 
the parametrized gates. Importantly, the \DLA of the circuit captures the ultimate breadth of unitaries that can be expressed via different parameter choices, and its precise characterization is an extremely powerful tool as it enables the study of barren plateaus~\cite{larocca2021diagnosing,ragone2023unified,fontana2023theadjoint,diaz2023showcasing}, overparametrization~\cite{larocca2021theory,schatzki2022theoretical}, as well as the classical simulability of the model~\cite{goh2023lie}.

Given the critical importance of the \DLA, it comes at no surprise that its analysis has gained significant
attention~\cite{zeier2011symmetry,zimboras2015symmetry,albertini2018controllability,lloyd2018quantum,
morales2020universality,kazi2023universality,wiersema2023,DAlessandro2024,aguilar2024full,allcock2024dla,
bakalov24,mansky2023permutation,mansky2024scaling,kokcu2021fixed}. Still, a careful inspection of the literature
reveals that many \DLAm studies are restricted to two cases: circuits whose gates are generated by single Pauli
operators, or ans\"atze where generators are very structured linear combinations of Paulis, for example invariant
under some symmetry group. The first case corresponds to circuits where each rotation is individually parametrized,
while the second one contains circuits with correlated gate angles. The reason why the \DLAm analysis is restricted
to the aforementioned cases is that the computation of the \DLA  becomes extremely challenging for arbitrary
generators consisting of linear combinations of Paulis.

In this work, we contribute to the literature of symmetry and \DLAm characterization of quantum circuits by studying the algebraic properties of three QAOA ansätze for maxcut. These include the standard ansatz introduced in the original QAOA manuscript~\cite{farhi2014quantum}. Here, there are only two generators where certain Paulis are summed over all the edges and vertices of the graph, respectively. Then, we consider the orbit ansatz proposed in~\cite{sauvage2022building} where the generators are instead defined  by summing operators according to the 
automorphism group of the graph. Finally, we also study the multi-angle~\cite{herrman2022angle,shi2022multi} (or free) ansatz, where one assigns a single parameter to each gate in the standard ansatz. 

In the case of multi-angle QAOA ansatz, we give a complete characterization of the \DLA for any graph, finding that they fall within one of six families. In particular, in all graph families (except for the cycle and the path graphs) the dimension of the \DLA grows exponentially with the number of vertices in the graph. We then analyze the implications of these results, and show that multi-angle case is prone to exhibiting barren plateaus, even when using a single layer of the ansatz.

In the orbit and the standard cases, we argue that a full characterization of the \DLA across all graphs is likely quite challenging
(with some notable exceptions as the path, cycle and complete graphs, for which we characterize the \DLA, see also~\cite{Onsager44,DAlessandro2024,allcock2024dla}). To this end, we take a closer look at the symmetries of the \DLA~\cite{zeier2011symmetry}, i.e., the set of operators that commute with the parametrized unitary. Firstly, we show how the symmetries of the graph get promoted to symmetries at the quantum level, and find that while the standard and the orbit ansätze respect them, the multi-angle ansatz does not (and this is precisely why we can characterize it so well). Then, we show that the orbit and the standard ansätze exhibit additional symmetries 
beyond the ``natural'' symmetries
arising from the parity-superselection operator $X\tn$ and the automorphisms of the considered graph.
These ``hidden'' symmetries make the \DLA highly graph dependent and harder to study. 
We can associate the natural symmetries to a ``natural'' \DLA that
is semi-universal~\cite{marvian2023non,kazi2023universality} and its dimension provides an upper bound for the
dimensions of the orbit and the standard \DLAs. Moreover, the dimensions of the invariant subspaces associated to
the ``natural'' symmetries are determined via explicit character formulas.
Finally, we conjecture that for the vast majority of graphs that the dimensions of the largest invariant component for
the natural \DLA and the standard \DLA only differ by a polynomial factor. If this is the case, then our work has important implications regarding the trainability and classical simulability of QAOA.   
In summary, we develop a symmetry framework and much needed tools to systematically analyze
QAOA and general variational quantum algorithms.

We also point to related work \cite{boulebnane2021,Basso2022TQC,Basso2022TQC_arXiv,boulebnane}
applicable to large-girth $D$-regular graphs which
establishes concentration results for the loss function with respect to variations in the
QAOA angles 
(see Sec.~\ref{sec:discussion}).
These concentration results
have been obtained in the context
of smart initialization techniques \cite{larocca2024review,zhou2020quantum,FarhiQuantum2022} and
follow as the optimal angles for the problem Hamiltonian
are rescaled with $1/\sqrt{D}$ as $D$ increases. 
In contrast, we aim at developing tools to analyze general connected graphs.

\section{Framework\label{sec:frame}}

\subsection{From maxcut to QAOA\label{sect:maxcut:qaoa}}

We recall the maxcut problem 
\cite{mooremertens2011,commander2008,newman2008,garey1979guide} and
define notation that will be used throughout the manuscript. An undirected, unweighted, graph $G$ (without loops) \cite{diestel2017} is defined by its vertex set $V=\{1,\ldots,n\}$ and its edge set $E$ consisting of 
unordered pairs $\{w,\tilde{w}\}$ of vertices $w,\tilde{w} \in V$ with $w\neq\tilde{w}$. 
Given a graph $G$, the (optimization variant of the) maxcut problem is to find a  partition of its vertices into two complementary sets
such that the number of edges between those sets is as large as possible [see Fig.~\ref{fig:schematic}(a)]. 
Mathematically, this is formalized by noting that a vertex subset $W \subset V$
determines a bipartition (or cut) of the vertices of a graph into 
the disjoint sets $W$ and $V\setminus W$. A natural way of quantifying the cut value is by assigning to each vertex a value of $0$ or $1$, depending on the subset it belongs to. Therefore, we represent a cut through a bitstring $x\in \{0,1\}^n$ with cut value
$\sum_{\{w,\tilde{w}\} \in E} \abs{x_w -x_{\tilde{w}}}$. An optimal solution 
to the maxcut problem is identified by a maximum cut that 
observes the maximum cut value for the considered graph. 
Despite the seemingly simplicity of this task, finding the maximum cut value is known to be NP-hard \cite{karp1972,garey1976simplified}.

Due to the NP-hardness of maxcut, there has been significant interest in classical approximation algorithms \cite{dezalaurent1997,mayr1998,vazirani2001,williamson2011}, including randomized ones.
Let $c \in (0,1]$ denote the expected \emph{approximation ratio} of an algorithm guaranteeing that the
(expected) cut value of the computed cut is lower bounded by $c$  times the maximum cut value
of the considered graph.
Goemans and Williamson \cite{goemans1995improved} devised a randomized classical algorithm based on semidefinite programming
which guarantees that the expected approximation ratio is at least $c_{\text{GW}} \approx 0.87856$. As shown in \cite{karloff1999,feige2002},
there are graphs for which the expected approximation ratio obtained via \cite{goemans1995improved} saturates the lower bound $c_{\text{GW}}$.
Moreover, approximating maxcut with an approximation ratio of $c \geq 16/17$ is NP-hard
\cite{hastad2001}.

Quantum computers can also be used to find approximate solutions to the maxcut problem~\cite{farhi2014quantum,hadfield2019quantum,harrigan2021quantum}.
Obtaining the maximum cut value can be mapped to finding the ground state energy of the $n$-qubit  Ising Hamiltonian 
\begin{equation}\label{eq:prob-Ham}
    H_p:= \sum_{\{w,\tilde{w}\} \in E} Z_{w} Z_{\tilde{w}},
\end{equation}
where $Z_k$ denotes the Pauli-$z$ operator acting on the $k$-th qubit.
One can therefore attempt to obtain the maxcut variationally  by defining the cost function
\begin{equation}\label{eq:cost}
C(\thv)=\langle \psi(\thv)|H_p|\psi(\thv)\rangle\,,
\end{equation}
where $|\psi(\thv)\rangle$ is a quantum state parametrized by  $\thv$, and
solving the optimization task 
$\argmin_{\thv}C(\thv)$.
In QAOA, the state $|\psi(\thv)\rangle$ is obtained 
as follows [see Fig.~\ref{fig:schematic}(b)]: First, one initializes $n$ qubits to the fiduciary state
\begin{equation}\label{eq:fiduciary}
\ket{+}^{\otimes n} =\frac{1}{\sqrt{2^n}} \sum_{x\in \{0,1\}^n} \ket{x},
\end{equation}
where $\ket{x}$ denotes the computational basis state
in an $n$-qubit Hilbert space $\HC:=(\mbb{C}^2)\tn$. Henceforth, $d:=\dim(\HC)=2^n$ denotes the dimension of $\HC$.
The initial state is sent through a parametrized quantum circuit with $L$ layers of the form
\begin{equation}\label{eq:circuit}
U(\thv) = \prod_{\ell=1}^L \hspace{-1mm} \left[
\prod_{k=1}^{\abs{\mathcal{G}_{m}}}
e^{-i \theta_{m\ell k}  H_{mk} } \hspace{-0.4mm} \right] \hspace{-1.5mm}  \left[\prod_{k=1}^{\abs{\mathcal{G}_{p}}}
e^{-i \theta_{p\ell k} H_{pk} } \hspace{-0.4mm} \right]
\end{equation}
for suitable sets $\mathcal{G}_p$ and $\mathcal{G}_m$
of hermitian generators
$H_{pk} \in \mathcal{G}_p$ and $H_{mk} \in \mathcal{G}_m$.
The real parameters $\theta_{p\ell k}$ and $\theta_{m\ell k}$ are collected in the 
vector $\thv$.
Then,  $\ket*{\psi(\thv)}=U(\thv)|+\rangle^{\otimes n}$ is
measured in the computational basis, and the measurement outcomes are classically post-processed to estimate the cut value.
This is the starting point for the hybrid quantum-classical optimization as illustrated in Fig.~\ref{fig:schematic}.

\begin{figure}[t]
\includegraphics{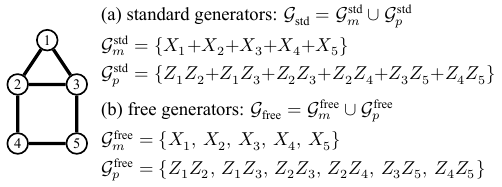}
\caption{\textbf{Example of house graph.} Generators for the (a) standard and (b) the free ansatz. 
\label{fig:ex:house}}
\end{figure}

In the original QAOA manuscript \cite{farhi2014quantum}, the 
\emph{standard} generator sets are
[see Fig.~\ref{fig:ex:house}(a)]
\begin{subequations}
\label{eq:std}
\begin{align}
\mathcal{G}_p^{\std} & :=\{H_p\} \;\text{ with }\; H_p=\sum_{\{w,\tilde{w}\} \in E} Z_{w} Z_{\tilde{w}}, \label{eq:std:problem}\\
\mathcal{G}_m^{\std} & :=\{H_m\} \;\text{ with }\; H_m:=\sum_{v\in V} X_v. \label{eq:std:mixer}
\end{align}
\end{subequations}
As we can see, in the standard ansatz, only two parameters are used per layer. One of the drawbacks of this approach is that it hinders the circuit's expressiveness and one usually requires a large number $L$ of layers to achieve a good approximation ratio.
We detail two ansätze with more parameters per layer. To mitigate this issue, and pack more parameters per later (thus making QAOA better suited for near-term quantum computers) several modifications to the QAOA generators have been proposed. First, we consider the  the  multi-angle~\cite{herrman2022angle,shi2022multi,gaidai2023performance,wilkie2024angle}, 
or \emph{free} ansatz, the sets of generators of which are [see Fig.~\ref{fig:ex:house}(b)]
\begin{subequations}
\label{eq:free}
\begin{align}
\mathcal{G}_p^{\free} &:=\{Z_{w}Z_{\tilde{w}} \;\text{for}\; \{w,\tilde{w}\} \in E\}, \label{eq:free:problem}\\
\mathcal{G}_m^{\free} &:=\{X_v \;\text{for}\; v \in V\}. \label{eq:free:mixer}
\end{align}
\end{subequations}
This approach assigns one parameter per layer for each edge and vertex.
Occasionally, a variant of the free ansatz with $\tilde{\mathcal{G}}_p^{\free} := \mathcal{G}_p^{\std} = \{H_p\}$
is experimentally more suitable.
Then, in the orbit ansatz of~\cite{sauvage2022building} not all parameters are independent, but rather they are correlated following orbits of the automorphism group of a graph. The orbit ansatz is then a middle ground between the standard and the free one in terms
of the number of generators
and we analyze the orbit ansatz in Sec.~\ref{sec:standard-ansatz-partial-results} as an upper bound to the 
standard ansatz.

Given the above freedom in choosing an ansatz for the QAOA, a natural question that arises is: \emph{Which one should we choose?} 
Clearly, having more trainable parameters at the same depth is an appealing idea. However,  increasing 
the expressiveness of an ansatz (for instance by adding experimentally more independent parameters) can also lead to barren plateaus and trainability
issues~\cite{holmes2021connecting,larocca2021diagnosing,larocca2024review}. 
In the following subsection, we explore symmetry methods for analyzing the free, standard, orbit ansätze and their expressiveness through
the lens of \DLAs.

\subsection{\DLAs and symmetries\label{sec:dla}}

For ease of notation, let us focus here on a  
general variational quantum algorithms with $L$ layers  described by a unitary
\begin{equation}\label{eq:circuit:general}
U(\thv)=\prod_{\ell=1}^L \prod_{k=1}^{\abs{\mathcal{G}}} e^{-i \theta_{\ell k} H_k}\,.
\end{equation}
Here, 
$\mathcal{G}$ is the set
of hermitian generators $H_k \in \mathcal{G}$, and 
$\theta_{\ell k}$ the trainable real parameters that we collect in the vector $\thv\in \R^{L{\times} \abs{\GC}}$. Note that this case recovers as a special case the unitary in Eq.~\eqref{eq:circuit} with $\mathcal{G}=\mathcal{G}_p \cup \mathcal{G}_m$. 
In what follows we now outline different symmetry methods that can be used to analyze variational quantum algorithms. In particular, we illustrate these techniques in Table~\ref{table:house:graph} for a 
low-dimensional computational example.
Further details are deferred to Appendix~\ref{appendix:symmtry:analysis}.

\begin{table}
\caption{\textbf{Symmetry analysis for the house graph.}
Based on the given generators, 
the commutant $\CC$ and its center $\cent(\CC)$ determine
the invariant subspaces, but, in general, 
\emph{not} the generated \DLA $\g$
(refer to Sec.~\ref{sec:dla} and Appendix~\ref{appendix:symmtry:analysis}).
For the standard ansatz, the center $\cent(\g)$ of $\g$ has support on all invariant subspaces
(refer also to Sec.~\ref{sec:standard-ansatz-partial-results}).
 \label{table:house:graph}}
\includegraphics{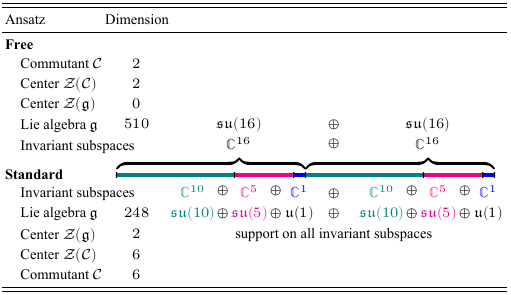}
\end{table}

As a first analysis tool,
the (real) \emph{\DLA} $\g = \lie{i\mathcal{G}}$
is the Lie closure of its generators $iH_k\in i\mathcal{G}$, i.e.,
it contains 
all real-linear combinations of repeated commutators of $iH_1, \ldots, iH_{\abs{\mathcal{G}}}$
\cite{hall2015,jacobson1979,deGraaf2000}.
Recall that the commutator of $A,B \in \C^{d\times d}$ is given by $[A,B]:=AB-BA$.
Thus $\g \subseteq \uu(d) \subseteq \C^{d\times d}$ is a subalgebra of the unitary \DLA $\uu(d)$ of skew-hermitian matrices. In our context, the \DLA is sometimes denoted as the \emph{dynamical} \DLA
\cite{dalessandro2022} in order to emphasize its relation to the unitary time evolution.
The corresponding \emph{Lie group} $\exp(\g) \subseteq \Ubb(d)$ is then contained in the unitary group $\Ubb(d)$.
The \DLA $\g$ plays a key role as it determines reachability properties
\cite{zeier2011symmetry,elliott2009bilinear,dalessandro2022,larocca2021diagnosing}, e.g., it characterizes the set of
quantum states  $\exp(\g)\ket{\psi}$ that are accessible from
a given initial state $\ket{\psi}$, when evolved by a unitary as in Eq.~\eqref{eq:circuit:general}.

As a second analysis tool, the linear symmetries of the generators $\mathcal{G}$ are given by the 
\emph{commutant}
\begin{equation}
\CC = \com(\mathcal{G})
= \{ S \in \C^{d\times d} \;\text{s.t.}\; [S,H_k]{=}0 \;\text{for all}\; H_k \in \mathcal{G}\}.
\label{eq:commutant}
\end{equation}
Note that the commutant $\com(\g)=\com(i\mathcal{G})=\com(\mathcal{G})$ of the \DLA 
$\g$ is equal to the commutant of its generators $i\mathcal{G}$. Also,
$\CC$ is closed under complex-linear combinations and matrix multiplication, 
thus constituting an associative subalgebra (under matrix multiplication) of $\mbb{C}^{d\times d}$.
The commutant $\CC$ and its center
\begin{equation}
\cent(\CC) = \{ Z \in \CC \;\text{s.t.}\; [Z,S]=0 \;\text{for all}\; S \in \CC\}
\label{eq:center:commutant}
\end{equation}
determine the \emph{invariant subspaces} of $\HC = \C^d$ under the action of $\mathcal{G}$,
i.e., $\mathcal{I} \subseteq \HC$ is an invariant subspace if $H_k  v \in \mathcal{I}$ for all $v \in \mathcal{I}$ and $H_k \in \mathcal{G}$.
And, the invariant subspaces induced by the action
of $\mathcal{G}$,  $\g$, and $\exp(\g)$ 
all agree. 
Elements of the commutant correspond to conserved quantities, 
since $\bra{\psi}e^{itH_j}Se^{-itH_j}\ket{\psi} = \bra{\psi}S\ket{\psi}$,
meaning that the expectation of $S$ is constant under the time evolution by any of the generators $H_j$.

An invariant subspace $\mathcal{J} \neq 0$ is \emph{irreducible}
if it contains no invariant subspaces
distinct from $0$ and $\mathcal{J}$ itself. 
The dimension $\dim[\cent(\CC)]$ of the center $\cent(\CC)$ specifies 
how many inequivalent irreducible subspaces occur
(see Appendix~\ref{appendix:symmtry:analysis}).
The invariant subspaces are revealed  
by a basis change so that the block decomposition
\begin{equation}\label{invariants_basis}
\CC \simeq \hspace{-3mm} \bigoplus_{\lambda=1}^{\dim[\cent(\CC)]} \hspace{-3mm}\C^{m_{\lambda}\times m_{\lambda}} \otimes \unity_{d_{\lambda}}
\;\text{ and }\;
\g \simeq  \hspace{-3mm} \bigoplus_{\lambda=1}^{\dim[\cent(\CC)]}  \hspace{-3mm} \unity_{m_{\lambda}} \otimes \g_{\lambda}\,,
\end{equation}
holds for certain \DLAs $\g_{\lambda} \subseteq \uu(d_{\lambda}) \subseteq \C^{d_{\lambda}\times d_{\lambda}}$ with
$\dim[\com(\g_{\lambda})]=1$. The commutant $\CC$ fixes
the dimensions $d_{\lambda} 
\geq 1$ and multiplicities $m_{\lambda} \geq 1$,
but does \emph{not} uniquely determine
the \DLA $\g$ (or the $\g_{\lambda}$).
For completeness, the corresponding decomposition 
\begin{equation*}
\HC\simeq \hspace{-3mm} \bigoplus_{\lambda=1}^{\dim[\cent(\CC)]}   \hspace{-3mm} \C^{m_\lm} \otimes \HC_\lm\,,
\end{equation*}
of the state space 
under the action of $\g$ and $\exp(\g)$ identifies
the subspaces $\HC_{\lambda} = \C^{d_{\lambda}}$ as invariant and irreducible.

As a third analysis tool, the \DLA $\g\subseteq \uu(d)$ observes
its \emph{reductive} decomposition \cite{BourbakiLie1989,Bourbaki2008b}
\begin{equation}\label{eq:reductive}
\g \iso \cent(\g) \oplus [\oplus_j \s_j] \iso [\com(\g) {\cap} \g] \oplus [\oplus_j \s_j]
\end{equation}
which includes its center 
$\cent(\g)$ and its semisimple part $\oplus_j \s_j$ consisting of a direct sum
of simple \DLAs $\s_j$.
The center $\cent(\g)$ can be readily obtained from
the commutant $\CC=\com(\g)$ and the generators $\mathcal{G}$.
But identifying the simple \DLAs $\s_j$ (such as $\su(r)$, $\so(r)$, and $\usp(r)$
for suitable $r \geq 2$ \cite{hall2015})
is usually more challenging, both for computations \cite{deGraaf2000,roozemond2010,zeier2011symmetry,magma1997}
and analytic approaches. The reductive decomposition of Eq.~\eqref{eq:reductive} can be either interpreted 
as an intrinsic vector-space decomposition of $\g$ with $[\mathfrak{s}_{j_1}, \mathfrak{s}_{j_2}]=0$ or
as being extrinsically embedded into $\g \subseteq \uu(d) \subseteq \C^{d\times d}$.

This embedding can be specified by an explicit representation which constitutes 
our fourth analysis tool. A \emph{representation} $\gamma$ of $\g$
is a linear map from
$\g$ to $\uu(m)\subseteq\C^{m\times m}$ with $m\geq 1$
such that $\gamma([g_1,g_2])= [\gamma(g_1),\gamma(g_2)]$ for $g_1,g_2 \in \g$ \cite{hall2015}. 
Notable examples are the standard representation $\kappa$ with $\kappa(g)=g$ for $g \in \g \subseteq \uu(d)$,
the trivial representation $\epsilon$ with $\epsilon(g)=0 \in \C$, 
and the dual $\overline{\gamma}$ of $\gamma$
with $\overline{\gamma}(g) = - [\gamma(g)]^t$.
Identifying a representation precisely characterizes the occurring symmetries but this is usually
more difficult to achieve.

Table~\ref{table:house:graph}
highlights significant differences between
the free and standard ansatz for the example of the house graph.
For the free ansatz, the commutant and its center are both two dimensional
which results into two inequivalent irreducible subspaces.
The commutant and center for the standard ansatz are, however, both six dimensional
and one obtains six inequivalent irreducible subspaces which refine the ones
for the free ansatz.
We will further explore the free ansatz in Sec.~\ref{SEC:FREE-ANSATZ}, while Sections~\ref{SEC:STD}-\ref{sec:lie:hierarchy} start the discussion of the standard ansatz.

\section{The Multi-angle or Free Ansatz\label{SEC:FREE-ANSATZ}}

We now apply the framework established in Sec.~\ref{sec:frame}  to the free ansatz of QAOA which is based on the 
problem and mixer
Hamiltonians in $\mathcal{G}_p^{\free}$ and $\mathcal{G}_m^{\free}$
from Eq.~\eqref{eq:free}. These generators define the 
\emph{free-mixer \DLA}
\begin{equation*}
\gfree := \lie{i\mathcal{G}_{\text{free}}} 
\subseteq  \C^{2^n\times 2^n}
 \text{ where } \mathcal{G}_{\text{free}} := \mathcal{G}_p^{\free} \cup \mathcal{G}_m^{\free}.
\end{equation*}
As a first observation, all generators $H_j \in \mathcal{G}_{\text{free}}$ commute
with $X^{\otimes n}$ (and trivially with $\unity_{2^n}$). Here, $X^{\otimes n}$ acts as bit-flip (or $\mathbb{Z}_2$) symmetry, i.e.,
$X\tn\ket{x}=\ket{\neg x}$ where $\neg$ is the bitwise NOT operation.
Alternatively,  $X^{\otimes n}$ can be interpreted as a parity superselection operator.
The states $\ket{x}$ and $\ket{\neg x}$ 
then have the same expectation value relative to $H_p$.
We show in Lemma~\ref{app:prop:center}(a) of
Appendix~\ref{app:free:preliminaries} that the commutant
\begin{equation}\label{eq:free-commutant}
\Cfree := \com(\mathcal{G}_{\text{free}}) = \spn_{\C}\{I\tn,X^{\otimes n}\}
\end{equation}
of all matrices simultaneously commuting with all $H_j \in \mathcal{G}_{\text{free}}$
consists solely of complex-linear combinations of $I\tn$ and $X^{\otimes n}$.
Clearly, $X^{\otimes n}$ has two distinct eigenvalues $+1$ and $-1$ which both have multiplicity $d/2=2^{n-1}$ and
it is equivalent to
\begin{equation*}
X^{\otimes n}\simeq Z_1 = \id_{d/2}\oplus(-\id_{d/2})
\end{equation*}
by switching to an eigenbasis. Thus the 
$\mathbb{Z}_2$ symmetry enforces a two-fold splitting of the Hilbert space into equally-sized invariant subspaces
and no further splitting is possible as the commutant $\Cfree$ is two dimensional for the free ansatz.
This symmetry analysis is true for any connected graph and, more generally, we completely determine 
all possible free-mixer \DLAs for any graph:

\begin{table}
\caption{\textbf{Free ansatz.} 
Six families of connected graphs with $n\geq 2$ vertices 
are identified by their generated \DLAs.
\label{tab:freemixer}}
\includegraphics{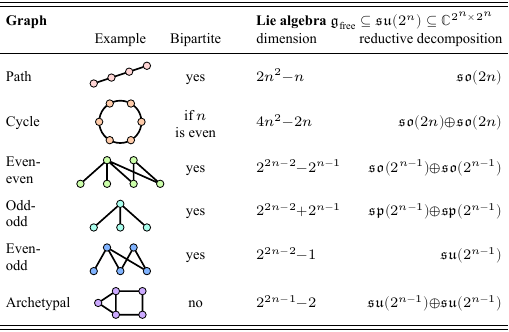}
\end{table}

\begin{theorem}[Free-mixer \DLAs]\label{thm:free-mixer-DLA-decomposition}
Given the generators $\GC_{\rm{free}}$ of the free \textnormal{QAOA} ansatz for any connected graph,
the generated \DLAs $\gfree$
fall into one of the six families depicted in Table~\ref{tab:freemixer}.
\end{theorem}

Theorem~\ref{thm:free-mixer-DLA-decomposition} identifies only six families of connected graphs
with different \DLAs and free-mixer generators: 
path graphs; cycle graphs; bipartite graphs of even-even, odd-odd, and even-odd
type (not equal to a path or a cycle graph); and finally any other graph: 
\begin{definition}[Archetypal graph]\label{def_other_graphs} 
An \emph{archetypal graph} is
connected but neither bipartite nor a cycle graph.
\end{definition}
The classification of Theorem~\ref{thm:free-mixer-DLA-decomposition}
includes one generic class given by the archetypal graphs from the last row
of Table~\ref{tab:freemixer}. 
Thus the \DLA for archetypal graphs is given by its particularly
simple and semiuniversal reductive
decomposition $\su(2^{n-1})\oplus \su(2^{n-1})$, while a more diverse set of \DLAs is obtained for bipartite and cycle graphs.
But a path or cycle graph has
trivial solutions to the maxcut problem,
while a polynomial algorithm exist for bipartite and planar graphs
 \cite{groetschel1981}.
Archetypal graphs cover all hard instances for the
maxcut problem as well as easy instances such as planar graphs
(as for the house graph in the last row of Table~\ref{tab:freemixer}).
Thus only archetypal graphs (or suitable subsets thereof) matter in practice.

The statement of Theorem~\ref{thm:free-mixer-DLA-decomposition}
is striking in light of the large number of non-isomorphic connected graphs:
one obtains twenty-two graphs for five vertices, but already more than eleven million for ten vertices.
Their number grows at least exponentially with the number of vertices.
But one can intuitively understand that the possible classes of \DLAs collapse 
to the few cases in Theorem~\ref{thm:free-mixer-DLA-decomposition}
as the generators from the free ansatz break almost all symmetries arising from a specific
graph.

Inspecting the families of \DLAs in Table~\ref{tab:freemixer}, the cases
of $\so(2^{n-1})\oplus \so(2^{n-1})$, $\usp(2^{n-1})\oplus \usp(2^{n-1})$, and
$\su(2^{n-1})\oplus \su(2^{n-1})$ nicely respect the block structure
induced by the eigenspaces of $X^{\otimes n}$. The case of $\su(2^{n-1})$ for the even-odd bipartite
case is more intriguing, but $\su(2^{n-1})$ is represented in a suitable block-diagonal basis
as $M {\oplus} ({-}M^{t})$ for complex $2^{n-1}\times 2^{n-1}$ matrices $M$
(see Appendix~\ref{appendix:explicit:reps:not:path:cycle}).
For path graphs, $\so(2n)$ is represented using a direct sum 
$\eta_{+} {\oplus} \eta_{-}$
of spinor representations $\eta_{\pm}$ of degree $2^{n-1}$
as detailed in Appendix~\ref{appendix:explicit:reps:path:cycle}.
The spinor representation $\eta_{+}$
is also utilized to map the simple parts of $\so(2n){\oplus}\so(2n)$
to their respective blocks.

\begin{table}
    \caption{\textbf{Pauli-string bases associated to Table~\ref{tab:freemixer}.}
    For Pauli strings consisting of $\mathrm{A} \in \{\mathrm{X}, \mathrm{Y}, \mathrm{Z}, \mathrm{I}\}$,
    let $\#\mathrm{A}$ denote the number of $\mathrm{A}$ in a Pauli string.
    Also, $\#\mathrm{A}|_{V_1}$ indicates the number of $\mathrm{A}$  in $V_1$ for a vertex bipartition $V=V_1 {\uplus} V_2$
    with $\abs{V}=n\geq 2$.
    Refer also to Table~\ref{tab:free-mixer-full-table} and Appendix~\ref{appendix:free-mixer}.
    \label{tab:free-basis-table}}
\includegraphics{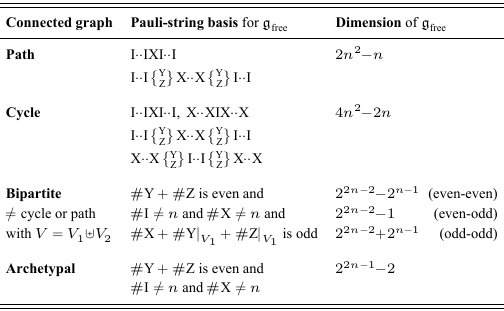}
\end{table}

Table~\ref{tab:free-basis-table} highlights explicit Pauli-string bases for the associated free-mixer \DLA
of each family of connected graphs from Table~\ref{tab:freemixer}, while 
Appendix~\ref{appendix:free-mixer} and Table~\ref{tab:free-mixer-full-table} provide more details.
For archetypal graphs, all Pauli strings appear that have an even parity  $\#\mathrm{Y} + \#\mathrm{Z}$ 
and differ from $\mathrm{I}\cddot\mathrm{I}$ and $\mathrm{X}\cddot\mathrm{X}$
(see Table~\ref{tab:free-basis-table}).
The extensive proofs leading to the results in this section are detailed in Appendices~\ref{appendix:free-mixer}
and \ref{appendix:free-mixer:bases}. Appendix~\ref{appendix:free-mixer}
employs \DLAc proof techniques, while Appendix~\ref{appendix:free-mixer:bases}
emphasizes graph properties.

\section{Implications for the free ansatz\label{sec:free:implications}}

We collect relevant consequences of Thm.~\ref{thm:free-mixer-DLA-decomposition}.
First, we clarify that the free ansatz of QAOA can---in principle---reach
at least one maximum cut (i.e., an optimal solution for maxcut)
using a finite number of layers for archetypal graphs from Def.~\ref{def_other_graphs}
(see Cor.~\ref{finite-depth-reach} below),
thus constituting a strict improvement over previous known convergence results,
namely reachability in the infinite-layer limit~\cite{herrman2022multi}.
The argument relies on three ingredients: (1)~the initial state and the target maxcut state
belong to the same invariant subspace, (2)~$\exp(\gfree)$ acts transitively on this invariant subspace,
and (3)~each unitary $U\in \exp(\gfree)$ can be generated in a finite number of layers.

For the second ingredient, the free-mixer \DLA 
$\gfree=\su(2^{n-1}){\oplus}\su(2^{n-1})$
for archetypal graphs (see Theorem~\ref{thm:free-mixer-DLA-decomposition})
acts transitively
on both invariant subspaces $\HC_{+}$
and $\HC_{-}$ of $X\tn$ (see App.~\ref{app:free:preliminaries})
which constitute the $+1$ and $-1$ eigenspaces of $X^{\otimes n}$ and
are spanned by all Hadamard basis states 
$\ket{b_1}\cdots\ket{b_n}$ with $b_j \in \{+, -\}$
for respectively an even or an odd number of minus
signs (i.e.\ $b_j = -$).
For the first ingredient, $\HC_{+}$
contains both the initial state $\ket{+}^{\otimes n}$ and at least one ground state of
$H_p$ as we detail now:
Consider a maximum cut $x\in\{0,1\}^n$ (see Sec.~\ref{sect:maxcut:qaoa}) and let  $\neg x$ denote its 
bitwise negation, which is also a maximum cut. Thus the state $\ket{x_+}=(\ket{x}{+}\ket{\neg x})/\sqrt{2}$
is a ground state of $H_p$ and also belongs to $\HC_{+}$
as $X^{\otimes n}(\ket{x}{+}\ket{\neg x})/\sqrt{2}=(\ket{x}{+}\ket{\neg x})/\sqrt{2}$.
Finally, our third ingredient relies on general arguments which
establish that reachable unitaries can be obtained
in a finite number of products \cite{lowenthal1971,silvaleite1991,dalessandro2002}:
\begin{lemma}\label{lem:finite:layers}
We consider the standard ansatz with the generators
$\{H_p, H_m\}$ and the free ansatz with the generators
$\mathcal{G}_{\free} = \mathcal{G}_p^{\free} \cup \mathcal{G}_m^{\free}$.
In both cases, if a target quantum state is reachable from a given initial state, then
it can be also obtained in a finite number of layers.
\end{lemma}

\begin{proof}
Based on the assumed reachability,
the target state can be obtained for the standard ansatz
using products of the form
$\prod_{\ell =1}^{L} e^{-i \theta_{\ell m} H_m} e^{-i \theta_{\ell p} H_p}$
or
$\prod_{\ell =1}^{L} e^{-i \theta_{\ell p} H_p} e^{-i \theta_{\ell m} H_m}$
with angles $\theta_{\ell m}$, $\theta_{\ell p}$ which can be zero.
General results in control theory on the so-called order of generation
\cite{lowenthal1971,silvaleite1991,dalessandro2002} imply 
a finite upper bound on the length of products
required to reach all possible unitaries in $\exp(\gfree)$.
Clearly, we can restrict ourselves to the first form of the product by enlarging $L$ by one.
This establishes the result for the standard ansatz.

For the free ansatz, we need to consider arbitrary products of
the form $\prod_{j =1}^{r} e^{-i \theta_j H_j}$
with angles $\theta_j$ and generators $H_j \in \mathcal{G}_{\free}$, where
the length $r$ of the products is again bounded by a finite integer $\tilde{L}$.
Clearly, we can rewrite all these products in the form of Eq.~\eqref{eq:circuit}
with $L=\tilde{L}+1$, suitable angles  $\theta_{m\ell k}$ and $\theta_{p\ell k}$ 
which are possibly zero,
as well as
$H_{mk} \in \mathcal{G}_m= \mathcal{G}_m^{\free}$ and $H_{pk} \in \mathcal{G}_p = \mathcal{G}_p^{\free}$.
This completes the proof.
\end{proof}

The length of the products required to 
obtain all reachable unitaries 
has to be greater or equal to the dimension of the generated \DLA
\cite{lowenthal1971,silvaleite1991,larocca2021theory,kiani2020learning},
but we lack general upper bounds as they 
depend on the particular generators.
These arguments neglect that the target ground state could possibly be reached
with a much smaller depth. 
By combining our three ingredients for the free-mixer ansatz,
we can at least rule out an infinite depth as compared to the usual convergence argument for QAOA
\cite{herrman2022multi,farhi2014quantum}, while not contradicting results on reachability
deficits for QAOA~\cite{akshay2020reachability} which apply to a fixed number of layers.
Initial results beyond the free ansatz are discussed at the end of Sec.~\ref{sec:lie:hierarchy}.

\begin{corollary}\label{finite-depth-reach}
The free QAOA ansatz can reach a maxcut state with a finite number of layers
for any archetypal graph from Def.~\ref{def_other_graphs}
\end{corollary}

In this context, one might wonder how effective QAOA will be for the free ansatz.
Unfortunately, the following results argue that barren plateaus dominate such optimization landscapes.
Recall that the QAOA cost function
$C(\thv)=\langle \psi(\thv)|H_p|\psi(\thv)\rangle$ 
from Eq.~\eqref{eq:cost} utilizes the unitaries 
$U(\thv)$ in the circuit from Eq.~\eqref{eq:circuit:general}
and $U(\thv)$
depends on 
the parameters
$\thv$ with real entries $\theta_{\vth}$.
For simplicity, the entries $\theta_{\vth}$
of $\thv$ are now indexed by
the numbers $\vth$.
Let $\partial_{\vth} C(\thv)$ denote the partial derivative of $C(\thv)$ with respect to the $\vth$-th
parameter $\theta_{\vth}$ in $\thv$.
Sampling the parameters $\thv$ 
according to a given distribution $d\thv$ over a chosen parameter domain $\delta_{L}$
induces a distribution on the associated unitaries $U(\thv)$ of an $L$-layered circuit
\cite{holmes2021connecting}.
We define the second-order moment operators
\begin{align*}
M_{e^{\g}}&:=\int_{U \in e^{\g}} d\mu_{e^{\g}}(U)\; U^{\otimes 2 }{\otimes} \bar{U}^{\otimes 2}\; \text{ and}\\
M_{L}&:=\int_{\thv\in \delta_L}d\thv\;  [U(\thv)]^{\otimes 2 }{\otimes} [\bar{U}(\thv)]^{\otimes 2 }.
\end{align*}
Then,  $\mathcal{A}_L:=M_{L}-M_{e^{\g}}$ is a positive semi-definite operator that quantifies how much the second moments
arising from $U(\thv)$ differ from those of the Haar measure over $e^{\g}$.
We say that a distribution $d\thv$ results in an $\varepsilon$-\emph{approximate unitary $2$-design}
for the unitaries $U(\thv) \subseteq e^{\g}$ if
\begin{equation*}
\norm{\mathcal{A}}_{\infty} =
\norm{M_L - M_{e^{\g}}}_{\infty} \leq \varepsilon.
\end{equation*}
Here, $\norm{\cdot}_{\infty}$ denotes the Schatten $\infty$-norm (or operator norm) which is given by
the largest singular value of its argument. We refer to \cite{harrow2009random,low2010pseudo,brown2010random,brandao2016local,
brandao2021models,haferkamp2021improved,oszmaniec2022epsilon,haferkamp2022random,
harrow2023approximate,schuster2024random,laracuente2024approximate,deneris2024exact} (and references therein)
for many similar definitions and notions related to approximate unitary $2$-designs.
We now obtain:

\begin{corollary}\label{cor:variance}
For the free ansatz, consider
any archetypal graph from
Def.~\ref{def_other_graphs} with $n> 3$ vertices and $\abs{E}$ edges.
Recall the QAOA cost function 
$C(\thv)$
from Eq.~\eqref{eq:cost}
and its partial derivative
$\partial_{\vth} C(\thv)$ with respect to the $\vth$-th
parameter $\theta_{\vth}$ in $\thv$. Assume that the multi-angle QAOA circuit
has enough layers such that the distribution of unitaries is an $\varepsilon$-approximate unitary $2$-design.
Then, the expectation value
of the partial derivatives is $E_{\thv}[\partial_{\vth} C(\thv)]=0$ and their variance is given by (with $d=2^n$)
\begin{equation*}
    \Var_{\thv}[\partial_{\vth} C(\thv)] = {4d^2|E|}/[(d^2{-}4)(d{+}2)] \leq 4n^2/2^n.
\end{equation*}
\end{corollary}

The proof of Corollary~\ref{cor:variance} is given in Appendix~\ref{proof:variance}, it
shows that, if the circuit is deep enough so that its unitaries form an approximate unitary $2$-design over $e^{\g}$,
then the partial derivatives of the cost function will, in average, vanish exponentially
with the system size. Thus the cost function will exhibit barren plateaus. This result can be further
understood by recalling that  the \DLA is a measure of expressiveness for the parametrized unitaries~\cite{larocca2021diagnosing}
and that highly expressive ansätze are known to exhibit barren plateaus~\cite{holmes2021connecting}. Indeed, for any archetypal
graph from Def.~\ref{def_other_graphs}, the \DLA is exponentially large, which implies exponentially small gradients~\cite{ragone2023unified}.

As previously mentioned, Corollary~\ref{cor:variance} holds for a multi-angle QAOA circuit when 
the generated unitaries form an $\varepsilon$-approximate unitary 2-design. This is fulfilled when
the number of layers  $L$ is bounded from below as~\cite{ragone2023unified}
\begin{equation}\label{eq:bound:L}
    L \geq \frac{\log(1/\varepsilon)}{\log({1/\norm{\mathcal{A}_L}_{\infty}})}.
\end{equation} 
Clearly, the exact value of $L$ will depend on the specific graph considered, as well as on the parameter distribution
and the chosen parameter domain. Still, empirical evidence suggests in meaningful scenarios (as for Pauli generators and parameters $\theta_{\vth}$ that are
independently and uniformly sampled from
$[-2\pi,2\pi]$) that the circuit unitaries become 
an $\varepsilon$-approximate unitary $2$-design 
if
$L\in\Omega(\poly(n))$~\cite{larocca2021diagnosing}. However, a general analytical characterization of 
$\norm{\mathcal{A}_L}_{\infty}$ is still lacking and beyond the scope of this work. We refer the reader
to~\cite{haferkamp2021improved,ODonnell2023,haah2024approximate,belkin2023approximate} for inspiration on how to tackle these and related questions.

\begin{figure}[t]
\centering
\includegraphics[width=0.95\columnwidth]{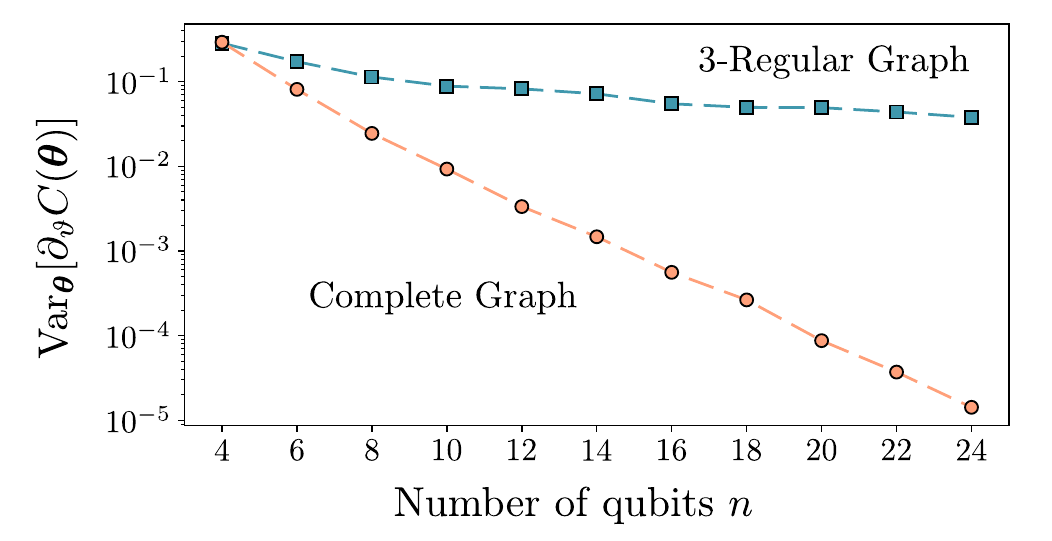}
\caption{\textbf{Gradient concentration for a single layer of the multi-angle QAOA applied to certain regular graphs.}
For each number of vertices $\abs{V}=n$, we pick a random $3$-regular or a complete graph.
The variance of the gradients is obtained by sampling $100$ parameter values. 
The cost function $C(\thv)$ has been renormalized so that its values lie
in $[-1,1]$. Refer also to the discussion in Sec.~\ref{sec:discussion}.
\label{fig:scaling}}
\end{figure}

While our results hint at the fact that the multi-angle QAOA will not be trainable for deep circuits, they---in principle---do not preclude the
possibility of good approximation ratios being reachable for very shallow circuits with, e.g., $L\in\OC(\log(n))$. Thus, one may wonder if shallow multi-angle QAOA are trainable, i.e., if they will be barren plateau-free. However, we
can find graphs for which the free ansatz exhibits  exponentially vanishing gradients already for a single layer. Usually, one does not
expect single layered or shallow circuits to have trainability issues when they are either composed of few local entangling gates
or have a small number of parameters~\cite{cerezo2020cost,pesah2020absence}. However, a single layer of the
free ansatz can simultaneously be highly entangling and contain up to a polynomial number of parameters in $n$.
For instance, consider the case of an $D$-regular graph. For constant $D$,
each layer of the circuit will not be highly entangling, and thus one can expect sufficiently large gradients. However, if $D$ scales
with $n$, then a single layer can be sufficiently expressive to observe barren plateaus. This intuition is exemplified in Fig.~\ref{fig:scaling}
where we compare the gradient scaling for random $3$-regular and complete graphs. Therein
the gradients for $3$-regular graphs are quite large, while $\Var_{\thv}[\partial_{\vth} C(\thv)]$ decays exponentially with $n$ for
complete graphs as highlighted by the straight line in the log-linear plot of Fig.~\ref{fig:scaling}.
A similar behavior of enlargement of barren plateaus with increasing $D$
even for a single layer was also observed in \cite[Fig.~9 of App.~C]{Vijendran2024} for the standard ansatz.

In summary, freeing the QAOA generators and assigning a single parameter to each gate can certainly help to reduce the circuit depth, but such an uncontrolled
increase in expressiveness may be undesirable. As illustrated, for example, by Corollary~\ref{cor:variance} and Fig.~\ref{fig:scaling},
this might lead to severe trainability issues. Although more expressible ansätze are generally expected to allow for reduced circuit
depths (as in the case of hardware efficient ansätze~\cite{kandala2017hardware}), oftentimes such additional
expressiveness needs to be well-directed. Following the ideas in geometric deep learning~\cite{larocca2022group}, it is usually possible to
restrict the expressiveness without sacrificing performance by building ansätze that actively exploit the symmetries in the problem.
Unfortunately, the $\rm{free}$ ansatz breaks one of the essential symmetries of the maxcut QAOA given by
the automorphism group of the considered graph.

\section{Symmetries of the standard ansatz\label{SEC:STD}}

In Sec.~\ref{SEC:FREE-ANSATZ}, we have completely characterized the possible \DLAs and their symmetries
for the free ansatz for any possible graph. We would ideally aim at a similar general result for the standard ansatz.
We now explore the standard ansatz and show how even characterizing the corresponding linear symmetries 
as given by the commutant becomes a much more intricate task.

Recall that the standard generators $\mathcal{G}_{\std} = 
\{H_p, H_m\}$ 
from Eq.~\eqref{eq:std} generate the \DLA
$\gstd := \lie{i\mathcal{G}_{\std}}$.
Clearly, the free-mixer \DLA $\gfree \supseteq \gstd$ contains the \DLA $\gstd$
for the standard ansatz. Conversely, the corresponding symmetries observe the inclusion
$\spn_{\C}\{I\tn,X^{\otimes n}\}=\Cfree \subseteq \Cstd$ of commutants.
Here, $X^{\otimes n}$ induces a $\ZTWO$ symmetry given by the bit-wise negation 
$X\tn\ket{x}=\ket{\neg x}$.
The matrix group generated by the matrices
$I\tn$ and $X\tn$ is given by
\begin{equation}\label{eq:Z2}
\ZTWO := \group{I\tn,X^{\otimes n}}.
\end{equation}
Then, $\ket{x}$ and $X\tn\ket{x}=\ket{\neg x}$
have the same expectation value with respect to $H_p$ and these states clearly correspond
to the same cut value. 
This is visualized in Fig.~\ref{fig:house_sym:new}
for the maximum cut vectors
[see Fig.~\ref{fig:decomposition:house}(e)]
for the example of
the house graph.

\begin{figure}[t]
\includegraphics{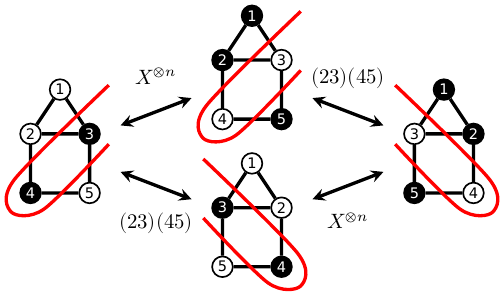}
\caption{\textbf{Symmetries in the standard-ansatz QAOA.} The free-mixer symmetries $I\tn$ and $X^{\otimes n}$
are complemented by symmetries arising from graph automorphism
such as the permutation $(23)(45)$ for the house graph. These symmetries
naturally act on quantum states (see text).
\label{fig:house_sym:new}}
\end{figure}

One additional class of symmetries arises from the \emph{automorphism group} $\Aut \subseteq \mathcal{S}_n$  \cite{diestel2017}
of the underlying graph $G$. Here, $\Aut$ is
the subgroup of vertex permutations from the symmetric group $\mathcal{S}_n$ \cite{sagan2001}
that map the edge set $E$ of $G$ to itself, i.e.,
\begin{equation*}
    \!\!\!\Aut = \{ \sg \in \mathcal{S}_n \,|\,  \{\sigma(k),\sigma(\ell)\}\in E \text{ iff } \{k,\ell\}\in E\}.
\end{equation*}
The action $\sigma\cdot \ket{x}:=\ket{x_{\sigma^{-1}(1)},\ldots,x_{\sigma^{-1}(n)}}$ 
is naturally extended 
to any $\sigma \in \mathcal{S}_n$
via the representation map $\zeta(\sigma)\in \C^{d\times d}$ (with $d=2^n$). For example, 
the transposition $(1,2)\in \mathcal{S}_2$ operating on two qubits is given by
\begin{equation*}
\zeta[(1,2)] = \SWAP = 
\left(
\begin{smallmatrix}
1 & 0 & 0 & 0\\
0 & 0 & 1 & 0\\
0 & 1 & 0 & 0\\
0 & 0 & 0 & 1
\end{smallmatrix}
\right).
\end{equation*}
The general form of $\zeta$ is uniquely defined for any number $n$ of qubits
as any element of the symmetric group $\mathcal{S}_n$ can be written as a product of transpositions
$(j,j{+}1)$ with $j\in \{1,\ldots,n{-}1\}$ \cite{sagan2001}. Henceforth, $\symone$ denotes the identity permutation.

One verifies that  $\ket{x}$ and $\sigma\cdot \ket{x}$
have the same cut value for any $\sigma\in\Aut$
as one can check that 
$(\langle x|\zeta(\sigma)^\dagger)H_p(\zeta(\sigma)|x\rangle)=\langle x|H_p|x\rangle$.
Figure~\ref{fig:house_sym:new} shows the 
house graph, where the maximum cut vectors
from Fig.~\ref{fig:decomposition:house}(e)
are permuted by the automorphism $(2,3)(4,5)$.

As $X\tn$ commutes
with $\zeta(\sigma)$ for any permutation $\sigma \in \mathcal{S}_n$,
the \emph{group of natural symmetries} of the standard ansatz of QAOA 
is defined as the direct product 
\begin{equation}\label{eq:nat:group}
\Gnat := \ZTWO\times \zeta[\Aut]
\end{equation}
of the group $\ZTWO$ as specified in Eq.~\eqref{eq:Z2}
and the subgroup given by the image of the automorphism group $\Aut$ under the representation $\zeta$.
We emphasize 
their linear structure and define the \emph{natural symmetries} as
\begin{equation}\label{eq:nat:sym}
\Snat := \spnC\{ \Gnat \} \subseteq \Cstd.
\end{equation}
We can always choose a linear-independent basis of the natural symmetries $\Snat$ as 
\begin{equation}\label{eq:nat:bas}
\Bnat = \{ b_1, \ldots, b_{\dim(\Snat)}\},
\end{equation}
while the dimension $\dim(\Snat)\leq 2 \abs{\Aut}$
is in general not equal to $2 \abs{\Aut}$. One example with $\dim(\Snat)< 2 \abs{\Aut}$
is presented in Fig.~\ref{fig:decomposition:four}, and $\Aut = \mathcal{S}_n$ with $n>3$
provides another one as $\Gnat$ then contains too many elements
for all of them to be linear independent. 
Here,
$\abs{\Aut}$ denotes the number of elements in the automorphism group $\Aut$.
For the house graph, $\Snat$ has a basis given by 
\begin{equation}\label{eq:nat:house}
\Bnat = \{  I\tn, X\tn, \zeta[(2,3)(4,5)], X\tn \zeta[(2,3)(4,5)] \},
\end{equation}
where the corresponding action of the group of natural symmetries is visualized in 
Fig.~\ref{fig:house_sym:new}.

At this point, it is worth asking whether the natural symmetries are the only symmetries
in the commutant, i.e., whether there is an equality in Eq.~\eqref{eq:nat:sym}.
But even for the house graph, the commutant $\Cstd$ has dimension six 
(see Table~\ref{table:house:graph})
while $\Snat$ has dimension four.
Indeed, the one-dimensional projectors
$P_3$ and $P_6$ complement the natural symmetries
as detailed in Fig.~\ref{fig:decomposition:house}(b) and (c).
We are using projectors $P_j \in \C^{d\times d}$ (with $d=2^n$) that observe $P_j^2  = P_j = P_j^{\dagger}$
to specify invariant subspaces $\mathcal{H}_j$ as the image of the corresponding projector $P_j$.

\begin{figure}[t]
\includegraphics{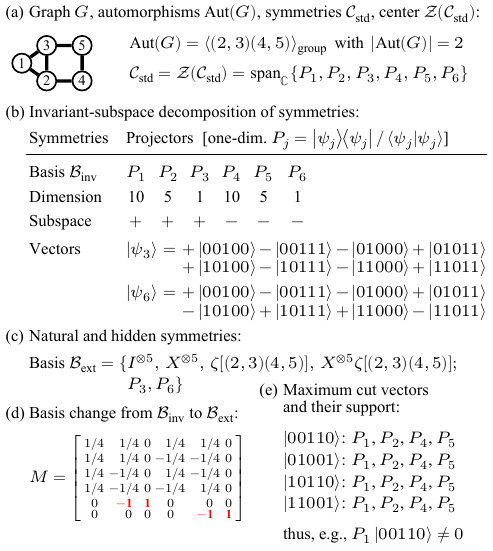}
\caption{\textbf{Symmetry decomposition of the standard-ansatz QAOA
for the house graph.} (a)~graph, automorphisms, symmetries;
(b)~invariant subspaces $\mathcal{H}_j$ 
as in Table~\ref{table:house:graph}
and projectors $P_j$ which are explicitly shown in 
Fig.~\ref{fig:matrices:house} in App.~\ref{app:projections}.
This
includes their dimension, their respective position in the $+1$ and the $-1$
eigenspace of $X^{\otimes 5}$
(denoted by $+$ or $-$), and 
explicit one-dimensional projectors.
(c)~natural and hidden symmetries;
(d)~transformation from $\Binv$ in (b) to $\Bext$ in (c),
$\Bext$ is
not unique as shown by the red, nonzero entries in the last two rows;
(e)~maximum cut vectors and their (nonzero) support in the invariant subspaces. 
\label{fig:decomposition:house}}
\end{figure}

Thus not all symmetries in the commutant $\Cstd$ are in general contained in $\Snat$.
We can extend any basis $\Bnat$ from Eq.~\eqref{eq:nat:bas} of the natural symmetries $\Snat$
to a basis 
\begin{equation*}
\Bext = \{b_1,\ldots, b_{q};\; \tilde{b}_1, \ldots, \tilde{b}_{\tilde{q}}\}
\end{equation*}
of all symmetries $\Cstd$ where 
$q=\dim(\Snat)$ and $\tilde{q} = \dim(\Cstd) - \dim(\Snat)$.
Hence, any symmetry $S \in \Cstd$ can be linearly expanded
as
\begin{equation}\label{eq:sym:expansion}
S = \Big(\sum_{j=1}^{q} c_j b_j \Big) + \Big(\sum_{k=1}^{\tilde{q}} \tilde{c}_k \tilde{b}_k \Big),
\end{equation}
where $c_j, \tilde{c}_k\in \C$.
We say a symmetry $S\in \Cstd$ is a \emph{hidden symmetry}
if $\tilde{q}\neq 0$ and there exists a $k \in \{1,\ldots,\tilde{q}\}$  in Eq.~\eqref{eq:sym:expansion} with 
$\tilde{c}_k \neq 0$. This definition is independent of the choice of
$\Bnat$ or its extension $\Bext$ to a full basis of $\Cstd$.
Formally, the hidden symmetries can also be identified with
all nonzero elements of the quotient vector space
of the commutant $\Cstd$ with respect to the natural symmetries $\Snat$.

The structure of the hidden symmetries and their relation to the invariant subspaces
(see Sec.~\ref{sec:dla}) are highlighted in Fig.~\ref{fig:decomposition:house} which details
the symmetries of standard-ansatz QAOA for the house graph. Following Table~\ref{table:house:graph},
Figure~\ref{fig:decomposition:house}(b) describes the commutant $\Cstd$ 
as the complex span of six (orthogonal)
projectors $P_j$, which define the basis $\Binv$ of $\Cstd$ and uniquely identify
the respective irreducible invariant subspaces. (In general, these invariant subspaces are not
irreducible but only isotypical as detailed in App.~\ref{appendix:symmtry:analysis}.)
For reference, Fig.~\ref{fig:matrices:house} in App.~\ref{app:projections}
provides the explicit matrix form of the projectors in Figure~\ref{fig:decomposition:house}. 
The one-dimensional projectors $P_3$ and $P_6$ are specified via their respective invariant, one-dimensional
subspaces that are respectively spanned by the vectors $\ket{\psi_3}$ and $\ket{\psi_6}$. Part (c) of Figure~\ref{fig:decomposition:house}
clarifies that the natural symmetries $\Bnat$ in Eq.~\eqref{eq:nat:house} can be extended
to a full basis $\Bext$ of the commutant $\Cstd$
by adding the one-dimensional projectors $P_3$ and $P_6$. But this choice is not unique as all projectors $P_j$ with $j \in \{2,3,5,6\}$
are hidden symmetries. This is directly implied by the red, nonzero entries in the last two rows in
Fig.~\ref{fig:decomposition:house}(d), which details
the basis change from the basis $\Binv$ of $\Cstd$ to $\Bext$.
Consequently, the description via natural and hidden symmetries 
provides an additional dimension complementing the decomposition into
invariant subspaces. We refer to Appendix~\ref{app:natural:hidden} and particularly
Fig.~\ref{fig:decomposition:four} for a more intricate example 
that is---in contrast to the house graph---not multiplicity-free
[in the sense that $m_{\lambda} \neq 1$ in Eq.~\eqref{invariants_basis}]. Both Figures~\ref{fig:decomposition:house} and
\ref{fig:decomposition:four} demonstrate the highly complex structure of the symmetries in the standard ansatz.

It is instructive to reinterpret a
vector $\ket{\psi}$ spanning a
one-dimensional invariant subspace
(as in Fig.~\ref{fig:decomposition:house})
as a simultaneous eigenvector of a set of generators $\mathcal{G}$, i.e.,
\begin{equation}\label{eq:sim:eigen}
H \ket{\psi} = \beta(H, \ket{\psi})\, \ket{\psi}\; \text{ for all }\; H \in  \mathcal{G}.
\end{equation}
The two formulations can be easily shown to be equivalent. The eigenvalues $\beta(H,  \ket{\psi})$
for the respective generators $H \in  \mathcal{G}$ and the simultaneous eigenvector $\ket{\psi}$ uniquely identify 
the corresponding one-dimensional representation up to multiplicity (see App.~\ref{app:natural:hidden}). The multiplicity
can be different from one.
Instead of relying on a complete irreducible decomposition, other methods can be---in principle---used
to identify these simultaneous eigenvectors more directly \cite{shemesh1984,pastuszak2015,malvetti2023,rozon2024}.
The condition in Eq.~\eqref{eq:sim:eigen} is a necessary condition for the existence of
quantum many-body scars~\cite{moudgalya2020,moudgalya2022hilbert, moudgalya2022exhaustive}, where $\mathcal{G}$
contains local terms (or sums of local terms) $H_j$ for a parametrized class of Hamiltonians
$H = \sum_j r_j H_j$ with parameters $r_j \in \R$. These classes of Hamiltonians are studied in the search and analysis of counterexamples
to the eigenstate thermalization hypothesis
\cite{rozon2024,turner2018,serbyn2021,pakrouski2020,pakrouski2021,sala2020,bull2020,moudgalya2020numerical}.

Combining locality in generators with the preservation of symmetries 
(such as natural symmetries in standard ansatz)
can give rise to additional symmetries \cite{marvian2022restrictions,marvian2023non,kazi2023universality},
while the absence of symmetries \cite{zeier2011symmetry,zimboras2015symmetry}
allows for universality as in the case of all one- and two-qubit gates \cite{divincenzo1995two,lloyd1995almost}.
In particular, the collection of 
$k$-local unitaries (that act on at most $k$ qubits with $k<n$) cannot in general generate all $n$-local unitaries
while respecting the same symmetries for both $k$- and $n$-local unitaries.  However, symmetries for the standard ansatz
also arise as possible generators
are restricted not only based on locality but also based on graph connectivity.
We discuss these points in more detail in Sec.~\ref{sec:discussion}.

\section{\DLAm hierarchy: from the free to the standard ansatz\label{sec:lie:hierarchy}}

We continue to study the symmetries of the standard ansatz.
In this section, we explore various \DLAs that 
lie between $\gfree$ and the standard-mixer \DLA $\gstd$
or, alternatively, between $\uu(2^n)$ and $\gstd$.
The \DLA $\gstd$ is contained in these
intermediate \DLAs which provide an additional approach for its 
characterization.
In particular, we analyze the group of natural symmetries $\Gnat$ from
Sec.~\ref{SEC:STD} and how this group
acts on the state space $\mathcal{H}=\C^d$ and the intermediate \DLAs.
For a given graph $G$,
Eq.~\eqref{eq:nat:group} defines
$\Gnat$ as a direct product of $\ZTWO = \group{I\tn,X^{\otimes n}}$
and the automorphism group $\Aut$ acting
as qubit permutations $\zeta(\sigma)\in \C^{d\times d}$ for $\sigma \in  \Aut$.

As a first step we introduce the orbit ansatz:
Given a graph $G$ with vertices $v\in V$ and edges $\{w_1,w_2\} \in E$,
we recall the vertex and the edge orbits (see Fig.~\ref{fig:ex:house:orbit})
\begin{align*}
\VV_v&:=\{\sigma(v) \;\text{ for }\; \sigma \in \Aut\} \;\text{ and }\\
\E_{\{w_1,w_2\}}&:= \{ \{\sigma(w_1),\sigma(w_2)\} \;\text{ for }\; \sigma \in \Aut\}
\end{align*}
and define
the sets of generators for the \emph{orbit ansatz} as
\begin{subequations}
\label{eq:orbit}
\begin{align}
\mathcal{G}_p^{\orb}& := \Big\{  \sum_{\{w,\tilde{w}\} \in \E_{\{v,\tilde{v}\}}} Z_{w}Z_{\tilde{w}} \;\text{ for }\; \{v,\tilde{v}\} \in E \Big\},
\label{eq:orbit:problem}
\\
\mathcal{G}_m^{\orb}& :=\Big\{  \sum_{\tilde{v} \in \VV_{v}} X_{\tilde{v}} \;\text{ for }\; v \in V \Big\}.
\label{eq:orbit:mixer}
\end{align}
\end{subequations}
We introduce the \DLA for the orbit ansatz as
\begin{equation*}
\gorb :=\lie{i\GC_{\orb}} \;\text{ where }\; \GC_{\orb} := \mathcal{G}_p^{\orb} \cup \mathcal{G}_m^{\orb}.
\end{equation*}

\begin{figure}[t]
\includegraphics{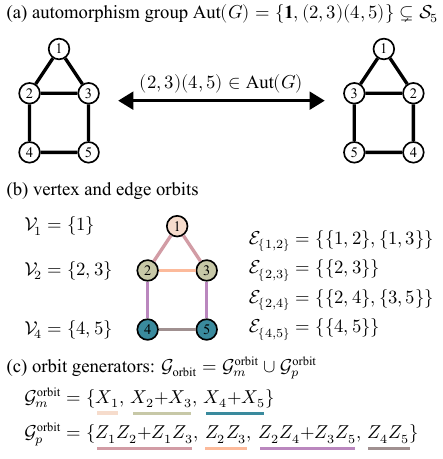}
\caption{\textbf{Orbit generators for the house graph.} (a)~automorphisms, (b)~vertex and edge orbits, (c)~generators. 
\label{fig:ex:house:orbit}}
\end{figure}

Interestingly, the generators of the orbit ansatz can also be obtained 
by suitably symmetrizing the free-mixer generators.
To this end, we first define three symmetrization operations on a given matrix $M \in \C^{d\times d}$:
\begin{subequations}
\begin{align}
\tau_{\ZTWO}(M) &:= \tfrac{1}{2} [I\tn\, M\, I\tn + X^{\otimes n}\, M\, X^{\otimes n} ],
\label{eq:sym:z}
\\
\tauaut(M) &:= \tfrac{1}{\abs{\Aut}}
\sum_{\sigma \in\Aut}\zeta[\sigma]\,M\,\zeta[\sigma^{-1}],
\label{eq:sym:aut}
\\
\taunat(M) &:= \tauaut[\tau_{\ZTWO}(M)] = \tau_{\ZTWO}[\tauaut(M)].
\label{eq:sym:nat}
\end{align}
\end{subequations}
Here, Eq.~\eqref{eq:sym:nat} follows as the actions of $\ZTWO$ and $\Aut$ commute
as detailed in Sec.~\ref{SEC:STD}.
Also, the free-mixer generators from Eq.~\eqref{eq:free} commute with $X^{\tn}$ 
[see Lemma~\ref{app:prop:center}(a)]. Thus the orbit-ansatz generators in Eq.~\eqref{eq:orbit:mixer}
are recovered up to suitable scalar factors
by symmetrizing the free-ansatz generators, i.e.,
for every $g \in \GC_{\orb}$, there exists
a positive integer $z$ and 
$\tilde{g} \in \mathcal{G}_{\text{free}}$ such that 
\begin{equation}\label{eq:sym:free:nat}
g = z\, \taunat(\tilde{g})=z\, \tauaut(\tilde{g}).
\end{equation}

\begin{figure}
\includegraphics{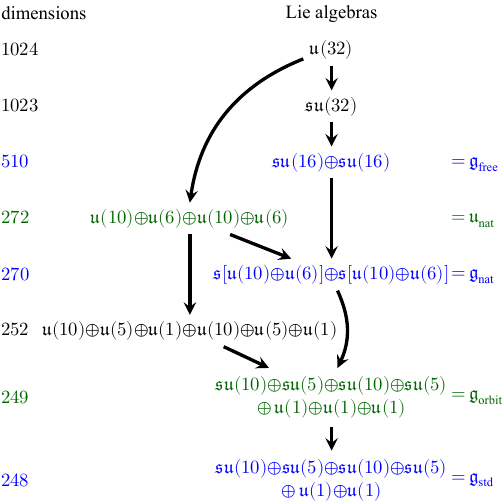}
\caption{\textbf{\DLAm inclusions for the house graph.} The \DLAs $\gfree$, $\unat$, $\gnat$, $\gorb$, and $\gstd$ are shown with their inclusion relations and dimensions. The decomposition $\mathfrak{s}[\uu(10){\oplus}\uu(6)]\iso \su(10){\oplus}\su(6){\oplus}\uu(1)$ has unitary blocks $\uu(10)$ and $\uu(6)$ where the $16$-dimensional block has zero trace (which details more than the reductive decomposition).
\label{fig:ex:house:hierarchy}}
\end{figure}

The symmetrization operation $\taunat$ from Eq.~\eqref{eq:sym:nat}
will be a key tool to better characterize the \DLA $\gstd$
associated to the standard ansatz. Let us define 
\begin{align}
\unat&:=\spn_{\mbb{R}} \{ \taunat(g) \,\text{ for }\, g\in \uu(d)\},
\intertext{as well as}
\gnat&:=\spn_{\mbb{R}} \{ \taunat(g) \,\text{ for }\, g\in \gfree\} \nonumber\\
&\phantom{:}= \spn_{\mbb{R}} \{ \tauaut(g) \,\text{ for }\, g\in \gfree\} \subseteq \unat.
\label{eq:gnat}
\end{align}
Both $\gnat$ and $\unat$ form \DLAs as they are closed under commutator due to (for $M_j  \in \C^{d\times d}$)
\begin{equation*}
[\taunat(M_1),\, \taunat(M_2)] = \taunat([M_1,\, M_2]).
\end{equation*}

Figure~\ref{fig:ex:house:hierarchy} highlights the corresponding \DLAm inclusion relations for the house graph.
The structure of the centers of the considered \DLAs is involved and we detail their structure for the house graph
in terms of the projectors $P_j$ from Figure~\ref{fig:decomposition:house}:
\begin{align*}
\cent(\unat) &= \spnR \{iP_1, i(P_2{+}P_3), iP_4, i(P_5{+}P_6) \},\\
\cent(\gnat) &= \spnR \{i({-}3 P_1{+}5P_2{+}5P_3),  \\
&\phantom{=\spnR\{\hspace{3pt} }
i({-}3P_4{+}5P_5{+}5P_6)\},\\
\cent(\gorb) &= \spnR \{   i({-}3 P_1{+}5P_2{+}5P_3{-}3P_4{+}5P_5{+}5P_6),\\
&\phantom{=\spnR\{\hspace{3pt} } i({-}1P_1{+}2P_2{-}1P_4{+}2P_5),\\
&\phantom{=\spnR\{\hspace{3pt} } i({+}1P_1{-}3P_2{+}5P_3{-}1P_4{+}3P_5{-}5P_6)\},\\
\cent(\gstd) &= \spnR \{  i({-}1P_1{+}2P_2{-}1P_4{+}2P_5),\\
&\phantom{=\spnR\{\hspace{3pt} } i({+}1P_1{-}3P_2{+}5P_3{-}1P_4{+}3P_5{-}5P_6)\}.
\end{align*}
On a more general note, Lemma~\ref{lemma:cent:cent} 
in Appendix~\ref{app:center}
highlights
how the center $\cent(\g)$ of a compact Lie algebra $\g$ (i.e.\ $\g \subseteq \uu(m)$ for a suitable $m$) is contained
in the center $\cent(\CC)$ of its commutant $\CC$, i.e.,
$\cent(\g) \subseteq \cent(\CC)$. But the projectors $P_j$ to (certain)
invariant subspaces 
(or more precisely to the isotypical subspaces as detailed in Appendix~\ref{app:symmetry:theory})
span the center $\cent(\CC)$ of the commutant. Thus the above decomposition
of the centers associated to the house graph is no longer that surprising.
We can readily derive a general strict hierarchy among the \DLAs and their commutants considered so far:
\begin{proposition}[\DLAm and commutant  hierarchy]\label{thm_dla_chain}
For any graph $G$, we observe $\mathcal{C}(\gnat) = \Snat$
and the following chain of inclusions:
\begin{equation*}
\begin{matrix}
&\gstd &\subseteq &\gorb &\subseteq &\gnat &\subseteq &\gfree,\\[1mm]
&\mathcal{C}(\gstd) &\supseteq &\mathcal{C}(\gorb)  &\supseteq &\mathcal{C}(\gnat)
&\supseteq
&\mathcal{C}(\gfree).
\end{matrix}
\end{equation*}
\end{proposition}
The power of the hierarchy in Prop.~\ref{thm_dla_chain} resides in the fact that $\gnat$ 
can be characterized efficiently and therefore provides an accessible upper bound to $\gorb$
and $\gstd$.

\begin{proof}[Proof of Proposition~\ref{thm_dla_chain}]
The inclusion relations on the commutants directly follow from the inclusion relations
on the \DLAs. The inclusion relation $\gstd \subseteq \gorb$ is a consequence of 
the form of the generators in  Eq.~\eqref{eq:orbit:mixer}, similar as for $\gstd \subseteq \gfree$.
We obtain $\gnat$ by symmetrizing the basis of $\gfree$ via Eq.~\eqref{eq:sym:nat} which implies $\gnat \subseteq \gfree$.
Moreover, $\gorb \subseteq \gnat$ as
$\gorb$ is generated by the symmetrized generators of $\gfree$ which are contained in $\gnat$.
Thus we are only left with verifying that the characterizations of $\gorb$
via Eq.~\eqref{eq:orbit} and Eq.~\eqref{eq:sym:free:nat} are equivalent, but this is detailed before
Eq.~\eqref{eq:sym:free:nat}.
\end{proof}

We continue by discussing the
structure of the symmetrized \DLAs $\gnat$ and $\unat$
for archetypal graphs following Def.~\ref{def_other_graphs}.
In this regard, we again point the reader to 
the example of the house graph in
Figure~\ref{fig:ex:house:hierarchy}
and the above discussed structure of the associated centers.
For archetypal graphs, the following
Proposition~\ref{theo:algebra-Z-Aut} shows that
$\gnat$ differs from $\unat$
by just a pair of \DLAm elements from the center $\cent(\unat)$ of $\unat$
which can be interpreted as $\gnat$ being 
semi-universal~\cite{marvian2023non,kazi2023universality}. 

\begin{proposition}\label{theo:algebra-Z-Aut}
For any archetypal graph from Def.~\ref{def_other_graphs},
\begin{equation}\label{eq:DLA-aut}
      \spnR\{\gnat \cup \{iP_{+},iP_{-}\}\}
      = \unat,
\end{equation}
where $P_{\pm}:= (I^{\otimes n} {\pm} X^{\otimes n})/2$
(see Sec.~\ref{app:free:preliminaries})
projects to the two invariant subspaces.
The reductive decomposition is
\begin{equation}
\gnat \oplus \uu(1)\oplus \uu(1)
\iso \unat
\end{equation}
and the elements $iP_{+}$ and $iP_{-}$ are chosen to respectively span each of the
two abelian subalgebras $\uu(1)$.
\end{proposition}

\begin{proof}
Recall that $\gfree\iso \su(2^{n-1})\oplus\su(2^{n-1})$ holds
for any archetypal graph (see Theorem~\ref{thm:free-mixer-DLA-decomposition}).
This implies
\begin{equation}\label{eq:prop:proof:a}
    \tau_{\ZTWO}[\uu(2^n)] \iso \uu(2^{n-1})\oplus\uu(2^{n-1}) \iso \gfree \oplus \ga_1 \oplus \ga_2,
\end{equation}
where $\ga_j$ are one-dimensional abelian \DLAs that are respectively spanned by
$iP_{+}$ or $iP_{-}$. We apply the symmetrization operation $\tauaut$ to Eq.~\eqref{eq:prop:proof:a} and obtain
\begin{align*}
\unat &= \tauaut[\tau_{\ZTWO}[\uu(2^n)]] \iso \tauaut[\gfree] \oplus \tauaut[\ga_1] \oplus \tauaut[\ga_2]\\
&\iso \tauaut[\gfree] \oplus \ga_1 \oplus \ga_2 \iso \gnat \oplus \ga_1 \oplus \ga_2,
\end{align*}
where the second isomorphism follows from the fact $I\tn$ and $X\tn$ commute with any
automorphism. The definition of $\gnat$ implies 
the third isomorphism.
\end{proof}

One can naturally extend the symmetrization from Equation~\eqref{eq:sym:nat}
to quantum states $\ket{\psi} \in \HC$ via
\begin{equation*}
\tauhat(\ket{\psi}) := 
\tfrac{1}{2\abs{\Aut}} (I\tn {+} X\tn) 
\sum_{\sigma \in\Aut}\zeta[\sigma]\,\ket{\psi}
\end{equation*}
and we introduce the related symmetrized subspace
\begin{equation*}
\HC_{\mathrm{nat}} := \{ \tauhat(\ket{\psi}) \text{ for all }\psi \in \HC \} \subseteq \HC_{+} \subseteq \HC.
\end{equation*}
For the house graph, the symmetrized subspace $\HC_{\mathrm{nat}}$
has a dimension of ten and
corresponds to the projector $P_1$ in Figure~\ref{fig:decomposition:house}.
Consider a maximum cut $x\in\{0,1\}^n$ (see Sec.~\ref{sect:maxcut:qaoa})
and the corresponding ground state $\ket{x}\in \HC$ of $H_p$ with ground state energy $E$.
The action of $(X\tn)^b \zeta(\sg)$ with $b\in \{0,1\}$ and $\sigma \in \zeta$ 
maps $\ket{x}$ to another ground state $(X\tn)^b \zeta(\sg)\ket{x}$
with
\begin{equation*}
H_p (X\tn)^b \zeta(\sg)\ket{x}
= E (X\tn)^b \zeta(\sg) \ket{x},
\end{equation*}
which follows as $H_p$ and $(X\tn)^b \zeta(\sg)$ commute.
Clearly,  $\tauhat(\ket{x}) \in \HC_{\mathrm{nat}}$ is another ground state of $H_p$.

\begin{proposition}\label{prop:H:nat}
Let $G$ denote a connected graph with at least two vertices.
(a) $\HC_{\rm nat}$ contains at least one ground state of $H_p$.
(b) If $G$ is any archetypal graph from Def.~\ref{def_other_graphs},
then $\exp(\gnat)$ acts transitively on $\HC_{\rm nat}$.
\end{proposition}

\begin{proof}
Statement (a) has been verified directly before the proposition
and (b) follows from Proposition~\ref{theo:algebra-Z-Aut}.
\end{proof}

The dimension of the space $\HC_{\rm nat}$ will be further characterized
via character computations in Sec.~\ref{sec:character:computations:formulas}.
Even though Prop.~\ref{prop:H:nat}(b) verifies the transitivity of $\exp(\gnat)$
on $\HC_{\rm nat}$ for archetypal graphs,
$\exp(\gstd)$ is in general \emph{not} transitive on $\HC_{\rm nat}$. This will be
exemplified in Sec.~\ref{sec:conjecture} for graphs $G$ with $\abs{\Aut}=1$.

\section{Partial results for
the standard ansatz and
particular graphs\label{sec:standard-ansatz-partial-results}}

In order to illustrate the results and methods of the previous sections,
we now analyze concrete cases given by path, cycle, and complete graphs.
We are able to completely determine the standard-ansatz \DLA for the path graphs and we provide upper bounds for the cycle and the complete graphs. The solvability of these cases may be
related to the well-known fact that the corresponding maxcut problems have trivial solutions.

\subsection{Standard ansatz for path graphs\label{sec:std:path}}

We determine the standard-ansatz \DLA $\gstd$ for path graphs $P_n$.
Applying Theorem~\ref{thm:free-mixer-DLA-decomposition}, $\gstd \subseteq \gnat \subseteq \gfree = \so(2n)$.
The \DLA $\gfree$ is spanned (up to factors of $i$) by the $2n^2{-}n$ Pauli strings (see Table~\ref{tab:free-basis-table})
\begin{equation*}
\mathrm{I}\cddot\mathrm{I}\mathrm{X}\mathrm{I}\cddot\mathrm{I}\;\text{ and }\;
\mathrm{I}\cddot\mathrm{I}\yz\mathrm{X}\cddot\mathrm{X}\yz\mathrm{I}\cddot\mathrm{I},
\end{equation*}
or equally, the basis consists of (with $1 \leq j < k \leq n$)
\begin{subequations}
\label{eq:gfree:path}
\begin{align}
iX_j,\; &iY_j X_{j+1} \cddot X_{k-1} Y_k,\; iY_j X_{j+1} \cddot X_{k-1} Z_k,\\
&iZ_j X_{j+1} \cddot X_{k-1} Y_k,\; iZ_j X_{j+1} \cddot X_{k-1} Z_k.
\end{align}
\end{subequations}
The automorphism group of the path graph is given by
\begin{align}
&\AUT(P_n)
= \{\symone,\, \tilde{\sigma}_n \} \subseteq \mathcal{S}_n \;\text{ where }\; \bar{n} := \floor{n/2}\;\text{ and} \nonumber\\
&\tilde{\sigma}_n := 
\begin{cases}
        (1\, n) (2\, n{-}1) \cdots (\bar{n}\, \bar{n}{+}1) & \text{for $n$ even}, \\
        (1\, n) (2\, n{-}1) \cdots (\bar{n}\, \bar{n}{+}2) & \text{for $n$ odd}.
    \end{cases} \label{eq:sigma:n}
\end{align}
Note that $\bar{n}=n/2$ for even $n$ and $\bar{n}=(n-1)/2$ otherwise.
Let us introduce the indices $o\in \{1,\ldots,\bar{n}\}$
and $p,q \in \{1,\ldots,n\}$
with $p < n+1- q$ and $p\neq q$.
Symmetrizing the Pauli strings associated to $\gfree$ following Eq.~\eqref{eq:gnat},
basis elements for $\gnat$ are given as
\begin{subequations}
\label{eq:gnat:path}
\begin{align}
&iX_{\bar{n}+1} \;\text{ if $n$ is odd},\;
iX_{o}{+}iX_{n+1-o},\label{eq:gnat:path:a} \\
&iP_{oo}^{YY},\;
iP_{oo}^{ZZ},\;
iP_{oo}^{YZ}{+} iP_{oo}^{ZY}, \label{eq:gnat:path:b}\\
&iP_{pq}^{YY}{+} iP_{qp}^{YY},\;
iP_{pq}^{ZZ}{+} iP_{qp}^{ZZ},\;
iP_{pq}^{YZ}{+} iP_{qp}^{ZY}, \label{eq:gnat:path:c}
\end{align}
\end{subequations}
where 
\begin{equation*}
P_{ab}^{AB} := A_{a}X_{a+1}\cddot X_{n-b} B_{n+1-b} \;\text{ with }\; a+b \leq n.
\end{equation*}
We count the number of possibilities for the different cases in Eq.~\eqref{eq:gnat:path} before and after the symmetrization.
Recall that there are $\binom{n}{2} = n(n{-}1)/2$ pairs $(\tilde{p},\tilde{q})$ with $\tilde{p},\tilde{q} \in  \{1,\ldots,n\}$
and $\tilde{p}+ \tilde{q} \leq  n$. Thus $2n(n{-}1)$ basis elements from Eq.~\eqref{eq:gfree:path} are symmetrized to the cases
\eqref{eq:gnat:path:b}-\eqref{eq:gnat:path:c} and $n$ different elements from Eq.~\eqref{eq:gfree:path} result in
the cases \eqref{eq:gnat:path:a}, which
agrees with the total of $2n^2{-}n$ basis elements in Eq.~\eqref{eq:gfree:path}. For $n$ even, there are $\bar{n}$ cases in \eqref{eq:gnat:path:a}, $n$ cases for $iP_{oo}^{YY}$ and $iP_{oo}^{ZZ}$,
$\bar{n}$ cases for $iP_{oo}^{YZ}{+} iP_{oo}^{ZY}$, 
and $n(n{-}2)$ cases in \eqref{eq:gnat:path:c}.
For $n$ odd, there are $\bar{n}{+}1$ cases in \eqref{eq:gnat:path:a}, $2\bar{n}$ cases
for $iP_{oo}^{YY}$ and $iP_{oo}^{ZZ}$, 
$\bar{n}$ cases for $iP_{oo}^{YZ}{+} iP_{oo}^{ZY}$,
and $(n{-}1)^2$ cases in \eqref{eq:gnat:path:c}.
Both for even and odd $n$, a total of $n^2$ basis elements are provided in Eq.~\eqref{eq:gnat:path}.
Thus we have shown

\begin{lemma}
For a path graph $P_n$ with $n$ vertices, 
$\gnat$ has dimension $n^2$ and is spanned by
the elements in Eq.~\eqref{eq:gnat:path}.
\end{lemma}

There are different strategies to identify $\gnat$ as $\uu(n)$ 
based on explicit isomorphisms, subalgebra chains, and by ruling out
all other possibilities 
(see, e.g., \cite{zeier2011symmetry,ZZKS14}).
And detailed arguments in Appendix~\ref{proof:lemma:std:path:2}
verify the following statement:

\begin{lemma}\label{lemma:std:path:2}
For a $n$-vertex path graph $P_n$, 
$\gnat\iso \uu(n)$.
\end{lemma}

As a last step in our analysis of the
standard ansatz for the path graph, we prove in Appendix~\ref{proof:prop:std:path}
that $\gstd= \gnat \iso \uu(n)$:

\begin{proposition}[Path graphs]\label{prop:std:path}
\label{prop:std-mixer-path-graphs}
For a path graph $P_n$ with $n$ vertices, the associated 
standard-mixer \DLA $\gstd$ is equal to $\gnat$. Moreover, its basis is 
given by Eq.~\eqref{eq:gnat:path} and its
reductive decomposition is $\uu(n) \iso \su(n)\oplus\uu(1)$.
\end{proposition}

\subsection{Standard ansatz for cycle graphs\label{sec:std:cycle}}
Similarly as in Sec.~\ref{sec:std:path}, we now consider
the standard ansatz for cycle graphs $C_n$. 
The free-mixer \DLA observes $\gstd \subseteq \gnat \subseteq \gfree = \so(2n){\oplus} \so(2n)$
(see Theorem~\ref{thm:free-mixer-DLA-decomposition})
and Table~\ref{tab:free-basis-table} states an explicit basis of $\gfree$ given (up to factors of $i$) by $4n^2{-}2n$
Pauli strings
\begin{alignat*}{3}
&\mathrm{I}\cddot\mathrm{I}\mathrm{X}\mathrm{I}\cddot\mathrm{I},\quad
&&\mathrm{X}\cddot\mathrm{X}\mathrm{I}\mathrm{X}\cddot\mathrm{X},\\
&\mathrm{I}\cddot\mathrm{I}\yz\mathrm{X}\cddot\mathrm{X}\yz\mathrm{I}\cddot\mathrm{I},\quad
&&\mathrm{X}\cddot\mathrm{X}\yz\mathrm{I}\cddot\mathrm{I}\yz\mathrm{X}\cddot\mathrm{X}.
\end{alignat*}
Equivalently, we can write the basis of $\gfree$ as 
\begin{equation}
\label{eq:gfree:cycle}
iX_a,\,  iX_a X^{\otimes n},\,
iQ_{ab}^{YY},\, iQ_{ab}^{ZZ},\, iQ_{ab}^{Y\!Z},\text{ and } iQ_{ab}^{ZY}
\end{equation}
where $1\leq a \leq n$,
$1\leq b \leq n{-}1$, and
\begin{align*}
&Q_{ab}^{AB} := A_{\iota(a)}X_{\iota(a+1)}\cddot X_{\iota(a+b-1)} B_{\iota(a+b)}\;\text{ where}\\[1mm]
&\iota(c) = \iota_n(c):=
\begin{cases}
        n & \text{if $c\bmod n = 0$}, \\
        c\bmod n & \text{otherwise}.
\end{cases}
\end{align*}

The automorphism group of a cycle graph with $n$ vertices is given by the dihedral group
\cite{DF2004,adkins1992} that is generated by the permutations $(1\, 2 \cdots n)$ and
$\tilde{\sigma}_n$ where $\tilde{\sigma}_n$ is defined in Eq.~\eqref{eq:sigma:n}, i.e.,
\begin{equation*}
\AUT(C_n)
 = \group{(1\, 2 \cdots n),\, \tilde{\sigma}_n}.
\end{equation*}
In order to determine $\gnat$, we apply the symmetrization $\taunat$ from Eq.~\eqref{eq:sym:nat}
to the basis elements of $\gfree$ from Eq.~\eqref{eq:gfree:cycle}. We introduce the notation
\begin{equation*}
\tilde{Q}_0 := \sum_{\tilde{a}=1}^{n} X_{\tilde{a}},\,  \tilde{Q}_n := \sum_{\tilde{a}=1}^{n} X_{\tilde{a}} X^{\otimes n},\,
\tilde{Q}_{b}^{AB} := \sum_{\tilde{a}=1}^{n} Q_{\tilde{a}b}^{AB}.
\end{equation*}
and obtain $3(n{-}1)+2$ different basis elements 
\begin{subequations}
\label{eq:gnat:cycle}
\begin{align}
&i\tilde{Q}_0 = n\,\taunat(iX_a),\;
i\tilde{Q}_n = n\,\taunat(iX_a X^{\otimes n}), \\
&i\tilde{Q}_{b}^{YY} = n\, \taunat(iQ_{ab}^{YY}),\;
i\tilde{Q}_{b}^{ZZ} = n\, \taunat(iQ_{ab}^{ZZ}),\\
&i(\tilde{Q}_{b}^{YZ} {+} \tilde{Q}_{b}^{ZY}) = 2n\, \taunat(iQ_{ab}^{YZ}) = 2n\,
\taunat(iQ_{ab}^{ZY})
\end{align}
\end{subequations}
spanning $\gnat$. We have verified the following lemma:
\begin{lemma}
For a cycle graph $C_n$ with $n$ vertices, 
$\gnat$ has dimension $3(n{-}1)+2$ and it is spanned by
the basis elements in Eq.~\eqref{eq:gnat:cycle}.
\end{lemma}

We know that now that $\gstd \subseteq \gnat$ and that
the dimension of $\gnat$ is given by $3(n{-}1)+2$.
We refer to \cite{DAlessandro2024,allcock2024dla} for an explicit proof of the fact that
\begin{equation*}
\gstd = \gnat  
         \iso \underbrace{\su(2)\oplus\cdots\oplus\su(2)}_{n-1} \oplus\; \uu(1)\oplus\uu(1).
\end{equation*}
In addition, we point the reader to the original work of Onsager \cite{Onsager44}.

\subsection{Standard ansatz for the complete graph\label{sec:std:complete}}

We analyze the standard-ansatz \DLA $\gstd$ that
is contained in
$\gnat$ due to Proposition~\ref{thm_dla_chain}. The following Proposition~\ref{prop:std:complete}
observes an exponential
separation between the dimension of $\gstd$ and $\dim(\gfree)=2^{2n{-}1}-2$ for the complete graph.
This suggests that a standard ansatz could in fact be
overparametrized~\cite{larocca2021theory} and it would be able to solve the (trivial)
task of finding the maximum cut for the complete graph, even though the free ansatz
would very likely fail due to the presence of barren 
plateaus (see Corollary~\ref{cor:variance}). We have shown in Appendix~\ref{proof:prop:std:complete}
the following characterization:

\begin{proposition}\label{prop:std:complete}
Consider the complete graph $K_n$ with $n\geq 3$ vertices.
The dimension of $\gstd$ is bounded from above 
by a polynomial in $n$. In particular,
we have
\begin{equation*}
    \dim(\gstd) \leq \dim(\gnat) =
    \begin{cases}
        \frac{1}{2}\binom{n+3}{3} - 2 & \text{for $n$ odd}, \\
        \frac{1}{2}\binom{n+3}{3} + \frac{n}{4} - \frac{3}{2} & \text{for $n$ even}.
    \end{cases}
\end{equation*}
\end{proposition}

The result of Prop.~\ref{prop:std:complete}  for the complete graph has been independently obtained in \cite{allcock2024dla}
as an upper bound to the standard-ansatz \DLA $\gstd$. In \cite{allcock2024dla}, they have also
determined an explicit basis for $\gstd$ and thereby also an explicit formula for the dimension of $\gstd$ in this case.
Moreover, they have provided explicit bases for the semisimple part $[\gstd, \gstd]$ and
the center $\cent(\gstd)$.

\section{Invariant subspaces and representation theory\label{sec:character:computations}}

In contrast to Sec.~\ref{sec:standard-ansatz-partial-results} which emphasized specific
classes of graphs and explicit (Paul-string) bases for the associated \DLAs, this section
aims at general results for characterizing the invariant subspaces connected to the standard ansatz of
QAOA. We start in Sec.~\ref{sec:character:computations:setting} with the general 
representation-theoretic setting and first implications. This is continued in
Sec.~\ref{sec:character:computations:formulas} with explicit character computations that
lead to upper bounds on dimensions of irreducible representations
connected to the standard-ansatz \DLA $\gstd$.

\subsection{Setting and immediate implications\label{sec:character:computations:setting}}

Building on the notation and results from Sections~\ref{SEC:STD} and \ref{sec:lie:hierarchy},
we now focus even more on invariant subspaces and representation theory related to
the standard ansatz.
Set $\CC_{\std}:=\CC(\gstd)$.
The invariant subspaces of $\gstd$, $\gnat$, and $\unat$ 
(or for their semisimple parts)
are revealed  
by suitable basis changes so that the block decompositions
\begin{subequations}
\label{invariants_basis_many}
\begin{align}
\gstd & \simeq  \;\bigoplus_{\lambda=1}^{\dim[\cent(\CC_{\std})]}\;  \unity_{m_{\lambda}} \otimes \g_{\lambda},
\label{invariants_basis_gstd}
\\
\gnat & \simeq  \;\bigoplus_{\mu=1}^{\dim[\cent(\Snat)]}\; \unity_{\tilde{m}_{\mu}} \otimes \tilde{\g}_{\mu},
\;\text{ and} \label{invariants_basis_gnat}
\\
\unat & \simeq   \:\bigoplus_{\mu=1}^{\dim[\cent(\Snat)]}\;   \unity_{\tilde{m}_{\mu}} \otimes \uu(\tilde{d}_{\mu})
\label{invariants_basis_unat}
\end{align}
\end{subequations}
hold for suitably chosen
\DLAs 
\begin{equation*}
\g_{\lambda} \subseteq \uu(d_{\lambda}) \subseteq \C^{d_{\lambda}\times d_{\lambda}} \;\text{ and }\; 
\tilde{\g}_{\mu} \subseteq \uu(\tilde{d}_{\mu}) \subseteq \C^{\tilde{d}_{\mu}\times \tilde{d}_{\mu}}
\end{equation*}
with 
$\dim[\com(\g_{\lambda})]=\dim[\com(\tilde{\g}_{\mu})]  =1$. The decomposition in Eq.~\eqref{invariants_basis_gstd}
refines the ones in Eqs.~\eqref{invariants_basis_gnat}-\eqref{invariants_basis_unat}.
The action of the corresponding centers $\cent(\gstd)$ and $\cent(\gnat)$
can be quite involved and they can have support on different irreducible blocks in Eq.~\eqref{invariants_basis_many}
(see Fig.~\ref{fig:ex:house:hierarchy} and the related discussion in Sec.~\ref{sec:lie:hierarchy}).
We establish in Appendix~\ref{proof:lem:ZQ-inv-subspace} 
some basic properties of 
the invariant subspaces corresponding to
the \DLAs $\gfree$, $\gnat$, $\gorb$, $\gstd$, and $\unat$:

\begin{lemma}
\label{lem:ZQ-inv-subspace}
Consider a graph with $n$ vertices and
an invariant subspace
$W\subseteq \C^d$ of any associated \DLA $\g \in \{\gfree, \gnat, \gorb, \gstd, \unat\}$.
Then $\left(Z^{\otimes n}\right)W$ is also an invariant subspace of $\g$.
\end{lemma}

For an arbitrary subspace $W$, $\left(Z^{\otimes n}\right)W$ may have a nontrivial intersection with $W$. 
However, when $n$ is odd and $W$ is a irreducible, invariant subspace of one of our \DLAs of interest,
we can guarantee that $W$ and $\left(Z^{\otimes n}\right)W$ have a trivial intersection. 
Recall that a \DLA $\g \in \{\gfree, \gnat, \gorb, \gstd, \unat\}$
observes the invariant subspaces $\HC_{+}$
and $\HC_{-}$ of the same dimension
given by the $+1$ and $-1$ eigenspaces of $X^{\otimes n}$
(see Sec.~\ref{app:free:preliminaries}).
They are spanned by all Hadamard basis states  $\ket{b_1}\otimes\cdots\otimes\ket{b_n}$ with $b_j \in \{+, -\}$
with respectively an even or odd  number of minus
signs. For odd $n$, the irreducible subspaces in $\HC_{+}$
and $\HC_{-}$ match as is verified in Appendix~\ref{proof:thm:n-odd-matching-subspaces}:

\begin{proposition}[Odd number of vertices.]
\label{thm:n-odd-matching-subspaces}
Given a graph with an odd number $n$ of vertices, we consider an associated \DLA
$\g \in \{\gfree, \gnat, \gorb, \gstd, \unat\}$.
The invariant subspaces $\HC_{+}$ and $\HC_{-}$ have matching decompositions
$\HC_{+} = V_1\oplus\cdots\oplus V_k$ and $\HC_{-} = W_1\oplus\cdots\oplus W_k$
into the irreducible, invariant subspaces
$V_j$ and $W_j$ with $W_j = \left(Z^{\otimes n}\right)V_j$
and $1\leq j \leq k$.
\end{proposition}

\subsection{Character computations\label{sec:character:computations:formulas}}

After these preparations, we characterize
the invariant subspaces of $\gnat$
based on explicit character 
computations \cite{serre1977,Ledermann,serre1977,sengupta2012}.
This will yield upper bounds on the dimensions
of the irreducible representations of $\gstd$; recall 
from Proposition~\ref{thm_dla_chain}
that $\gstd \subseteq \gnat$.

For $\sigma_1 \in \mathcal{S}_2$, we set the element 
$\vartheta(\sigma_1) \in \ZTWO$ as
\begin{equation}\label{eq:z2:mat}
\vartheta(\sigma_1) := 
\begin{cases}
        I^{\tn} & \text{for $\sigma_1 = \symone $}, \\
        X^{\tn} & \text{for $\sigma_1 = (1\,2)$}.
    \end{cases}
\end{equation}
The permutation group 
$\Gnatr := \mathcal{S}_2 \times \Aut$ is represented 
via $\Upsilon(\sigma_1,\sigma_2):=\vartheta[\sigma_1] \zeta[\sigma_2]\in \C^{d\times d}$
as
the
group
of natural symmetries 
$
\Gnat = \ZTWO\times \zeta[\Aut]
$
from Eq.~\eqref{eq:nat:group}.
The associated character $\chi_{\text{nat}}$ is defined as 
\begin{equation}\label{eq:chi:hat}
\chi_{\text{nat}}(\sigma_1,\sigma_2)  := \Tr[\Upsilon(\sigma_1,\sigma_2)] = \Tr(\vartheta[\sigma_1]\, \zeta[\sigma_2])
\end{equation}
for $\sigma_1 \in  \mathcal{S}_2$ and $\sigma_2 \in  \Aut$. 
For a  permutation $\sigma \in  \mathcal{S}_n$,
its cycle type is defined as \cite{sagan2001}
\begin{equation}\label{eq:cycle:type}
(1^{b_1}, \ldots, n^{b_n})
\end{equation}
where $b_a$ denotes the number of cycles of length $a$ in $\sigma$.
In particular, $\sigma$ can be uniquely decomposed into a product 
\begin{equation}\label{eq:cycle:decomp}
\sigma = \prod_{a=1}^{c(\sigma)} (z_{a},\sigma[z_{a}],\ldots,\sigma^{q_a-1}[z_{a}])
\end{equation}
of disjoint cycles
of length $q_a$ such that $\sigma^{q_a}[z_{a}] = z_a$,
$z_a\in\{1,\ldots,n\}$, and the number of cycles is 
\begin{equation*}
c(\sigma):= \sum_{a=1}^{n} b_a.
\end{equation*}
For example, the identity $\symone=(1)\cdots(n) \in S_n$ can be uniquely written
as $n$ cycles of length one.
We set 
\begin{equation*}
\pp(\sg) :=
\begin{cases}
1 &\text{if all cycles of $\sg$ are of even length},\\
0 &\text{otherwise}.
\end{cases}
\end{equation*}
For odd $n$, no $\tilde{\sigma} \in \mathcal{S}_n$
has only even cycles, i.e., $\pp(\tilde{\sigma})=0$.
We efficiently compute
the character from Eq.~\eqref{eq:chi:hat}
which is shown in Appendix~\ref{proof:lem:character-formula}:
\begin{lemma}[Character formula]
\label{lem:character-formula}
We observe
\begin{equation*}
\chi_{\text{nat}}(\sigma_1,\sigma_2) =
    \begin{cases}
         2^{c(\sg_2)} & \text{if $\sigma_1=\symone$ or $\pp(\sg_2)=1$,}\\
        0 & \text{otherwise.}
    \end{cases}
\end{equation*}
\end{lemma}

We aim to determine the multiplicities of the irreducible
components of the character $\chi_{\text{nat}}$.
In particular, for each representation index $\nu= (\nu_1,\nu_2)$,
we are computing the character $\chi_{\nu}$ of the group $\Gnatr$
as
\begin{equation*}
\chi_{\nu}(\sigma_1,\sigma_2) 
= \alpha_{\nu_1}(\sigma_1)\,  \beta_{\nu_2}(\sigma_2)
\end{equation*}
where $\alpha_{\nu_1}$ and $\beta_{\nu_2}$ are the characters
of $\mathcal{S}_2$ and $\Aut$ with associated representation indices $\nu_1$ and $\nu_2$
as well as $\sigma_1 \in \mathcal{S}_2$ and $\sigma_2 \in \Aut$.
The multiplicity of $\chi_{\nu}$ in $\chi_{\text{nat}}$ is given by the Schur orthogonality relations
\cite{serre1977,Ledermann,sengupta2012}
\begin{equation}\label{eq:mults}
\mathbf{m}_{\nu}
:= \tfrac{1}{2\abs{\Aut}} \hspace{-5mm} \sum_{(\sigma_1,\sigma_2) \,\in\, \mathcal{S}_2 {\times} \Aut}
\hspace{-5mm} \bar{\chi}_{\nu}(\sigma_1,\sigma_2)\, \chi_{\text{nat}}(\sigma_1,\sigma_2).
\end{equation}
By duality (as detailed in Sec.~\ref{appendix:symmtry:analysis}), each multiplicity $\mathbf{m}_{\nu}$ corresponds to one of the dimensions
$\tilde{d}_{\mu}$ in Eq.~\eqref{invariants_basis_gnat}, while the associated multiplicity $\tilde{m}_{\mu}$ is given by
the degree $\mathbf{d}_{\nu}:={\chi}_{\nu}(\symone,\symone)$ of ${\chi}_{\nu}$.

We consider now the example of the house graph and we recall that its
automorphism group is isomorphic to $\mathcal{S}_2$, which has only the trivial and the sign representation
with the associated characters $\chi_{\triv}$ and $\chi_{\signrepr}$.
We obtain the following multiplicities and degrees
\begin{equation*}
\begin{matrix}
&\mathbf{m}_{(\triv,\triv)} &= 10, &\mathbf{m}_{(\signrepr,\triv)} &=10, &\mathbf{m}_{(\triv,\signrepr)} &= 6, &\mathbf{m}_{(\signrepr,\signrepr)} &= 6,\\
&\mathbf{d}_{(\triv,\triv)} &=\phantom{0}1, &\mathbf{d}_{(\signrepr,\triv)} &=\phantom{0}1, &\mathbf{d}_{(\triv,\signrepr)} &=1, &\mathbf{d}_{(\signrepr,\signrepr)} &=1.
\end{matrix}
\end{equation*}

We evaluate the multiplicity formula in Eq.~\eqref{eq:mults} for $\nu = (\triv,\triv)$ 
using 
$\chi_{(\triv,\triv)}(\sigma_1,\sigma_2)=1$, while Lemma~\ref{lem:character-formula}
determines
$\chi_{\text{nat}}(\sigma_1,\sigma_2)$.
We obtain the following formula:
\begin{lemma}\label{lem:triv:mult}
For the group $\Gnatr = \mathcal{S}_2 \times \Aut$ and
following Eq.~\eqref{eq:mults},
the multiplicity of the trivial character  $\chi_{(\triv,\triv)}$
in the character $\chi_{\text{nat}}$ from Eq.~\eqref{eq:chi:hat}
is given by
\begin{equation*}
\mathbf{m}_{(\triv,\triv)} =  \tfrac{1}{2\abs{\Aut}} \sum_{\sg\in\Aut} 2^{c(\sg)}\,[1+\pp(\sg)]\,.
\end{equation*}
\end{lemma}

The summands only depend on the cycle type of the permutations in $\Aut$, and therefore the
summation over the whole automorphism group could be replaced by a summation over conjugacy classes.
It will be instructive to characterize the multiplicity $\mathbf{m}_{(\triv,\triv)}$ for the two limiting cases of large and 
small automorphism groups $\Aut$:

\begin{proposition}\label{prop:trivial:multiplicity}
Let $\Aut$ denote the
automorphism group of a graph $G$
with $n$ vertices. The multiplicity $\mathbf{m}_{(\triv,\triv)}$
of the trivial character of $\Gnatr = \mathcal{S}_2 \times \Aut$
observes the following properties:
\begin{subequations}
\label{eq:trivial:multiplicity}
\begin{alignat}{3}
& \text{(a)}\quad && \Aut = \mathcal{S}_n \quad \text{implies}  \quad \mathbf{m}_{(\triv,\triv)} = \floor{\tfrac{n}{2}}+1,
\label{eq:trivial:multiplicity:a}
\\
& \text{(b)} && \abs{\Aut} = 1 \quad \text{implies}  \quad \mathbf{m}_{(\triv,\triv)} = 2^{n-1},
\label{eq:trivial:multiplicity:b}
\\
& \text{(c)} && \mathbf{m}_{(\triv,\triv)} \geq 2^{n-1} / \abs{\Aut}.
\label{eq:trivial:multiplicity:c}
\end{alignat}
\end{subequations}
\end{proposition}

Appendix~\ref{proof:prop:trivial:multiplicity} details the proof of
Proposition~\ref{prop:trivial:multiplicity}.
Proposition~\ref{prop:trivial:multiplicity}(a) shows that the trivial multiplicity
$\mathbf{m}_{(\triv,\triv)}$ is linear in the number $n$ of vertices if the automorphism group
is maximal. Using again the duality from Sec.~\ref{appendix:symmtry:analysis}
between one of the dimensions $\tilde{d}_{\mu}$
in Eq.~\eqref{invariants_basis_gnat} and $\mathbf{m}_{(\triv,\triv)}$, this dimension
$\tilde{d}_{\mu}$ of a particular representation of $\gnat$ is identified as being equal to $\mathbf{m}_{(\triv,\triv)}$.
Similarly for a trivial automorphism group
with $\abs{\Aut}=1$, Prop.~\ref{prop:trivial:multiplicity}(b) identifies a representation
in Eq.~\eqref{invariants_basis_gnat} of dimension $2^{n-1}$. For small values of $\Aut$
that are, e.g., polynomially bounded in the number $n$ of vertices, Prop.~\ref{prop:trivial:multiplicity}(c) leads to
an exponential lower bound for the corresponding dimension in Eq.~\eqref{invariants_basis_gnat}.

What does this imply for the standard ansatz of QAOA? 
Recall from Proposition~\ref{thm_dla_chain} that $\gstd \subseteq \gnat$.
So, the results in this section and particularly Proposition~\ref{prop:trivial:multiplicity}(c)
only provide an exponential lower bound on an upper bound for the dimension of an irreducible subspace
of $\gstd$
assuming that $\abs{\Aut}$ is polynomially bounded.
But the dimension of this irreducible subspace of $\gstd$ could be smaller.
The following Section~\ref{sec:conjecture} considers
 the particular case of graphs $G$ with $\abs{\Aut}=1$ (which are known as \emph{asymmetric} graphs),
which in particular implies that $G$ is archetypal (see Def.~\ref{def_other_graphs}).

\section{Standard ansatz for connected asymmetric graphs\label{sec:conjecture}}

We now explore properties of the standard ansatz for the case of connected  asymmetric graphs.
Connected \emph{asymmetric} graphs are connected graphs such that their automorphism group
consists only of the identity permutation \cite{Erdoes1963,Godsil2001}. Almost all graphs
are asymmetric and
there are no connected asymmetric graphs for $2 \leq n \leq 5$.
All connected  asymmetric graphs are archetypal in the sense of Def.~\ref{def_other_graphs}.
Applying the results of Sec.~\ref{sec:lie:hierarchy} to connected asymmetric graphs with $n$ vertices, one obtains
\begin{equation*}
\gstd \subseteq \gorb = \gnat = \gfree = \su(2^{n-1}) \oplus \su(2^{n-1}).
\end{equation*}
Also, Proposition~\ref{prop:trivial:multiplicity}(b) shows 
that the largest dimension for the invariant subspaces of $\gnat$ is given by $2^{n-1}$.
But how much smaller is the largest invariant dimension for the action of the standard-ansatz \DLA $\gstd$?

\begin{figure}[t]
\includegraphics{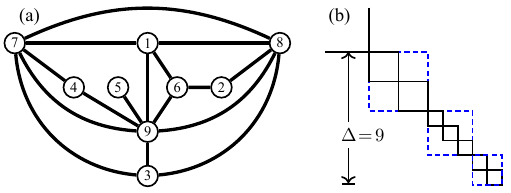}
\caption{\textbf{Asymmetric graph with nine vertices and its gap $\Delta = 9 = 2 {\times} 2 + 1 {\times} 3 + 1 {\times} 2$.} (a) Graph. 
(b) $\gstd^+$ acts on the upper-left component $\HC_{+}$ which splits into irreducible subspaces with
the respective dimensions
$247$, $2$, $1$, $1$ and multiplicities $1$, $2$, $3$, $2$.
Thus hidden symmetries can entail
irreducible subspaces of dimension larger than one.
\label{fig:graph:nine}}
\end{figure}

Appendix~\ref{app:free:preliminaries}
introduces a basis change $\tilde{h} \circ \pi_n \circ h$ that transforms
the projectors $P_{\pm}=(X\tn {\pm} I\tn)/2$ into
the respective block-diagonal form $I^{\otimes (n-1)} \oplus \zero^{\otimes (n-1)}$
and $\zero^{\otimes (n-1)} \oplus I^{\otimes (n-1)}$.
Moreover, the generators $H_p$ and $H_m$ of $\gstd$ are also transformed
into their block-diagonal form $H_p^+ \oplus H_p^-$ and $H_m^+ \oplus H_m^-$,
where $H_j^+$ and $H_j^-$ for $j \in \{p,m\}$ are respectively the upper-left 
and lower-right components. After these preparations, the \DLA acting on the upper-left block
is defined as
\begin{equation*}
\gstd^+ := \lie{H_p^+, H_m^+}, \;\text{ where }\; \Cstd^+:= \com(\gstd^+)
\end{equation*}
is its commutant. Let $\mathbf{m}$ denote the largest dimension
of the invariant subspaces of $\gstd^+$ and the gap from the case of $\gnat$ 
for connected asymmetric graphs
is given by
\begin{equation}\label{eq:gap}
\Delta := \mathbf{m}_{(\triv,\triv)} - \mathbf{m} = 2^{n-1} - \mathbf{m}.
\end{equation}
The graph in Fig.~\ref{fig:graph:nine} with nine vertices yields
a largest dimension of $\mathbf{m}=247$ for the action of $\gstd^+$
and a gap of $\Delta = 9 = 2^{9-1} - \mathbf{m}$. In this example, one observes irreducible subspaces
with a dimension larger than one as a consequence of hidden symmetries.
Thus additional symmetries beyond natural ones appearing in the standard 
ansatz \emph{cannot}  be fully explained by quantum
many-body scars~\cite{moudgalya2020,moudgalya2022hilbert, moudgalya2022exhaustive}
connected to one-dimensional invariant subspaces
as discussed at the end of Sec.~\ref{SEC:STD}.

\begin{figure}[t]
\includegraphics{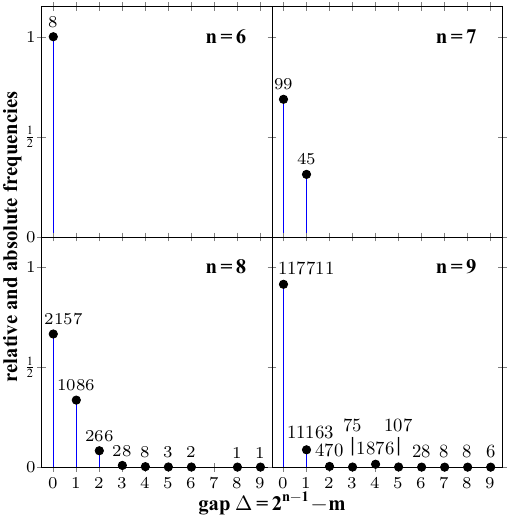}
\caption{\textbf{Relative and absolute frequencies for the gap $\Delta$
of connected asymmetric graphs with $n$ vertices.} The gap concentrates at
small values and is tightly bounded.
\label{fig:frequencies}}
\end{figure}

Figure~\ref{fig:frequencies} presents data on values for the gap $\Delta$
of connected asymmetric graphs and $6 \leq n \leq 9$. One observes that the gap concentrates at
small values and it is tightly bounded, at least for $n\leq 9$. 
This offers only very limited data, but the difference between
the cases of $\gnat$ and $\gstd$ (as quantified by $\Delta$) is quite small.
Similarly,
the corresponding dimensions $\dim(\Cstd^+)$ and $\dim[\cent(\Cstd^+)]$ 
for the commutant and its center
are
detailed in Table~\ref{tab:asymmetric}.
Thus the dimension
of the commutant $\Cstd^+$ is much more frequently small than large.
Also, the dimension of the center $\cent(\Cstd^+)$ appears to be quite restricted.
One might suggest that $\Delta$ will stay polynomially bounded in $n$:

\begin{conjecture}
Consider the standard ansatz for connected asymmetric graphs
with $n$ vertices. The gap $\Delta$ as defined in 
Eq.~\eqref{eq:gap} is polynomially bounded in $n$.
\end{conjecture}

But, from this data, it is
equally possible that $\Delta$ will eventually approach the magnitude of $2^{n-1}$ as $n$ grows.
Still, assuming a polynomially bounded $\Delta$, several situations of interest could occur.
On the one hand, if the state fully (or partially) belong to one of the polynomially sized subspace, then the action of QAOA could likely be simulated therein~\cite{cerezo2023does,goh2023lie}. Then, if the state belongs to an exponentially large invariant subspace, and we have full control in this invariant subspace, similar arguments as in Appendix~\ref{proof:variance} would imply the existence of barren plateaus (for quantum circuits
with enough layers). We leave these questions to future work.

\begin{table}
\caption{\textbf{Frequencies for $\dim(\Cstd^+)$ 
and $\dim[\cent(\Cstd^+)]$
of connected asymmetric graphs with $n$ vertices.}
There are 
$8$, $144$, $3552$, and $131452$ 
of these graphs for $6 \leq n \leq 9$.
    \label{tab:asymmetric}}
    \begin{tabular}{@{\hspace{1mm}}
    l @{\hspace{2mm}}
    l 
    @{\hspace{-6mm}} r
    @{\hspace{2mm}} r
    @{\hspace{2mm}} r 
    @{\hspace{2mm}} r 
    @{\hspace{2mm}} r 
    @{\hspace{2mm}} r 
    @{\hspace{2mm}} r
    @{\hspace{2mm}} r
    @{\hspace{2mm}} r
    @{\hspace{1mm}}}
    \\[-1mm]
    \hline\hline
    \\[-4mm]
    $n$ &&  \multicolumn{9}{@{}l}{Frequencies for $\dim(\Cstd^+)$ and $\dim[\cent(\Cstd^+)]$}
    \\ 
    & $\dim(\Cstd^+)=$ &
    1 & 2 & 3 & 4 & 5 & 6 & 7 & 9 & 10 
    \\[0.5mm] \hline
    \\[-4.5mm]
    6 &&  8 
    \\[-1mm]
    7 && 99 & 45 
    \\[-1mm]
    8 && 2157 & 1086 & 136 &  & 130 & 13 & 1 & 1 & 15 
    \\[-1mm]
    9 && 117711 & 11163 & 927 & 708 & 422 & 229 & 42 & 1 & 69
    \\[1mm]
    & $\dim(\Cstd^+)=$ &
    11 & 12 & 13 & 14 & 17 & 18 & 20 & 21 & 26  
    \\[0.5mm] \hline
    \\[-4.5mm]
    8 && 4 &  &  & 1 & 2 & 2 &  & 1 &
    \\[-1mm]
    9 && 107 & 8 & 3 & 1 & 20 & 12 & 1 &  & 10
    \\[1mm] 
    & $\dim(\Cstd^+)=$ &
    27 & 29 & 37 & 38 & 50 & 51 & 65 & 66 & 82 
    \\[0.5mm] \hline
    \\[-4.5mm]
    8 && &  & 1 & & &  & 1 &  & 1 
    \\[-1mm]
    9 && 6 &  1 & 2 & 4 & 2 & 2 & & 1
    \\[1mm] 
    & $\dim[\cent(\Cstd^+)]=$ &
    1 & 2 & 3 & 4 & 5 
    \\[0.5mm] \hline
    \\[-4.5mm]
    6 &&  8 
    \\[-1mm]
    7 && 99 & 45 
    \\[-1mm]
    8 && 2157 & 1236 & 158 & 1
    \\[-1mm]
    9 && 117711 & 11680 & 1287 & 772 & 2    
    \\[1mm]
    \hline\hline
    \end{tabular}
\end{table}

\section{Connection to the literature\label{sec:literature}}
 Given that our work has some overlap with other recent manuscripts, we begin by first discussing how our work connects to the existing literature. First, we recall that  Ref.~\cite{aguilar2024full} recently provided a characterization of all Lie algebras generated by Pauli operators based on frustration graph techniques~\cite{chapman2020characterization}. While our Lie algebras for the multi-angle case are contained within the families derived in~\cite{aguilar2024full}, the classification based on simple properties of a graph as the one presented here is less clear from the optics of frustration graphs, therefore making our results easier to follow from a maxcut perspective. 
 
Then,  shortly before our work was finalized, the work of Ref.~\cite{bakalov24} presented a classification for Lie algebras generated
by spin interactions on undirected graphs using an approach based on interaction graphs. The results in~\cite{bakalov24} are very close to ours in the sense that the generators are composed from single-qubit operators over vertices of a graph, and of two-qubit operators defined over the vertices of a graph. As such,  the results of Theorem~\ref{thm:free-mixer-DLA-decomposition} are a special case of those in Ref.~\cite{bakalov24} for special choices of operators (i.e., $X$ on vertices and $ZZ$ on edges). Still, while Ref.~\cite{bakalov24} characterizes the \DLAs, our work goes beyond these results as we provide an in-depth description of the
symmetries and invariant subspaces, isomorphisms, as well as explicit bases for the \DLAs. 

Here, we also remark that Refs.~\cite{aguilar2024full,bakalov24} only consider Pauli generators, and not summations thereof. In this context, we instead recall that Ref.~\cite{allcock2024dla} does study Lie algebras for the standard mixer (although they do not consider the multi-angle QAOA ansatz). Therein, the authors present an upper bound on the dimension of the standard \DLA by defining a \DLA that respects the automorphisms of the considered graph. 
In this sense, our upper bound of Proposition~\ref{thm_dla_chain} in terms of the natural \DLA (which respects automorphism group and parity) is in general 
tighter than that of Ref.~\cite{allcock2024dla}. 
But our and their upper bound agree for the complete graph as
they consider further symmetries (as we do). 
The results of \cite{DAlessandro2024,allcock2024dla} for cycle graphs and 
of \cite{allcock2024dla} for
complete graphs go beyond the upper bounds
presented in our work. 
The path graphs are not considered in \cite{DAlessandro2024,allcock2024dla}.
For the cycle graphs, we also point the reader to the work of Onsager~\cite{Onsager44}.

\section{Discussion\label{sec:discussion}}

Computing the \DLA of parametrized circuits has become a central and important tool for the study and characterization of variational quantum models. In this context, our work contributes to this body of knowledge by evaluating the Lie algebras  for three QAOA ansätze. Our main result constitute a full characterization of the multi-angle \DLA for arbitrary circuits. Then, we argue that the presence of hidden symmetries for the orbit and the standard ansätze will likely make a general classification unlikely. Here, we instead define the natural \DLA---which respects only the natural symmetries---and use it to upper bound the dimension of the orbit and the standard-ansatz \DLAs. Importantly, we find evidence that the largest component of the invariant subspaces of the standard-ansatz \DLA is of very similar dimension to that of the natural \DLA (i.e., their difference may be only polynomially in the number of vertices), and conjecture that this result could hold for the vast majority of graphs. 

Our results have several important implications for QAOA, as well as for variational algorithms in general. On the one hand, we show that the multi-angle QAOA ansatz leads to exponentially large \DLAs, and therefore is extremely prone to exhibiting trainability barriers such as barren plateaus, even when the circuit contains a single layer. Then, if our conjecture is true, this would mean that the standard QAOA \DLA has only exponentially large, or polynomially small components. In this case, the prospects of QAOA become bleak as its action on the polynomial subspace could be classically simulable, while that in the exponentially large subspace could be prone to barren plateaus. As QAOA can likely only significantly outperform state-of-the-art classical methods with a deep enough circuit~\cite{crooks2018performance}, one would almost be guaranteed to have untrainable models in the regimes where they become truly useful.

A caveat to this argument is the fact that barren plateaus are an average statement that hold for randomly initialized circuits, whereas smart initialization techniques could avoid such average-case issues~\cite{larocca2024review}. Indeed, \cite{zhou2020quantum,FarhiQuantum2022,boulebnane2021,Basso2022TQC,Basso2022TQC_arXiv}
have proposed pre-optimized initial angles 
for $D$-regular graphs.
This is possible as
measured cut values
concentrate---in the infinite vertex limit---tightly around maxima both with repeated measurements
and varying graph instances \cite{brandao2018fixed,boulebnane2021,Basso2022TQC,Basso2022TQC_arXiv,
boulebnane,Basso2022,Basso2022arXiv}. 
For this analysis, 
properties of $D$-regular graphs are connected (in the infinite vertex limit) with the Sherrington-Kirkpatrick model of complete graphs with randomly
weighted edges \cite{FarhiQuantum2022,dembo2017,boulebnane2021,Basso2022TQC,Basso2022TQC_arXiv,
boulebnane} (refer to \cite{blekos2024} for a review). In particular,
the optimal angles for the problem Hamiltonian in Eq.~\eqref{eq:std:problem}
are rescaled with $1/\sqrt{D}$ as $D$ increases while 
the angles for the mixer Hamiltonian in Eq.~\eqref{eq:std:mixer} are not rescaled
\cite{boulebnane2021,Basso2022TQC,Basso2022TQC_arXiv,boulebnane}.
For proving this, the $D$-regular graph
is assumed to be similar to a tree, where its girth (which is the length of
its shortest cycle) is larger than $2L{+}1$ for $L$ layers.
The effect of rescaling the optimal angles for the problem Hamiltonian can be clearly seen, e.g., in Figure~9 of \cite{Vijendran2024}
where the size of the barren plateau is increasing with $D$. A similar effect is apparent in Figure~\ref{fig:scaling} of 
Sec.~\ref{sec:free:implications}. The rescaling can thus be interpreted as
a concentration of the cut values with respect to variations in the angles for large-girth $D$-regular graphs
in the infinite vertex limit.

All of this nicely fits with barren plateaus for the cut value where
the variance in the angles vanishes exponentially fast 
with the number of vertices. Thus our work aims
at arbitrary connected graphs (and particularly asymmetric ones)
and not only at $D$-regular ones with a large girth.
Even though barren plateaus have only been shown for the multi-angle ansatz, our work prepares the ground for a more
general symmetry analysis. Connections to the study of smart-initializations are left for future work.

More generally, there has been a string of recent work
highlighting limitations of QAOA \cite{bravyi2020obstacles,farhi2020typical,farhi2020wholegraph,chou2022,Basso2022,Wilhelm23,Scriva2024,
MunozArias2024,gerblich2024advantages}, e.g.,
as a result of clustering of solutions \cite{farhi2020typical,chou2022}, or the fact that QAOA can be outperformed by quantum walks~\cite{gerblich2024advantages}, or the detrimental effect of intrinsic shot noise in the measurements for the optimization \cite{Scriva2024},
or by quantum-inspired classical algorithms~\cite{Wilhelm23,MunozArias2024}. On the classical side, there exists highly competitive 
algorithms \cite{rendl2010,dunning2018,juenger2021,nguyen2021,rehfeldt2022,charfreitag2022}, also including
approximation algorithms
\cite{williamson2011,papadimitrou1991,goemans1995improved,arora1999,
jansen2005,gharibian2019,gaar2020}. Moreover, there has been a revolution in solving
satisfiability problems (to which maxcut can be reduced to)
\cite{fichte2023,knuth4B}. This means that QAOA is under pressure
due to its seeming limitations and its strong classical competition.

Then, another important implication of our work is the fact that it unveils the crucial importance of the encoding scheme used to embed classical problems in quantum ones. Indeed, while the standard maxcut encoding seems extremely simple and natural, as it uses local operators that respect 
the parity superselection rule and graph automorphisms, it also leads to hidden symmetries. These spurious symmetries, in turn, lead to unexpected invariant subspaces. Critically, the understanding of these symmetries, and their ensuing effects on the irreducible decompositions, is crucial, as they could avoid using initial states that are not mostly constrained to some subspace which could potentially have a small (or no) overlap with  the solution manifold.  

Indeed, our work also contributes to developing symmetry tools for a systematic analysis of variational quantum algorithms.
Building on tools from \cite{zeier2011symmetry,ZZKS14,zeier2015squares,zimboras2015symmetry,schulte2017,schulte2018},
our symmetry analysis can be seen as an initial step towards a better understanding
of variational quantum algorithms and their strengths and limitations.
For QAOA, our work connects methods based on graph automorphisms as in \cite{shaydulin2021symmetry,shaydulin2021classical}
with their impact on the symmetries of the full Hilbert space. We have modeled this using our natural symmetries, also 
incorporating the parity-superselection operator $X\tn$. The presence of hidden symmetries
is often related to one-dimensional invariant subspaces and this is curiously connected
to quantum many-body scars~\cite{moudgalya2020,moudgalya2022hilbert, moudgalya2022exhaustive},
for which one-dimensional invariant subspaces are a necessary condition.
Figure~\ref{fig:graph:nine} highlights an example
with hidden symmetries resulting in invariant subspaces of dimension larger than one
which cannot be fully explained by quantum many-body scars.

The graph properties also play an important role at the interplay between symmetry and locality and this interplay
has been extensively studied in recent years
\cite{marvian2022restrictions, marvian2022rotationally, marvian2022qudit, marvian2023non, kazi2023universality}.
Thus local unitaries that act only on a fixed number of qubits at a time and that also respect a certain symmetry group 
cannot in general generate all unitaries that respect that symmetry group, which is in contrast to
the universality of one- and two-qubit quantum gates \cite{divincenzo1995two,lloyd1995almost}.
Such locality restrictions (as studied in \cite{Kraus2007,zeier2011symmetry,ZZKS14,wiersema2023}) are also exemplified
by the parity-superselection symmetry operator $X^{\otimes n}$: it is not contained in the multi-angle \DLA,
despite obviously respecting its symmetries.
This can be understood as a limitation on the control of the
relative phase between the even-parity and odd-parity sectors \cite{marvian2022restrictions}.
Or in a more algebraic interpretation, the center of the \DLA is highly constraint which has been identified 
as a potential limitation for the ability to simulate certain quantum dynamics \cite{zimboras2015symmetry}.
Section~\ref{sec:lie:hierarchy} highlights
the intricacies
related to the role that the center of the \DLA plays in the hierarchy of QAOA ansätze.

In the grander scheme of things, we believe that our work bring a new and fresh perspective to the \DLAm analysis 
of quantum circuits
through the optics of symmetries. As such,  we hope that the insights presented here could serve as guiding principle towards the final frontier: \DLAs that arise from generators expressed as sums of Paulis.     

\begin{acknowledgments}

RZ thanks Thomas Schulte-Herbrüggen, Zoltán Zimborás, and Michael Keyl for many discussions and insights
on symmetries of controlled quantum systems. RZ also appreciates the illuminating discussions
with Armin Römer and Juhi Singh on related projects as well as
with Nikkin Devaraju on the original work of \cite{farhi2014quantum}.
Last but not least, RZ thanks Roberto Gargiulo for many closely related discussions
which have also led to the follow-up work \cite{gargiulo24}. We acknowledge computations
with the computer algebra system \textsc{Magma} \cite{magma1997}.

SK acknowledges initial support by the U.S. DOE through a quantum computing program sponsored by the LANL Information Science \& Technology Institute. 
ML was also supported by the Center for Nonlinear Studies at Los Alamos National Laboratory (LANL). ML and MC acknowledge support by the Laboratory Directed Research and Development (LDRD) program of LANL under project numbers 20230049DR and 20230527ECR. MC was initially supported by LANL's ASC Beyond Moore’s Law project.
RZ acknowledges funding
under Horizon Europe programme HORIZON-CL4-2022-QUANTUM-02-SGA via the project 
\href{https://doi.org/10.3030/101113690}{101113690} (PASQuanS2.1)
and from the European High-Performance Computing Joint Undertaking (JU) under grant agreement No 
\href{https://doi.org/10.3030/101018180}{101018180} (HPCQS). The JU receives support from the 
European Union’s Horizon 2020 research and innovation programme and Germany, France, Italy, Ireland, Austria, Spain.
\end{acknowledgments}

\appendix

\section{Methods for analyzing symmetries\label{app:symmetry:theory}}

\subsection{Symmetries and isotypical decomposition\label{appendix:symmtry:analysis}}

We extend Section~\ref{sec:dla} by further detailing how to
analyze symmetries
while pointing to the house graph and the explicit information
in Table~\ref{table:house:graph} (see also Fig.~\ref{fig:decomposition:house}).
Recall that $\g = \lie{i\mathcal{G}}$ denotes the (real) \DLA
generated from a set $\mathcal{G}$ of hermitian generators.
Note that $\g$ is contained in the \DLA $\uu(d) \subseteq \C^{d\times d}$
of skew-hermitian matrices for $d=2^n$. We obtain the Lie group $\exp(\g) \subseteq \Ubb(d)$
which is contained in the unitary group.

The linear symmetries have been specified by the 
commutant $\CC = \com(\mathcal{G})$
[see Eq.~\eqref{eq:commutant}]
which is closed under complex-linear combinations and matrix multiplication.
We also introduce the matrix algebra $\mathcal{A}= \alg{\mathcal{G}} = \alg{i\mathcal{G}}$
that contains all complex-linear combinations of 
the identity matrix $\unity_d \in \C^{d\times d}$ and
arbitrary products of generators in $\mathcal{G}$. The properties
\begin{align}
\CC & =\com(\mathcal{A}) = \com(\mathcal{G}) =\com(\g) \,\text{ and} \label{eq:dual:a}\\
\mathcal{A} & = \com(\CC) = \com(\com(\mathcal{A})) \neq \g \label{eq:dual:b}
\end{align}
establish
a duality between $\mathcal{A}$ and $\CC$. Equation~\eqref{eq:dual:b}
is the double commutant property of $\mathcal{A}$.
We point to textbooks in algebra
\cite{bresar2014,lorenz2008,jacobson1985,jacobson1989}
and representation theory 
\cite{lux2010,curtisreiner1962,curtisreiner1981}
as well as original work by Emmy Noether
\cite{noether1929,noether1933,noether1983},
which has been publicized by
Bartel Leendert van der Waerden \cite{waerden1991}
and Hermann Weyl \cite{weyl1936,weyl1953}. 

The invariant subspaces induced by the actions
of $\mathcal{G}$, $\mathcal{A}$, $\g$, and $\exp(\g)$ 
all agree
and they decompose into irreducible subspaces
as $\exp(\g)$ is contained in the unitary group.
Moreover, the centers $\cent(\mathcal{A})= \cent(\CC) = \mathcal{A} \cap \CC$
are equal, where [generalizing the definition in Eq.~\eqref{eq:center:commutant}]
\begin{equation*}
\cent(\mathcal{M}) = \{ Z \in \mathcal{M} \;\text{s.t.}\; [Z,M]=0 \;\text{for all}\; M \in \mathcal{M}\}.
\end{equation*}
This implies that
the actions of $\mathcal{A}$ and $\CC$ 
induce the same isotypical decomposition $\C^d=\oplus_{\lambda} {\mathcal{I}}_{(\lambda)}$
as described now:
An \emph{isotypical decomposition} 
consists of invariant subspaces ${\mathcal{I}}_{(\lambda)}$ that are the 
direct sum of all irreducible subspaces isomorphic to one particular irreducible subspace.
The isotypical decomposition refines in general into the irreducible decomposition.
However, for the free and the standard ansatz corresponding to the house graph (see Table~\ref{table:house:graph}
and Fig.~\ref{fig:decomposition:house}),
the isotypical and the irreducible decomposition 
coincide and one respectively obtains 
$\C^{16} \oplus \C^{16}$
and $\C^{10} \oplus \C^{5} \oplus \C \oplus \C^{10} \oplus \C^{5} \oplus \C$.
Also, the matrix algebra $\mathcal{A}$ is respectively isomorphic to
$\C^{16\times 16} \oplus \C^{16\times 16}$ and
$\C^{10\times 10} \oplus \C^{5\times 5} \oplus \C^{1\times 1} 
\oplus \C^{10\times 10} \oplus \C^{5\times 5} \oplus \C^{1\times 1}$ and
it has the respective dimension $512$ and $252$.

But the fine structure in the isotypical components ${\mathcal{I}}_{(\lambda)}$ differs for the action of $\mathcal{A}$ and $\CC$
as ${\mathcal{I}}_{(\lambda)}$ contains irreducible subspaces isomorphic to
${\mathcal{I}}_{\lambda}^{\mathcal{A}}$ and ${\mathcal{I}}_{\lambda}^{\CC}$
with 
\begin{align*}
m_{\lambda} &:= \mult({\mathcal{I}}_{\lambda}^{\mathcal{A}},\C^d) = \dim({\mathcal{I}}_{\lambda}^{\CC})\,,\\
d_{\lambda} &:= \dim({\mathcal{I}}_{\lambda}^{\mathcal{A}}) = \mult({\mathcal{I}}_{\lambda}^{\CC},\C^d).
\end{align*}
Here, the \emph{multiplicity} $\mult({\mathcal{J}},{\mathcal{I}})$
counts how many times an irreducible ${\mathcal{J}}$
is equivalent to an irreducible ${\mathcal{I}}_j$ in the decomposition
${\mathcal{I}} =\oplus_j {\mathcal{I}}_j$.
The dimensions $d_{\lambda}$ and multiplicities $m_{\lambda}$ 
are determined by the commutant $\CC$.
Complementing Eq.~\eqref{invariants_basis}, we have
\begin{equation*}
\mathcal{A} \simeq  \hspace{-3mm} \bigoplus_{\lambda=1}^{\dim[\cent(\CC)]}  \hspace{-3mm}
\unity_{m_{\lambda}} \otimes \C^{d_{\lambda}\times d_{\lambda}}.
\end{equation*}

We now discuss a few prototypical examples
in order to highlight the intricacies in identifying \DLAs and their
representations via their linear symmetries.
We first consider 
three two-qubit examples which stress that
the \DLA $\g$ is not uniquely determined even
if the generators act irreducibly, i.e., even if $\dim(\CC)=1$ \cite{zeier2011symmetry}.
Let 
\begin{align*}
\mathcal{G}_a &:= \mathcal{G}_b \cup \{Z_1Z_2\},\;
\mathcal{G}_b :=\{X_1, X_2, Z_1, Z_2\},\;\text{and}\\
\mathcal{G}_c &:=\{
\left(
\begin{smallmatrix}
0 & \sqrt{3}/2 & 0 & 0\\
\sqrt{3}/2 & 0 & 1 & 0\\
0 & 1 & 0 & \sqrt{3}/2\\
0 & 0 & \sqrt{3}/2 & 0
\end{smallmatrix}
\right),
\left(
\begin{smallmatrix}
{3}/{2} & 0 & 0 & 0\\
0 & {1}/{2} & 0 & 0\\
0 & 0 & -{1}/{2} & 0\\
0 & 0 & 0 & -{3}{2}
\end{smallmatrix}
\right)
\}
\end{align*}
denote the corresponding generators.
We obtain three different Lie
algebras of dimension $15$, $6$, and $3$ which are isomorphic 
to $\su(4)$, $\su(2) {\oplus} \su(2)$, and $\su(2)$. We detail how
these \DLAs are irreducibly embedded into $\C^{4\times 4}$.
Recall the standard representation $\kappa$ of a \DLA.
For the three examples, the representations are $\kappa$, $\kappa {\otimes} \kappa$,
and the spin-$3/2$ representation, respectively. Here, 
\begin{equation*}
(\gamma \otimes \tilde{\gamma})(g,\tilde{g}):=\gamma(g) \otimes \unity_{\dim(\tilde{\gamma})} + \unity_{\dim(\gamma)} \otimes \tilde{\gamma}(\tilde{g})
\end{equation*}
is the 
tensor product of the representations $\gamma$ and $\tilde{\gamma}$ for the \DLAs $\mathfrak{k}$ and $\tilde{\mathfrak{k}}$
with  $g \in \mathfrak{k}$ and $\tilde{g} \in \tilde{\mathfrak{k}}$. In summary, the examples
$\mathcal{G}_a$, $\mathcal{G}_b$, and 
$\mathcal{G}_c$ illustrate three prototypical variants on how generators and the generated \DLA can act irreducibly.

Extending the discussion from an irreducible action to a block-diagonal action with two blocks,
we examine three three-qubit examples with an irreducible decomposition $\C^4 {\oplus} \C^4$.
The corresponding generators are
\begin{align*}
\mathcal{G}_d &:= \{X_2, X_3, Z_2, Z_3, Z_2Z_3\} {\times} (\unity_8{\pm}Z_1),\\
\mathcal{G}_e &:= \{X_2, X_3, Z_2, Z_2Z_3, Z_1X_2X_3\},\\
\mathcal{G}_f &:=  \{X_2, X_3, Z_2, Z_3, Z_2Z_3\}.
\end{align*}
For the first two examples, $\dim(\CC)=\dim(\cent(\CC))=2$ and
the isotypical and the irreducible decompositions agree. We have $\dim(\CC)=4$
and $\dim(\cent(\CC))=1$ for the third example. This results in an isotypical
decomposition with a single isotypical component $\C^8$ which contains the irreducible component
$\C^4$ with multiplicity two. The corresponding \DLAs have dimensions $30$, $15$, and $15$
and are isomorphic to $\su(4) {\oplus} \su(4)$, $\su(4)$, and $\su(4)$. We obtain the respective representations
$[\kappa{\otimes}\epsilon]{\oplus}[\epsilon{\otimes}\kappa]$,
$\kappa{\oplus}\overline{\kappa}$, and 
$\kappa{\oplus}\kappa$. Here, $\epsilon$ denotes the trivial representation
and $\overline{\gamma}$ is the dual of a representation $\gamma$.
One observes characteristic differences and that even the isomorphic \DLAs
of $\mathcal{G}_e$ and $\mathcal{G}_f$ are represented differently.

With all these examples, it is evident that linear symmetries
and the commutant 
will be in general not sufficient to resolve the
intricate structure of possibly occurring \DLAs.
Additional techniques are required to analyze general
dynamical quantum systems.
The intuition leading to the results 
in Appendices~\ref{appendix:free-mixer} and
\ref{appendix:free-mixer:bases} for the free-mixer ansatz
is to certain degree based on so-called
quadratic symmetries \cite{zeier2011symmetry,
zeier2015squares,zimboras2015symmetry,schulte2017,schulte2018}.

\subsection{Details on natural and hidden symmetries\label{app:natural:hidden}}

In this subsection,
we further discuss the structure of natural and hidden symmetries following 
the approach of Sec.~\ref{SEC:STD}. To this end, we consider
the symmetry decomposition in a more intricate standard-ansatz QAOA example following Fig.~\ref{fig:decomposition:four}.
Here, we also rely on the isotypical
decomposition as detailed in Appendix~\ref{appendix:symmtry:analysis}
and the corresponding isotypical projectors that project onto the respective 
isotypical component. 

\begin{figure}[t]
\includegraphics{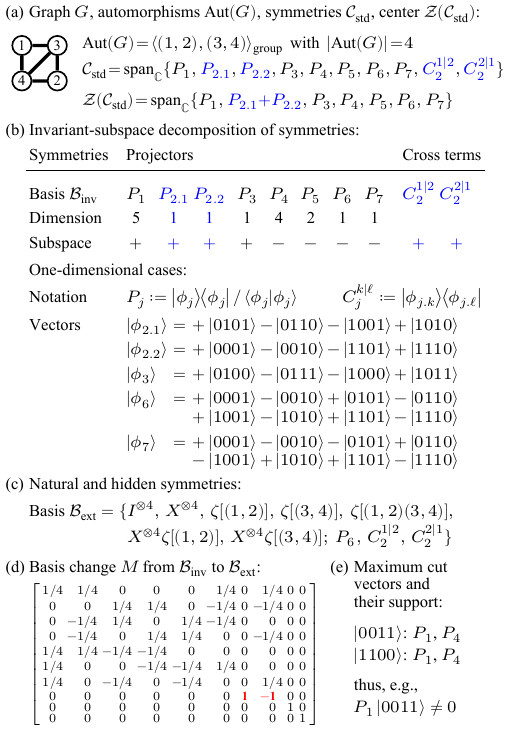}
\caption{\textbf{Symmetry decomposition of a standard-ansatz QAOA example.} Notation as in Fig.~\ref{fig:decomposition:house},
but we discuss a case that is not multiplicity-free: the isotypical projector (see App.~\ref{appendix:symmtry:analysis})
$P_{2.1}{+}P_{2.2}$ splits into the projectors $P_{2.1}$ and $P_{2.2}$ and cross terms $C_2^{k|\ell}$ appear in $\Cstd$
(see also Fig.~\ref{fig:matrices:four} in App.~\ref{app:projections}).
Relevant parts are marked in blue in (a) and (b). For the natural symmetries in (c), $X^{\otimes 4}\zeta[(1,2)(3,4)]$ does not appear
in the basis of natural symmetries $\Snat$ as it is linear dependent to the other basis elements. Thus $\dim(\Snat)< 2 \abs{\Aut}$.
\label{fig:decomposition:four}}
\end{figure}

Figure~\ref{fig:decomposition:four} highlights a four-vertex graph with an automorphism group
consisting of four elements. The symmetries given by the commutant $\Cstd$
are ten dimensional, while the corresponding center $\cent(\Cstd)$ is
seven dimensional. We observe that the isotypical projector
$P_{2.1}{+}P_{2.2}$ splits into the one-dimensional projectors $P_{2.1}$ and $P_{2.2}$, which are explicitly specified
in Fig.~\ref{fig:decomposition:four}(b). This implies that the corresponding representation is not multiplicity-free, i.e., the multiplicity $m_{\lambda}$ in the decomposition of Eq.~\eqref{invariants_basis} is not equal to one.
Moreover, cross terms $C_2^{k|\ell}$ appear in $\Cstd$ as specified in Fig.~\ref{fig:decomposition:four}(b).
As for the example of the house graph in Fig.~\ref{fig:decomposition:house}, the basis change in Fig.~\ref{fig:decomposition:four}(d)
clarifies that the division into natural and hidden symmetries is not unique. The red, nonzero entries in the third-last row of
Fig.~\ref{fig:decomposition:four}(d)
show that both $P_6$ and $P_7$ are hidden symmetries, but $P_6$ and $P_7$ are---up to natural symmetries---linear-dependent.
The projectors $P_{2.1}$, $P_{2.2}$, and $P_3$ are examples of one-dimensional projectors that are part of the natural symmetries
as one can infer from the basis change in Fig.~\ref{fig:decomposition:four}(d).
Even this simple, low-dimensional example makes the intricate structure of standard-ansatz symmetries 
and their variability evident.

We finally characterize the one-dimensional projectors and the corresponding 
cross terms in Fig.~\ref{fig:decomposition:four} while utilizing
the corresponding simultaneous eigenvectors 
$\ket{\psi_j}$
for a given set of Hamiltonians $\mathcal{G}$ 
following Eq.~\eqref{eq:sim:eigen} in Section~\ref{SEC:STD}. In particular,
\begin{equation*}
H\ket{\psi_j} = \beta(H,\ket{\psi_j})\ket{\psi_j}\, \text{ holds for every }\, H\in \mathcal{G}.
\end{equation*}
Clearly, the one-dimensional subspace spanned by the vector $\ket{\psi_j}$ is invariant under the action
of every $H$ and the same holds for its orthogonal complement. Thus one
verifies
$H \ketbra{\psi_j}{\psi_j} - \ketbra{\psi_j}{\psi_j} H = 0$ and the simultaneous eigenvectors 
$\ket{\psi_j}$ can be chosen as part of an orthogonal basis. One immediately obtains
\begin{align*}
& (H \ketbra{\psi_j}{\psi_k} - \ketbra{\psi_j}{\psi_k} H)\, \ket{\psi_k}\\
=\,& H \braket{\psi_k}{\psi_k} \ket{\psi_j} - \ketbra{\psi_j}{\psi_k} \beta(H) \ket{\psi_k}\\
=\,& [\beta(H, \ket{\psi_j})- \beta(H,\ket{\psi_k}) ] \braket{\psi_k}{\psi_k} \ket{\psi_j}
\end{align*}
and it follows that $\ketbra{\psi_j}{\psi_k}$ is in the commutant of $\mathcal{G}$
if and only if the eigenvalues $\beta(H,\ket{\psi_j})$ and 
$\beta(H,\ket{\psi_k})$ are equal for each $H \in \mathcal{G}$.
In other words, the eigenvalues $\beta(H,\ket{\psi_j})$
of the simultaneous eigenvectors $\ket{\psi_j}$ can used
to uniquely identify the corresponding one-dimensional
representations up to multiplicity.

\section{Analysis of the free-mixer ansatz\label{appendix:free-mixer}}

This appendix proves Theorem~\ref{thm:free-mixer-DLA-decomposition}
and thereby determines the free-mixer \DLAs for connected graphs. 
We then establish the free-mixer \DLAs for the particular cases of 
path graphs (see Thm.~\ref{app:thm:path}), cycle graphs (see Thm.~\ref{app:thm:cycle}),
connected bipartite graphs different from path and cycle graphs (see Thm.~\ref{thm:free-mixer-bipartite}),
as well as connected non-bipartite graphs
different from cycle graphs (see Thm.~\ref{thm:free-mixer-revisited}).
The results are summarized in 
Table~\ref{tab:free-mixer-full-table},
while the general discussion of explicit Pauli string bases is deferred to Appendix~\ref{appendix:free-mixer:pauli}.
We start with some preliminaries in Appendix~\ref{app:free:preliminaries}
and summarize two techniques to characterize \DLAs
in Appendix~\ref{app:free:irred} where the first technique is applicable if the \DLA acts irreducibly and
the second one is based on maximal subalgebras.
Path and cycle graphs are then treated in Appendix~\ref{appendix:sub:pathcycle}, while 
Appendix~\ref{appendix:sub:bipartite} contains the most extensive discussion detailing the case of bipartite graphs.
The non-bipartite graphs are considered in Appendix~\ref{appendix:sub:non:bipartite}. The explicit representations connected
to the various cases of free-mixer \DLAs are determined in Appendices~\ref{appendix:explicit:reps:path:cycle}
and \ref{appendix:explicit:reps:not:path:cycle}.
The approach in this appendix is complemented in Appendix~\ref{appendix:free-mixer:bases} with a focus
on bases of Pauli strings and their relation to edges of graphs.

Some of the results and the applied techniques are reminiscent of those in 
\cite{zeier2011symmetry,ZZKS14}. In particular, the maximal \DLA $\su(2^{n-1}){\oplus}\su(2^{n-1})$
(see Lemma~\ref{app:prop:structure} and Thm.~\ref{thm:free-mixer-revisited}) is reflected
in the parity superselection rule for fermionic systems as discussed in \cite{ZZKS14,BCW2007,WWW1952}.
For path and cycle graphs,  we obtain the free-mixer
\DLAs $\so(2n)$ and $\so(2n){\oplus}\so(2n)$ 
with their particular embeddings into $\su(2^n)$, which
manifest as spinor representations of the orthogonal group \cite{BrauerWeyl1935}
and which are widely encountered as symmetries in physics 
\cite{Murnaghan1938,Kaufman1949,Boerner1969,Miller1972,SW86,FH91,Georgi1999,Zee2016}.

\begin{table*}
    \caption{Basis, isomorphism type, and dimension of the free-mixer \DLA
    for path, cycle, and connected bipartite graphs, as well as connected non-bipartite graphs different from cycle graphs.
    Here, $\#\mathrm{X}$, $\#\mathrm{Y}$, $\#\mathrm{Z}$,
    and $\#\mathrm{I}$ denote respectively the number of $\mathrm{X}$, $\mathrm{Y}$, $\mathrm{Z}$, and $\mathrm{I}$ in a Pauli string.
    For instance, $\#\mathrm{X}|_{V_1}$ indicates the number of $\mathrm{X}$ in $V_1$ for a vertex bipartition $V=V_1 {\uplus} V_2$
    with $\abs{V}=n\geq 2$.
    \label{tab:free-mixer-full-table}}
    \begin{tabular}{@{\hspace{1mm}}l@{\hspace{-14mm}}r@{\hspace{5mm}}l@{\hspace{-3mm}}r@{\hspace{3mm}}c@{\hspace{3mm}}l@{\extracolsep{1mm}}r@{\hspace{1mm}}}
    \\[-6mm]
    \\ \hline\hline
    \\[-4mm]
    & & & \multicolumn{3}{@{}c@{}}{Connected} & Connected \\[-1mm]
    & \textbf{Path graph} & \textbf{Cycle graph} & \multicolumn{3}{@{}c@{}}{\textbf{bipartite graph} (with $V=V_1 {\uplus} V_2$)} & \textbf{non-bipartite graph} \\[-1mm]
    & & & \multicolumn{3}{@{}c@{}}{$\neq$ cycle or path graph} & $\neq$ cycle graph
    \\[1mm] \hline
    \\[-3.5mm]
    \textbf{Pauli strings} &
    $\mathrm{I}\cddot\mathrm{I}\mathrm{X}\mathrm{I}\cddot\mathrm{I}$
    &
    $\mathrm{I}\cddot\mathrm{I}\mathrm{X}\mathrm{I}\cddot\mathrm{I}$,\, $\mathrm{X}\cddot\mathrm{X}\mathrm{I}\mathrm{X}\cddot\mathrm{X}$
    &
    \multicolumn{3}{@{}c@{}}{$\#\mathrm{Y} + \#\mathrm{Z}$ is even and}
    &
    $\#\mathrm{Y} + \#\mathrm{Z}$ is even and
    \\
    &
    $\mathrm{I}\cddot\mathrm{I}\yz\mathrm{X}\cddot\mathrm{X}\yz\mathrm{I}\cddot\mathrm{I}$
    &
    $\mathrm{I}\cddot\mathrm{I}\yz\mathrm{X}\cddot\mathrm{X}\yz\mathrm{I}\cddot\mathrm{I}$
    &
    \multicolumn{3}{@{}c@{}}{ $\#\mathrm{I}\neq n$ and $\#\mathrm{X}\neq n$ and}
    & 
    $\#\mathrm{I}\neq n$ and $\#\mathrm{X}\neq n$
    \\
    &
    &
    $\mathrm{X}\cddot\mathrm{X}\yz\mathrm{I}\cddot\mathrm{I}\yz\mathrm{X}\cddot\mathrm{X}$
    &
    \multicolumn{3}{@{}c@{}}{$\#\mathrm{X}+ \#\mathrm{Y}|_{V_1} + \#\mathrm{Z}|_{V_1}$ is odd}
    \\[1mm]
    & & & \textbf{even-even} & \textbf{even-odd} & \textbf{odd-odd} \\ \cline{4-6}
    \\[-4mm]
    \textbf{\DLA}   & $\so(2n)$ & $\so(2n){\oplus}\so(2n)$
    & $\so(2^{n-1}){\oplus}\so(2^{n-1})$ & $\su(2^{n-1})$ & $\usp(2^{n-1}){\oplus}\usp(2^{n-1})$ &
    $\su(2^{n-1}){\oplus}\su(2^{n-1})$\\[1mm]
    \textbf{Lie dimension} & $2n^2{-}n$ & $4n^2{-}2n$ & 
    $2^{2n-2}{-}2^{n-1}$ & $2^{2n-2}{-}1$ & $2^{2n-2}{+}2^{n-1}$ & $2^{2n-1}{-}2$
    \\[1mm]
    \hline\hline
    \\[-2mm]
    \end{tabular}
\end{table*}

\subsection{Preliminaries\label{app:free:preliminaries}}
In Appendices~\ref{appendix:free-mixer} and \ref{appendix:free-mixer:bases},
$G$ denotes a graph and $V$ its vertex set,
which is usually given by $\{1,\ldots,n\}$,
while $n:=\abs{V}$ and
$E \subset V {\times} V$  describes the edge set
(and similarly for graphs $\tilde{G}$ and $\bar{G}$).
Let $X, Y,  Z \in \C^{2\times 2}$ denote the Pauli operators
and $I\in \C^{2\times 2}$ the identity operator.
The corresponding operators acting on the $u$-th qubit
are $X_u, Y_u, Z_u, I_u \in  \C^{2^n \times 2^n}$, e.g.\
$X_u:= I^{\otimes u-1} \otimes X  \otimes I^{\otimes n-u}$.
We will usually represent a \DLA $\g \subseteq \uu(2^n)$ 
explicitly by $2^n \times 2^n$ skew-hermitian matrices
[and similarly for $\g \subseteq \su(2^n)$]. 
In this work, a Pauli string is given by a tensor-product
operator of the form
$P_j=\bigotimes_{u\in V} A_u$ with  $A_u
\in \{X, Y, Z, I\}$ and the corresponding Lie-algebra element is $i P_j$.
As in Section~\ref{SEC:FREE-ANSATZ}, the free-mixer \DLA
for a connected graph with $n\geq 2$ is given by
\begin{align}
\gfree &:= \lie{iX_u \text{ for } u\in V;\; iZ_uZ_v  \text{ for } \{u,v\} \in E} \label{freegenerators}\\
&\phantom{:}= \lie{iX_u \text{ for } u\in V;\hspace{-1mm} \sum_{\{u,v\} \in E} \hspace{-1mm} iZ_uZ_v} =: \tilde{\g}_{\rm{free}}. \label{nonfreegens}
\end{align}
The generators $iZ_uZ_v$
span $\sum_{\{u,v\} \in E} iZ_uZ_v$
and $\gfree \supseteq \tilde{\g}_{\rm{free}}$.
Given an edge $(a,b)\in E$,
$[iX_b,\allowbreak{}[iX_a,\allowbreak{}[iX_b,\allowbreak{}[iX_a,\allowbreak{}\sum_{\{u,v\}\in E}iZ_uZ_v]]]]=\allowbreak{}16 i Z_a Z_b$
and $\gfree = \tilde{\g}_{\rm{free}}$.
Having defined $\gfree$, we prove its following properties

\begin{lemma}\label{app:prop:center}
(a)~The set $\com(\gfree)$
of complex matrices
commuting with $\gfree$ is spanned by $I^{\otimes n}$ and $X^{\otimes n}$.
(b)~$iI^{\otimes n}, iX^{\otimes n} {\notin} \gfree$.
(c)~The center $\cent(\gfree):=\com(\gfree) \cap \gfree$
of $\gfree$ is trivial and $\gfree$ is semisimple.
\end{lemma}

\begin{proof}
Clearly, $I^{\otimes n}$ and $X^{\otimes n}$
commute with all generators of $\gfree$. All elements $S=\sum_j c_j P_j$
of  $\com(\gfree)$ with $0\neq c_j \in \C$ can be expanded into
Pauli strings
$P_j=\otimes_{u\in V} A_u \text{ with } A_u
\in \{X, Y, Z, I\}$. For the Pauli strings $P_k$ and $P_{j_1}\neq P_{j_2}$,
$[P_{k},P_{j_1}]$ and $[P_{k},P_{j_2}]$ are either linearly independent or at least one commutator is zero.
We obtain $[S,P_k]=0$ iff $[P_j,P_k]=0$ for all $P_j$,
hence $[S,\gfree]=0$ iff $[P_j, \gfree]=0$ for all $P_j$
as all generators of $\gfree$ in Eq.~\eqref{freegenerators} are of tensor-product form.
Thus we can restrict
elements of $\com(\gfree)$ to the form of $S=\bigotimes_{u\in V} A_u$.
However, $A_u = Y$ or $A_u = Z$ for any $u$ implies
$[S,iX_u]\neq 0$. Thus $S=\bigotimes_{u\in V} A_u$
with $A_u  \in \{X, I\}$. Now let $A_u=I$ and $A_v=X$ for two $u\neq v$.
As the graph is connected, there is a vertex path $u = u_1$, $u_2$, \ldots, $u_m = v$ using edges
$\{u_a, u_{a+1}\}$
and there exists an index $1\le \ell \le m{-}1$ such that $A_{u_\ell}=I$ and $A_{u_{\ell+1}}=X$. This implies
$[S,iZ_{u_\ell}Z_{u_{\ell+1}}]\neq 0$ which proves (a).
Recall that $\uu(2^n)$ and all its subalgebras $\g$ (such as $\gfree$) are compact \DLAs
and decompose
as $\g = [\g,\g] \oplus \cent(g)$ into their semisimple part $[\g,\g]$ and their center $\cent(g)$ \cite{BourbakiLie1989,Bourbaki2008b}.
Also, $iI^{\otimes n}$ and  $iX^{\otimes n}$ (and any real-linear combination thereof)
can only be in $\gfree$ if they are in its center $\cent(\gfree)=\com(\gfree) {\cap} \gfree$.
As all projections of generators in Eq.~\eqref{freegenerators} onto either $iI^{\otimes n}$ or $iX^{\otimes n}$
are zero, all generators are contained in the semisimple part $\mathfrak{s}:=[\gfree,\gfree]$.
As $[\mathfrak{s}, \mathfrak{s}] \subseteq \mathfrak{s}$ \cite{BourbakiLie1989},
we cannot generate any potential nonzero element of $\cent(\gfree)$.
Alternatively, the zero dimensionality of the center $\cent(\gfree)$ can also
be verified using Lemma~\ref{lem:std:free:center} in Appendix~\ref{app:center}
which characterizes the centers of subalgebras of $\gfree$.
\end{proof}

Lemma~\ref{app:prop:center} concludes that $\com(\gfree)$
is two dimensional and commutative.
Together with representation-theoretic arguments
(see, e.g., Theorem~1.5 in \cite{Ledermann}), this implies that the representation of $\gfree$
splits  exactly into two irreducible representations (or two irreducible blocks in a suitable basis).
Indeed, we can determine these two representations by noting that up to a constant phase factor,
$Y$ and $Z$ swap the Hadamard, or $X$-basis states $\ket{+}$ and $\ket{-}$,
while $I$ and $X$ preserve them. Thus $\gfree$ has invariant subspaces $\HC_{+}$
and $\HC_{-}$ which are the span of all Hadamard basis states $\ket{b_1}\otimes\cdots\otimes\ket{b_n}$ with $b_j \in \{+, -\}$
with respectively an even or odd  number of minus
signs.
Equivalently, $\HC_{+}$ and $\HC_{-}$ are respectively the $+1$ and $-1$ eigenspaces of $X^{\otimes n}$.
Also, $\HC_{+}$ and $\HC_{-}$ have the same dimension.
It follows that $\gfree$ is isomorphic to a subalgebra of
$\uu(2^{n-1})\oplus\uu(2^{n-1}) \iso
\su(2^{n-1})\oplus\su(2^{n-1})\oplus\uu(1)\oplus\uu(1)$. As the center of $\gfree$ is trivial [see Lemma~\ref{app:prop:center}(c)],
we obtain

\begin{lemma}\label{app:prop:structure}
The free-mixer \DLA $\gfree$ is isomorphic to a subalgebra of $\su(2^{n-1}){\oplus}\su(2^{n-1})$.
\end{lemma}

\begin{proof}
We project $\gfree$ onto its two irreducible components 
$\gfree^{\pm}:= P_{\pm} \gfree P_{\pm}$
where $P_{\pm}:= (I^{\otimes n} {\pm} X^{\otimes n})/2$.
We verify $P_{\pm} g_j P_{\pm} = P_{\pm} g_j$ 
for $g_j \in \gfree$
as $[X^{\otimes n},\gfree]=0$ and
then establish $[\gfree^{+}, \gfree^{-}]=\allowbreak{}[P_{+} g_1, P_{-} g_2] = \allowbreak{}P_{+} P_{-} [g_1, g_2] = 0$.
The generators from  Eq.~\eqref{freegenerators}
are projected to $i(X_u {\pm} X_u X^{\otimes n})$ and $i(Z_uZ_v {\mp} \otimes_{w\in V}  A_w)$ 
where $A_w = Y$ if $w \in \{u,v\}$ and $A_w = X$ otherwise.
And $[P_{\pm} g_1, P_{\pm} g_2] =  P_{\pm}^2 [g_1, g_2] = P_{\pm} [g_1,g_2]$
implies that
$\gfree^{\pm}$ both form a \DLA and
they are isomorphic
as every element $P_+ g P_+  \in \gfree^{+}$ for $g\in \gfree$ is mapped to
$P_- g P_- \in \gfree^{-}$.
For a center element $0 \neq P_+ c P_+ \in \gfree^{+}$ with $c \in \gfree$, $P_{-} c P_{-}$
is in the center of $\gfree^{-}$. This implies that $c$ is in $\cent(\gfree)$ which is impossible.
Consequently, both $\gfree^{\pm}$ have trivial centers (and are semisimple).
\end{proof}

To finish this section, we will present a series of maps that will be useful in the proofs of our main result.
To begin, consider the map $h(M):=\had^{\otimes n} M \had^{\otimes n}$
on $2^n \times 2^n$ complex matrices $M$
uses the Hadamard operator
$$\had := \tfrac{1}{\sqrt{2}}
\left(\begin{smallmatrix}
1 & 1\\
1 & -1
\end{smallmatrix}\right)\,,
$$
and
maps $iX_j$, $iY_j$, $iZ_j$ to 
$iZ_j$, $-iY_j$, $iX_j$. It is
a Lie-algebra automorphism of $\su(2^n)$.
The generators from Eq.~\eqref{freegenerators} are mapped 
by $h$
to
$iZ_u$ for $u\in V$ and $iX_uX_v$ for $\{u,v\}\in E$. 
We apply 
the map $\pi_n(g) :=\Pi_n g \Pi_n$ where 
\begin{align*}
& \Pi_n := I^{\otimes n}\, {\textstyle \prod_{k=2}^n \CNOT(k,1,n)}   \\   
& =\tfrac{1}{2} [I^{\otimes n} + Z_2\cdots Z_n + X_1 - X_1 Z_2\cdots Z_n] \; \text{ and }\\
& \CNOT(c,t,n) := \tfrac{1}{2} [I^{\otimes n} + Z_c + X_t - Z_c X_t].
\end{align*}
That is, $\CNOT(c,t,n)$ simply denotes a CNOT gate on an $n$-qubit system where $c$ denotes the control qubit and $t$ the target qubits. 
The combined transformation $\pi_n \circ h$ maps the symmetry $X^{\otimes n}$ to $Z_1$ 
and the projections $P_{+}$ and $P_{-}$ are 
mapped respectively to $I^{\otimes (n-1)} \oplus \zero^{\otimes (n-1)}$
and  $\zero^{\otimes (n-1)} \oplus I^{\otimes (n-1)}$, where $\zero$ is the $2\times 2$ zero matrix.
The Hamiltonians have been block diagonalized and $\pi_n \circ h$ maps
the basis elements $X_u$, $Z_u Z_v$, $Z_1 Z_v$, $X_1$ for $u,v \geq 2$ to
$Z_u$, $X_u X_v$, $X_v$, $Z^{\otimes n}$.
A final basis change $\tilde{h}(g):= \mathrm{H}_{n-1} g \mathrm{H}_{n-1}$
with  $\mathrm{H}_{n-1}:= I  \otimes \mathrm{H}^{\otimes n-1}$
leads to ($u,v \geq 2$)
\begin{equation}\label{generators_sym}
\text{(a) } X_u,\, \text{(b) } Z_u Z_v,\,
\text{(c) } Z_v,\, \text{(d) } Z_1X_2 \!\cdot\cdot X_n.
\end{equation}
For later reference, $Z_j$ is transformed by $\tilde{h} \circ \pi_n \circ h$ to
\begin{equation}\label{generators_sym_add}
\text{(e) } X_1 \text{ if } j=1 \text{ and } X_1Z_j \text{ otherwise},
\end{equation}
and the combined basis change leading to 
Eqs.~\eqref{generators_sym}-\eqref{generators_sym_add} is given for $n\geq 2$ by the unitary block matrix
$$\mathrm{H}_{n-1} \Pi_n \mathrm{H}_{n}
= \tfrac{1}{\sqrt{2}}
\left(
\begin{smallmatrix}
I^{\otimes (n{-}1)} & \phantom{-}X^{\otimes (n{-}1)}\\
I^{\otimes (n{-}1)} & -X^{\otimes (n{-}1)}\\
\end{smallmatrix}
\right).$$
In the following, we denote the corresponding map by
\begin{align}
&\Lambda(M):=
\tilde{h}(\pi_n(h(M)))
=(\mathrm{H}_{n-1} \Pi_n \mathrm{H}_{n})\, M 
(\mathrm{H}_{n} \Pi_n \mathrm{H}_{n-1}) \nonumber\\
&= \tfrac{1}{2}
\left(
\begin{smallmatrix}
I^{\otimes (n{-}1)} & \phantom{-}X^{\otimes (n{-}1)}\\
I^{\otimes (n{-}1)} & -X^{\otimes (n{-}1)}\\
\end{smallmatrix}
\right) M 
\left(
\begin{smallmatrix}
I^{\otimes (n{-}1)} & \phantom{-}I^{\otimes (n{-}1)}\\
X^{\otimes (n{-}1)} & -X^{\otimes (n{-}1)}\\
\end{smallmatrix}
\right). \label{basis_change}
\end{align}
For a Pauli string 
$P = \bigotimes_{u=1}^n A_u$ with $iP \in \gfree$,
Lemma~\ref{app:prop:center}(a) implies that 
the number of $A_u$ with $A_u \in \{Y,Z\}$ is even and 
$\Lambda(P) = \pm Q$ for a Pauli string $Q$, i.e.,
\begin{equation}\label{basis_change_Pauli_string}
\Lambda(P)
=
 \begin{cases}
\bigotimes_{u=2}^n A_u & \text{for $A_1 \in \{I, Z\}$}, \\
Z_1 \bigotimes_{u=2}^n (X A_u) & \text{for $A_1 = X$}, \\
i Z_1 \bigotimes_{u=2}^n (X A_u) & \text{for $A_1 = Y$}.
\end{cases}
\end{equation}
The Hamiltonians under (a) and (b)
recover the initial ones from Eq.~\eqref{freegenerators}
without the first vertex. Removing the first vertex will lend itself for an induction
step in $n$.

\subsection{Irreducibility and maximal subalgebras\label{app:free:irred}}

We recall two tools for characterizing
free-mixer \DLAs. The first one relies on the assumption
that a \DLA $\g \subseteq \su(d)$ acts irreducibly [as on the last $n{-}1$ qubits for (a)-(c) in Eq.~\eqref{generators_sym}].
In this case, $\g$ can be categorized depending whether or not it is conjugate to 
a subalgebra of either $\so(d)$ or $\usp(d)$. Based on
Sec.~3.11 of \cite{Samelson1999}, \cite{Obata1958}, and Sec.~7 of \cite{zeier2011symmetry},
we summarize this in the following proposition. 
\begin{proposition}\label{prop_obata}
Assume that a subalgebra $\g \subseteq \su(d)$ generated by a set of $iH_j$ 
is irreducibly embedded into $\su(d)$.
It is conjugate to a subalgebra of either (i) $\so(d)$ and (ii) $\usp(d)$
iff there exists a nonzero
$d {\times} d$ complex matrix $S$ such that $S H_j {+} H_j^{t} S = 0$ for all $H_j$.
Every matrix $S$ is either symmetric or skew-symmetric.
If $S\neq 0$ exists, it is
unique up to a scalar and one observes either (i) if $S$ is symmetric or (ii) if $S$ is skew-symmetric.
Also, $S\bar{S} = \pm \alpha \unity_d$  for $S$ symmetric ($+$) or skew-symmetric ($-$)
where $\bar{S}$ is the complex-conjugated $S$ and $0<\alpha \in \R$. 
\end{proposition}

\noindent We refer the reader to~\cite{zeier2011symmetry} for a proof of Proposition~\ref{prop_obata}. 

Our second tool 
often enables us to
simplify 
and streamline the identification of \DLAs
and it relies on maximality relations between
them (\cite{Dynkin57a,Dynkin57b,GOV94,AFG12}):

\begin{definition}\label{def_max}
A \emph{maximal subalgebra} $\gh$ of a \DLA $\g$ is a proper subalgebra
$\gh \subsetneq \g$ such that any subalgebra
$\gk$ of $\g$ containing $\gh$ ($\gh \subseteq \gk \subseteq \g$) 
is equal to $\gh$ or $\g$. In this case,
$\gh$ and any element $g\in \g$ with $g \notin \gh$ generate $\g$.
\end{definition}

\begin{table}
\caption{Example maximal subalgebras $\g^{+}$ of compact $\g$.
    \label{tab:maximal}}
    \begin{tabular}{@{\hspace{1mm}}l@{\hspace{-5mm}}r@{\hspace{1mm}}}
    \hline\hline
    \\[-4mm]
    $\g$ & $\g^{+}$
    \\[1mm] \hline
    \\[-3.5mm]
    $\su(2^n)$ & $\so(2^n);\; \usp(2^n), n{\geq}2;\; \su(2^{n-1}){\oplus}\su(2^{n-1}){\oplus}\uu(1)$
    \\
    $\so(2^n)$ &
    $\su(2^{n-1}){\oplus}\uu(1), n{\geq}3;\;
    \so(2^{n-1}){\oplus}\so(2^{n-1}), n{\geq}4$
    \\
    $\usp(2^n)$ &
    $\su(2^{n-1}){\oplus}\uu(1), n{\geq}2;\;
    \usp(2^{n-1}){\oplus}\usp(2^{n-1})$
    \\
    $\su(2^{n}){\oplus}\su(2^{n})$ & $\su(2^{n})$
    \\
    $\so(2m){\oplus}\so(2m)$ & $\so(2m), m{\geq}3$
    \\
    $\so(m)$ & $\so(m{-}1), m{\geq}5$
    \\
    $\usp(2m){\oplus}\usp(2m)$ & $\usp(2m)$
    \\[1mm]
    \hline\hline
    \end{tabular}
\end{table}

Fortunately, many relevant cases can also be more easily
explained (see \cite[Chap.~IX, \S 1, Ex.~7]{Bourbaki2008b} and
\cite{Borel1998,Helgason1978}) using
an \emph{involution} $s$ of a \DLA $\g$, i.e.\ an 
automorphism $s$ of $\g$ where $s^2$ is the identity map.
Given the eigenspaces
$\g^{\pm}$ of $s$ for the respective eigenvalues $\pm1$, 
$\g^{+}$ is a subalgebra of $\g$ ($[\g^{+},\g^{+}] \subseteq \g^{+}$),
$\g^{+}$ acts on $\g^{-}$ via the commutator with $[\g^{+},\g^{-}] \subseteq \g^{-}$,
and $[\g^{-},\g^{-}] \subseteq \g^{+}$. The commutator action of 
$\g^{+}$ on $\g^{-}$ is irreducible if and only if $\g^{+}$ is a maximal subalgebra of $\g$.
If in addition, $\g^{+}$ does not contain any nonzero ideal of $\g$, then the pair $(\g,s)$
is called \emph{irreducible}. For compact semisimple \DLAs, $(\g,s)$
is irreducible iff $\g$ is either simple or a sum of two simple ideals exchanged by $s$.
Hence, the maximality of $\g^{+}$ is often immediately implied and all possible cases
for $(\g,s)$ relate to particular symmetric spaces \cite{Borel1998,Helgason1978}:

\begin{proposition}\label{lem:max}
Table~\ref{tab:maximal} lists maximal subalgebras $\g^{+}$ of 
compact \DLAs $\g$
for (a) $\g$ simple and for
(b)~$\g = \gh {\oplus} \gh$ with $\gh$ simple,
$\g^{+} = \{(h,h) \text{ for } h \in \gh\} \iso \gh$, and
$\g^{-} = \{(h,-h) \text{ for } h \in \gh\}$.
\end{proposition}

\begin{proof}
For (a), we refer to the classification 
of symmetric spaces \cite{Helgason1978} 
and the maximality of $\g^{+}$ follows
for compact simple $\g$ via
\cite[Chap.~IX, \S 1, Ex.~7]{Bourbaki2008b} and
\cite{Borel1998,Helgason1978}.
We also point to the classification of maximal subalgebras \cite{Dynkin57a,Dynkin57b,GOV94,AFG12}.
For (b), we apply
the preceding analysis or 
Theorem~15.1 in \cite{Dynkin57b} on maximal subalgebras 
of semisimple $\g$.
\end{proof}

At this point it is worth noting that to ease the notation, unless otherwise stated, $\{A,\!B\}$ will denote the set containing $A$ and $B$ (and not their anticommutator).
We shortly discuss cases which will be used later:
\begin{example}\label{def_ga_gb}
The \DLAs
\begin{align*}
&\g_{a}=\spn_{\mbb{R}} \{i\, I \otimes \CC_{n-1}\} \iso \su(2^{n-1}) \text{ and } \\
&\g_{b}=\spn_{\mbb{R}} \{i\, \{I,\!Z\} \otimes \CC_{n-1} \} 
\iso \su(2^{n-1}) {\oplus} \su(2^{n-1}) \\
& \text{with } \CC_{n-1}:=\{I,X,Y,Z\}^{\otimes n-1} \setminus \{I^{\otimes n-1}\}\,
\end{align*}
are block diagonal with
equal ($iB_k {\oplus} B_k$) and  
independent ($iB_k {\oplus} C_\ell$)
blocks where $B_k,C_{\ell} \in \CC_{n-1}$. Also, 
$\g_a$ is maximal in $\g_b$. 
Adding $iZ_1$ to $\g_b$ generates
$\g_{c}\iso \su(2^{n-1}){\oplus}\allowbreak{}\su(2^{n-1}){\oplus}\allowbreak{}\uu(1)$
which is maximal in $\su(2^n)$. 
\end{example}

Clearly, $iZ_1X_2 \!\cdot\cdot X_n\notin\g_{a}$
but it is contained in $\g_{b}$. Proposition~\ref{lem:max} and Example~\ref{def_ga_gb} directly imply
the first statement in the following lemma:

\begin{lemma}\label{lem_ga_gb}
(a) The $n$-qubit \DLA $\g_a$ together with the element $iZ_1X_2 \!\cdot\cdot X_n$ 
generates the \DLA
$\g_{b}$. 
(b) We now assume that $n\geq 3$.
(b1) Adding $X_1Z_j$ with $j\geq 2$ 
to $\g_{b}$
generates also $X_1$.
(b2) Adding $X_1$
to $\g_{b}$
generates also $Z_1$.
(b3) Adding $X_1$ or $X_1Z_j$ with $j\geq 2$ 
to $\g_{b}$ generates $\su(2^n)$.
\end{lemma}

\begin{proof}
Statement~(a) follows as has been detailed before the lemma.
Note the commutator chains
\begin{align}
-\tfrac{i}{2}X_1 &{=}\,
[\tfrac{i}{2}Z_1Z_2Z_3,\![\tfrac{i}{2}Z_1Z_2,\!\tfrac{i}{2}X_1Z_3]], \label{eq_chain_one}\\
-\tfrac{i}{2}Z_1 &{=}\,
[\tfrac{i}{2}X_1,\![\tfrac{i}{2}Z_1X_2,\![\tfrac{i}{2}Z_1X_2X_3,\![
\tfrac{i}{2}Z_1X_3,\!\tfrac{i}{2}X_1]]]].
\label{eq_chain_two}
\end{align}
Equation~\eqref{eq_chain_one} implies (b1) in general and
(b2) follows from Eq.~\eqref{eq_chain_two}.
Statements (b1) and (b2) combined with Example~\ref{def_ga_gb}
show that $\g_c$ is generated. As $\g_c$ is maximal in $\su(2^n)$,
we obtain (b3) via Prop.~\ref{lem:max} as the additional element $X_1$ or $X_1Z_j$
is not contained in $\g_c$.
\end{proof}

\subsection{Path and cycle graphs\label{appendix:sub:pathcycle}}

\begin{figure*}
\begin{equation*}
N= -\tfrac{i}{2}
\left(\!\!
\begin{smallmatrix}
\zero^{\otimes 4} &
Z_1 & X_1 Z_2 & X_1 X_2 Z_3 & X_1 X_2 X_3 Z_4 &
Y_1 & X_1 Y_2 & X_1 X_2 Y_3 & X_1 X_2 X_3 Y_4 &
X^{\otimes 4}\\
-Z_1 & \zero^{\otimes 4} & 
-Y_1Z_2 & -Y_1 X_2 Z_3 & Y_1 X_2 X_3 Z_4 &
X_1 & -Y_1 Y_2 & -Y_1 X_2 Y_3 & -Y_1 X_2 X_3 Y_4 &
-Y_1 X_2 X_3 X_4\\
-X_1 Z_2 & Y_1 Z_2 & \zero^{\otimes 4} & -Y_2 Z_3 &-Y_2 X_3 Z_4 & 
-Z_1 Z_2 & X_2 & -Y_2 Y_3 & -Y_2 X_3 Y_4 & -Y_2 X_3 X_4\\
-X_1 X_2 Z_3 & Y_1 X_2 Z_3 & Y_2 Z_3 & \zero^{\otimes 4} & -Y_3 Z_4 &
-Z_1 X_2 Z_3 & -Z_2 Z_3 & X_3 & -Y_3 Y_4 & -Y_3 X_4\\
-X_1 X_2 X_3 Z_4 & Y_1 X_2 X_3 Z_4 & Y_2 X_3 Z_4 & Y_3 Z_4 &
\zero^{\otimes 4} & -Z_1 X_2 X_3 Z_4 & -Z_2 X_3 Z_4 & -Z_3 Z_4 & X_4 & -Y_4\\
-Y_1 & -X_1 & Z_1 Z_2 & Z_1 X_2 Z_3 & Z_1 X_2 X_3 Z_4 & \zero^{\otimes 4} &
Z_1 Y_2 & Z_2 X_2 Y_3 & Z_1 X_2 X_3 Y_4 & Z_1 X_2 X_3 X_4\\
-X_1 Y_2 & Y_1 Y_2 & -X_2 & Z_2 Z_3 & Z_2 X_3 Z_4 & - Z_1 Y_2 &
\zero^{\otimes 4} & Z_2 Y_3 & Z_2 X_3 Y_4 & Z_2 X_3 X_4\\
- X_1 X_2 Y_3 & Y_1 X_2 Y_3 & Y_2 Y_3 & -X_3 & Z_3 Z_4 &
-Z_1 X_2 Y_3 & -Z_2 Y_3 & \zero^{\otimes 4} & Z_3 Y_4 & Z_3 X_4\\
-X_1 X_2 X_3 Y_4 & Y_1 X_2 X_3 Y_4 & Y_2 X_3 Y_4 & Y_3 Y_4 & -X_4 &
-Z_1 X_2 X_3 Y_4 & -Z_2 X_3 Y_4 & -Z_3 Y_4 & \zero^{\otimes 4} & Z_4\\
-X^{\otimes 4} & Y_1 X_2 X_3 X_4 & Y_2 X_3 X_4 & Y_3 X_4 & Y_4 &
-Z_1 X_2 X_3 X_4 & -Z_2 X_3 X_4 & -Z_3 X_4 & - Z_4 & \zero^{\otimes 4} 
\end{smallmatrix}
\!\!\right)\! 
\end{equation*}
\caption{\textbf{Example for the matrix $N$ from Eq.~\eqref{eq_N}.} Case of $n=4$ which contains $16 {\times} 16$ matrices as its elements.
\label{fig:N}}
\end{figure*}

Our analysis of the free-mixer ansatz starts with
the path and the cycle graphs with $n$ vertices. An important aspect
relates to a Lie-algebra isomorphism from $\so(2n{+}2)$
to a subalgebra of $\su(2^n)$. This isomorphism is detailed
by adapting and slightly extending the
classical work of \cite{BrauerWeyl1935}
(see, e.g., \cite{Murnaghan1938,Kaufman1949,Boerner1969,Miller1972,SW86,FH91,Georgi1999,Zee2016})
on spinor representations of the orthogonal group.
To this end, we provide a matrix $N$ with $(2n{+}2) {\times} (2n{+}2)$ entries which are themselves 
elements from $\su(2^n)$ as well as a matrix $M$ with $(2n{+}2) {\times} (2n{+}2)$ entries from $\so(2n{+}2)$.
The intended Lie isomorphism is then defined
by mapping the entries $N_{jk}$ to the entries $M_{jk}$.
Let us introduce the two vectors
\begingroup
\allowdisplaybreaks
\begin{align}
\alpha := -\tfrac{i}{2} & [Z_1, X_1 Z_2, \ldots, X_1\!\cdot\cdot X_{n-1} Z_n,\nonumber \\
& \hspace{1.4mm} Y_1, X_1 Y_2, \ldots, X_1 \!\cdot\cdot X_{n-1} Y_n], \label{eq:alpha}\\
\tilde{\alpha} := -\tfrac{i}{2} & [-Y_1 X_2 \!\cdot\cdot X_n, \ldots,
-Y_{n-1} X_n, -Y_n, \nonumber \\
& \hspace{0.7mm} \phantom{-}Z_1 X_2 \!\cdot\cdot X_n, \ldots, Z_{n-1} X_n, Z_n], \label{eq:tilde:alpha}
\end{align}
\endgroup
where the entries $\alpha_j$ and $\tilde{\alpha}_j$ are $2^n{\times} 2^n$ matrices and $1\leq j \leq 2n$.
The $(2n{+}2){\times} (2n{+}2)$ matrix $N$ has entries
\begin{align}\label{eq_N}
N_{jk} = - N_{kj} :=
 \begin{cases}
\zero^{\otimes n} & \text{if $j=k$}, \\
-\tfrac{i}{2} X^{\otimes n} & \text{if $j=0$, $k=2n{+}1$},\\
[\alpha_k,\alpha_j] & \text{if $1\leq j,k \leq 2n$},\\
\alpha_{k} & \text{if $0=j<  k \leq 2n$},\\
\tilde{\alpha}_{j} & \text{if $1\leq j < k =2n{+}1$}
\end{cases}
\end{align}
which are $2^n{\times} 2^n$ matrices for $j,k \in \{0,\ldots,2n{+}1\}$.
For $n=4$, we obtain the $10{\times} 10$ matrix 
in Figure~\ref{fig:N} which contains $16 {\times} 16$ matrices as its elements.

The matrix $M$ is defined by its $(2n{+}2){\times} (2n{+}2)$ entries given by the
$(2n{+}2){\times} (2n{+}2)$ matrices
$M_{jk}:=e_{jk}-e_{kj}$
where the 
matrices $e_{jk}$ have entries
$(e_{jk})_{ab} := \delta_{ja} \delta_{kb}$.
We map the $(2n{+}2) {\times} (2n{+}2)$ matrices $M_{jk}$
to the $2^n {\times} 2^n$ matrices $N_{jk}$ for $0 \leq j < k \leq 2n{+}1$ which induces
the desired Lie-algebra isomorphism as one can easily check.
This isomorphism is applied to characterize relevant subalgebras of
$\su(2^n)$ by limiting us to suitable submatrices of
the matrix $N$. In particular, we will prove below the following result. 
\begin{proposition}\label{prop:spinor}
Let us consider the vector-space bases
\begingroup
\allowdisplaybreaks
\begin{align*}
\mathcal{B}_a & := \{iX^{\otimes n}\},\quad \mathcal{B}_b :=  \{iX_j \text{ for } 1{\leq} j {\leq} n\},\\
\mathcal{B}_c & := \{i \{Y_j,\!Z_j\} X_{j+1} \!\cdots\! X_{k-1}\{Y_k,\!Z_k\} \text{ for } 1 {\leq} j {<} k {\leq} n \}, \\
\mathcal{B}_d & := \{i X_1 \!\cdots\! X_{k-1} \{Y_k,\!Z_k\} \text{ for } 1 {\leq} k {\leq} n\},\\
\mathcal{B}_e & := \{i \{Y_j,\!Z_j\} X_{j+1} \!\cdots\! X_n \text{ for } 1 {\leq} j {\leq} n \}.
\end{align*}
\endgroup
Using the notation $\BC_{j\ldots k}=\BC_j\cup \dots\cup \BC_k$, we have
\begingroup
\allowdisplaybreaks
\begin{align*}
&\spn_{\mbb{R}}\mathcal{B}_{bc} =:  \gh_1  \iso\so(2n),\;
\spn_{\mbb{R}} \mathcal{B}_{bcd} =: \gh_2 \iso\so(2n{+}1),\\
& \spn_{\mbb{R}}\mathcal{B}_{bce} =: \gh_3 \iso \so(2n{+}1),\\
& \spn_{\mbb{R}}\mathcal{B}_{abcde} =: \gh_4 \iso \so(2n{+}2).
\end{align*}
\endgroup
\end{proposition}

Note that $\gh_2,\gh_3\supset \gh_1$ and $\gh_4 \supset \gh_3,\gh_2,\gh_1$. 
We present the first result for a free-mixer \DLA: 
\begin{theorem}[Path graphs]\label{app:thm:path}
The free-mixer \DLA for 
a path graph
with $n {\geq} 2$ and edges $\{v,v{+}1\}$ for $1 {\leq} v {<} n$ is
$\gfree=\gh_1 \iso \so(2n)$ where
$\dim(\gfree)= 2n^2{-}n$.
\end{theorem}
\begin{proof}
The generators are $iX_u$ for $1{\leq} u{\leq} n$
and $iZ_{v}Z_{{v+1}}$ for $v{<}n$, i.e., 
$\gfree\subseteq \gh_1$.
Clearly, $iY_vY_{v+1},\allowbreak{}iY_vZ_{v+1},\allowbreak{}iZ_{v}Y_{v+1} \in \gfree$ for $v{<}n$.
We apply $[iZ_1Z_2,iY_2Z_3] = 2 i Z_1X_2Z_3$ and similar relations to
obtain all of $\mathcal{B}_{bc}$. 
\end{proof}

\begin{figure}[b]
\includegraphics{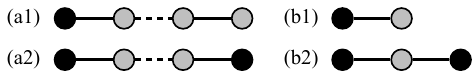}
\caption{\textbf{Path graphs with black end vertices.}
Generators are $iX_u$ for all $n$ vertices $u$,
$iZ_vZ_{v+1}$ for $1{\leq} v {<}n$, and $iZ_w$ for 
black vertices $w$.
(a) General form with \DLAs (a1) $\so(2n{+}1)$ and (a2) $\so(2n{+}2)$.
(b) Low-dimensional cases with (b1) $\so(2{\times}2{+}1) \iso \usp(2^2)$
and
(b2) $\so(2{\times}3{+}2)=\so(2^3)$.
\label{fig:sporadic}}
\end{figure}

Theorem~\ref{app:thm:path} leads together with
Propositions~\ref{lem:max} and \ref{prop:spinor} to
\begin{corollary}(Fig.~\ref{fig:sporadic})\label{coro:spinor}
For a path graph with $n{\geq} 2$,
the \DLAs $\gdot$ and $\gdotdot$ are generated by
$\mathcal{G}:=\{iX_u \text{ for } 1{\leq} u {\leq} n,\allowbreak{} iZ_vZ_{v+1} \allowbreak{} \text{ for }  1 {\leq} v {<} n\}$
and $iZ_1$ and by 
$\mathcal{G}$, $iZ_1$, and $iZ_n$, respectively.
(i)~They are irreducibly embedded into $\su(2^n)$.
(ii)~$\gdot = \gh_2\iso \so(2n{+}1)$;
$\so(5) \iso \usp(4)$.
(iii)~$\gdotdot =\gh_4 \iso \so(2n{+}2)$;
$\so(6) \iso \su(4)$.
\end{corollary}

\begin{proof}
As $[iZ_1,X^{\otimes n}]\neq 0$, 
Lemma~\ref{app:prop:center}(a) implies that $\com(\gdot)$
and $\com(\gdotdot)$  have dimension one which proves (i).
The case $n=2$ is verified directly.
Theorem~\ref{app:thm:path} implies that
$\mathcal{G}$ generates a \DLA that is isomorphic to $\so(2n)$ and
spanned by $\mathcal{B}_{bc}$.
Proposition~\ref{prop:spinor} shows that $\gh_1 \subsetneq\gdot\subseteq \gh_2$.
The maximality of $\gh_1$ in $\gh_2$ from Prop.~\ref{lem:max} implies (ii). 
The case (iii) is similar.
\end{proof}

The \DLA $\gfree\iso \so(2n)$ in Thm.~\ref{app:thm:path}
splits into two irreducible blocks of dimension $2^{n-1}$ [see 
Lemma~\ref{app:prop:center}(a)].
The basis change leading to Eq.~\eqref{generators_sym} is applied
to the generators of $\gfree$
and we
get (with $2 {\leq} u {\leq} n$, $2 {\leq} v {<} n$)
\begin{equation}\label{generators_sym_path}
\text{(a') } iX_u, \text{(b') } iZ_v Z_{v+1},
\text{(c') } iZ_2, \text{(d') } iZ_1 X_2 \!\cdot\cdot X_n.
\end{equation}
Using Corollary~\ref{coro:spinor}(ii),
(a')-(c') generate block-diagonal $L_k {\oplus} L_k$ where 
the matrices $L_k$ span $\gh_2 {\iso} \so(2n{-}1)$. The generator in (d') is equal to
$iX^{\otimes (n-1)} {\oplus} [-iX^{\otimes (n-1)}]$
and a $\so(2n)$ is generated in each block
using the analysis from Eq.~\eqref{eq_N}. But the \DLA is not
spanned by block-diagonal $\tilde{L}_k {\oplus} \tilde{L}_k$ as the blocks are partially
intertwined.

\begin{theorem}[Cycle graphs]\label{app:thm:cycle}
The free-mixer \DLA of a cycle graph
with edges $\{1,n\}$, $\{v,v{+}1\}$ for $1 {\leq} v {<} n {\geq} 3$ is
$\gfree \iso \so(2n){\oplus}\so(2n)$ where
$\dim(\gfree)= 4n^2{-}2n$.
\end{theorem}
\begin{proof}
The generators from (a')-(c') in Eq.~\eqref{generators_sym_path}
are extended by
$iZ_n$ (which is $iZ_1Z_n$ in the original 
basis) and get basis matrices $K_j {\oplus} K_j$ where the
$K_j$ span $\so(2n)$ [see Cor.~\ref{coro:spinor}(iii)]. With
Prop.~\ref{lem:max}, the span $\so(2n)$ of $K_j {\oplus} K_j$
is maximal in the span $\so(2n){\oplus} \so(2n)$ of $K_j {\oplus} K_j$ and $K_j {\oplus} [-K_j]$.
Including (d') completes the proof.
\end{proof}

Note that $\gfree \iso \so(2n){\oplus}\so(2n)$ in Thm.~\ref{app:thm:cycle} is spanned by $\mathcal{B}_{bc}$ and
$\tilde{\mathcal{B}}_{bc}:=\tilde{\mathcal{B}}_b {\cup} \tilde{\mathcal{B}}_c$
where $\tilde{\mathcal{B}}_r := \{A X^{\otimes n} \text{ for } A\in B_r\}$.
This is clear from the form of $\gfree$ in the basis of Eq.~\eqref{generators_sym_path} as detailed
in the proof of Thm.~\ref{app:thm:cycle} (here we recall that $X^{\otimes n}$ corresponds to $Z_1$ in the basis of
Eq.~\eqref{generators_sym_path}).

\subsection{Explicit representations: path and cycle graphs \label{appendix:explicit:reps:path:cycle}}

We complement the results in Section~\ref{appendix:sub:pathcycle}
by determining the explicit
form of the corresponding representations. Our presentation influenced by
the related discussions in \cite{BrauerWeyl1935}, \cite{Boerner1969} (see Chapter VIII and
particularly Section~\S4) and \cite{Miller1972} (see Sections~9.5 and 9.6).
To enable a more concise presentation, this subsection 
assumes some familiarity with the highest weight
theory for representations of (complex) \DLAs
\cite{hall2015,FH91,Bourbaki2008a,Bourbaki2008b}.

Recall from Theorem~\ref{app:thm:path} that the free-mixer \DLA
$\gfree$
for the path graph is spanned by 
\begin{align*}
&iX_j \text{ for } 2{\leq} j {\leq} n,\, iX_1,\\
&i \{Y_j,\!Z_j\} X_{j+1} \!\cdots\! X_{k-1}\{Y_k,\!Z_k\}
\text{ for } 2 {\leq} j {<} k {\leq} n, \text{ and }\\
&i Z_1 X_{2} \!\cdots\! X_{k-1} \{Y_k,\!Z_k\},\,
iY_1 X_{2} \!\cdots\! X_{k-1} Z_k,\\
&iY_1 X_{2} \!\cdots\! X_{k-1} Y_k \text{ for } 2 {\leq} k {\leq} n.
\end{align*}
Applying the basis change from Eqs.~\eqref{basis_change}-\eqref{basis_change_Pauli_string} leads to
\begin{align}
&iX_j \text{ for } 2{\leq} j {\leq} n,\, iZ_1 X_2 \!\cdots\!X_n,\nonumber \\
&i \{Y_j,\!Z_j\} X_{j+1} \!\cdots\! X_{k-1}\{Y_k,\!Z_k\}
\text{ for } 2 {\leq} j {<} k {\leq} n, \text{ and } \nonumber \\
&i X_{2} \!\cdots\! X_{k-1}\{Y_k,\!Z_k\},\,
iZ_1 Y_k  X_{k+1} \!\cdots\! X_{n}, \nonumber \\
&(-1)\times iZ_1 Z_k X_{k+1} \!\cdots\! X_{n} \text{ for } 2 {\leq} k {\leq} n.
\label{eq:path:basis}
\end{align}
Clearly, $iX_j$ for $2{\leq} j {\leq} n$ together with $iZ_1 X_2 \!\cdots\!X_n$
span a maximal abelian subalgebra $\mathfrak{t}$ of $\gfree$. The
complexification \cite{hall2015} of $\gfree$ 
is denoted by $\g_{\C}:=\gfree \otimes \C$.
We now assume $n\geq 4$, while all cases below with $n<4$
can be directly verified.
We choose a particularly suitable basis
\begin{subequations}
\label{eq:chevalley:path}
\begin{align}
&\tfrac{1}{2}(- X_2 {+} Z_1 X_2 \!\cdots\!X_n),\,
\tfrac{1}{2}(X_3 {-} Z_1 X_2 \!\cdots\!X_n),\\
&\tfrac{1}{2}(X_4 {-} X_3), \ldots,\,
\tfrac{1}{2}(X_{n-1} {-} X_{n-2}), \label{potentially_missing}\\
&\tfrac{1}{2}(-X_{n} {-} X_{n-1}),\,
\tfrac{1}{2}(X_{n} {-} X_{n-1})\,,
\end{align}
\end{subequations}
in the Cartan subalgebra \cite{hall2015} $i\mathfrak{t}$ of $\g_{\C}$,
where the parts in Eq.~\eqref{potentially_missing} are missing for $n=4$.
The basis elements are denoted by $H_{\alpha}$ 
(as in \cite{hall2015,Bourbaki2008b})
and they are ordered according to the simple roots $\alpha=\alpha_j$ 
of $\g_{\C} \iso \so(2n,\C)$ with $1 \leq j \leq n$ \cite{hall2015}.
Restricting to the irreducible $2^{n-1}\times 2^{n-1}$ block in the upper-left corner,
the $H_{\alpha}$
are represented as matrices $H_{\alpha}^{+}$ given by
\begin{subequations}
\label{eq:chevalley:path:upper}
\begin{align}
&\tfrac{1}{2}(- X_1 {+} X_1 \!\cdots\!X_{n-1}),\,
\tfrac{1}{2}(X_2 {-} X_1 \!\cdots\!X_{n-1}),\\
&\tfrac{1}{2}(X_3 {-} X_2), \ldots,\,
\tfrac{1}{2}(X_{n-2} {-} X_{n-3}),\\
&\tfrac{1}{2}(-X_{n-1} {-} X_{n-2}),\,
\tfrac{1}{2}(X_{n-1} {-} X_{n-2}).
\end{align}
\end{subequations}
One can identify
a particular one-dimensional eigenspace among
the common eigenspaces of the $H_{\alpha}^{+}$ which is known
as the highest weight space.
In our case, it is spanned by the highest weight vector \cite{hall2015}
(for $n\geq 4$)
\begin{equation}\label{eq:highest:weight:vector}
v^{+} :=
 \begin{cases}
\left(
\begin{smallmatrix}
+1\\
-1
\end{smallmatrix}\right)^{\otimes (n-1)}
 & \text{if $n$ is even}, \\
\left(
\begin{smallmatrix}
+1\\
-1
\end{smallmatrix}\right)^{\otimes (n-2)}
\otimes 
\left(
\begin{smallmatrix}
1\\
1
\end{smallmatrix}\right)
 & \text{if $n$ is odd}.
\end{cases}
\end{equation}
The eigenvalues $\omega_j^{+}$ for
$H_{\alpha_j}^{+} v^{+} = \omega_j^{+} v^{+}$ are collected in the highest weight
$\omega^{+} := (\omega_1^{+}\!,\ldots,\omega_n^{+})$
\cite{hall2015} with
\begin{equation}\label{eq:highest:weight}
\omega^{+} =
 \begin{cases}
(0,\ldots,0,1,0)
& \text{for even $n$}, \\
(0,\ldots,0,0,1)
& \text{for odd $n$}.
\end{cases}
\end{equation}
We identify the irreducible representation in the upper-left block as
one of the spinor representations of $\so(2n)$ 
[or equivalently for $\so(2n,\C)$] which we denote by $\eta_+$.
Similarly, the irreducible representation in
the lower-right block is the other spinor representation $\eta_-$ with
\begin{equation}\label{eq:highest:weight:b}
\omega^{-} =
 \begin{cases}
(0,\ldots,0,0,1)
& \text{for even $n$}, \\
(0,\ldots,0,1,0)
& \text{for odd $n$}.
\end{cases}
\end{equation}
Note that $\eta_{\pm}$ is conjugate
to the dual of $\eta_{\mp}$ for odd $n$ (i.e.\
$\eta_{\pm}  \simeq \bar{\eta}_{\mp}$); the
$\eta_{\pm}$ are self-dual
for even $n$ (i.e., $\eta_{\pm}  \simeq \bar{\eta}_{\pm}$).

We continue with the free-mixer \DLA $\gfree$ for the cycle graph.
Theorem~\ref{app:thm:cycle} implies that $\gfree\iso \g_{+} \oplus \g_{-}$
with $\g_{+} \iso \g_{-} \iso \so(2n)$. Using the same frame as in
Eq.~\ref{eq:path:basis}, a basis of $\g_{\pm}$ is given by
\begin{align*}
&\tfrac{i}{2}(I{\pm}Z) {\otimes} b_j \text{ with }
b_j \in [ X_j \text{ for } 1\leq j \leq n{-}1, X_1\cdots X_{n-1},\\
&\{Y_j,\!Z_j\} X_{j+1} \!\cdots\! X_{n},
X_{1} \!\cdots\! X_{j-1}\{Y_j,\!Z_j\} \text{ for } 1\leq j \leq n{-}1,\\
&\{Y_j,\!Z_j\} X_{j+1} \!\cdots\! X_{k-1}\{Y_k,\!Z_k\}
\text{ for } 1 {\leq} j {<} k {\leq} n{-}1].
\end{align*}
Thus a maximal abelian subalgebra is spanned by
$i (X_j \pm Z_1 X_j)/2$ for $2\leq j \leq n$ and
$i (X_2\cdots X_n \pm Z_1 X_2\cdots X_n)/2$. 
In the Cartan subalgebra of the complexification
of $\g_{+}$, the basis elements corresponding 
to Eq.~\eqref{eq:chevalley:path} are denoted
by $\tilde{H}_{\alpha}$ and are given by
\begin{subequations}
\label{eq:chevalley:cycle}
\begin{align}
&\tfrac{1}{4}(I{+}Z) {\otimes} \tilde{b}_j \text{ with } \tilde{b}_j \in [
{-} X_1 {+} X_1 \!\cdots\!X_{n-1},\\
&X_2 {-} X_1 \!\cdots\!X_{n-1},\,
X_3 {-} X_2, \ldots,\, X_{n-2} {-} X_{n-3},\\
&{-}X_{n-1} {-} X_{n-2},\, X_{n-1} {-} X_{n-2}].
\end{align}
\end{subequations}
But the matrices $H_{\alpha}^{+}$
from Eq.~\eqref{eq:chevalley:path:upper}
are again recovered
when the basis elements $\tilde{H}_{\alpha}$ are restricted
to the irreducible $2^{n-1}\times 2^{n-1}$ block in the upper-left corner.
Also the highest weight space is spanned by the highest weight
vector $v^{+}$ from Eq.~\eqref{eq:highest:weight:vector}. Consequently,
the irreducible representation of $\g_{+}$ in the upper-left block is given
by $\eta_{+}$. Note that $\g_{+}$ acts on the lower-right block via the
trivial representation $\epsilon$.
Similarly, the
irreducible representation of $\g_{-}$ in the lower-right block is also given
by $\eta_{+}$; $\g_{-}$ acts on the upper-left block via $\epsilon$.
Putting together the previous, we have thus shown that the following result  holds
\begin{proposition}[Free-mixer representations for path and cycle graphs]\label{prop:reps:path:cycle}
The free-mixer \DLAs $\gfree$ for the (i) path and (ii) cycle graph
are embedded 
into $\su(2^n)$ via the respective representations
(i)~$\eta_+ {\oplus} \eta_-$ and
(ii)~$[\eta_+ {\otimes} \epsilon] \oplus [\epsilon {\otimes} \eta_+]$
(up to automorphisms of $\gfree$).
\end{proposition}

\subsection{Bipartite graphs\label{appendix:sub:bipartite}}

We continue with connected
bipartite
graphs $G$ different from path and cycle graphs.
Their free-mixer \DLAs $\gfree$ are determined
in Theorem~\ref{thm:free-mixer-bipartite} below.
To this end, $\gfree$
is augmented with generators $iZ_w$ 
for vertices 
$w\in W$ 
from a nonempty vertex subset $W$
such that two of its vertices are never connected by an edge (vide infra).
The resulting \DLA $\bipartiteW$ is,
surprisingly, isomorphic to
either $\so(2^n)$ or $\usp(2^n)$ 
(see Proposition~\ref{prop-bipartite-add-Z}), 
except for most path graphs from Corollary~\ref{coro:spinor}.
This enables us to prove Theorem~\ref{thm:free-mixer-bipartite}.

Working towards this goal, we first establish some basic notation
to streamline our discussion. As before, $G$ denotes a graph
with $n:=\abs{V}$ vertices $V$ and edges $E$.
If $G$ is connected and bipartite,
its vertex bipartition $V=V_1 {\uplus} V_2$ is unique
and every edge $\{u,v\}$ connects $u\in V_1$ and $v \in V_2$
(or vice versa). Here, $V_1 {\uplus} V_2$ denotes the disjoint union
such that $V = V_1 {\cup} V_2$ and $V_1 {\cap} V_2 = \emptyset$.
The parts $V_k$ and $V_{\bar{k}}$ are identified by 
$k,\bar{k}\in \{1,2\}$ with $k\neq \bar{k}$. Let 
$\emptyset \neq W\subseteq V_k$ denote a nonempty subset of $V_k$.
Similar notation will be used for further graphs $\tilde{G}$, $\bar{G}$.
We formally introduce the \DLA $\bipartiteW$:

\begin{definition}\label{def:algebra:W}
For a connected bipartite graph $G$,
the \DLA $\bipartiteW$ is generated by
$iX_u$ for $u\in V$, $iZ_uZ_v$ for $\{u,v\} \in E$,
and $iZ_w$ for $w\in W$ with $\emptyset \neq W\subseteq V_k$.
\end{definition}

Our objective is to determine $\bipartiteW$. But as this will require 
multiple steps, we introduce the target \DLA $\bipartite$
with the aim of proving that $\bipartiteW=\bipartite$ except for
most path graphs from Corollary~\ref{coro:spinor}:

\begin{definition}\label{def:so:sp}
For a connected bipartite graph $G$,
$\bipartite$ is 
generated by
$i\bigotimes_{v\in V} A_v$ with $A_v \in \{I,X,Y,Z\}$ such that the number of $A_v=X$ for
$v \in V$ has the opposite parity of  the number of $A_w\in \{Y,Z\}$ for $w \in V_k$.
\end{definition}

We can now separately
establish properties of $\bipartite$ and we initially only show that $\bipartiteW \subseteq \bipartite$:

\begin{lemma}[Properties of $\bipartite$]\label{prop:so:sp}
(i)~$\bipartite$ is spanned by its generators.
(ii)~$\bipartite \iso \so(2^n)$
iff $\abs{V_k}$ is even.
(iii)~$\bipartite\iso \usp(2^n)$
iff $\abs{V_k}$ is odd. 
(iv)~$iZ_v \in  \bipartite$ iff $v \in V_k$.
(v)~$iX^{\otimes n} \in \bipartite$
iff $n$ is odd.
(vi)~$\bipartiteW \subseteq \bipartite$.
\end{lemma}

\begin{proof}
The statements $iZ_v \in  \bipartite$ if $v \in V_k$ and
$iX^{\otimes n} \in \bipartite$
if $n$ is odd follow from Definitions~\ref{def:algebra:W}-\ref{def:so:sp}.
As $[iZ_v,X^{\otimes n}]\neq 0$ for $v \in V_k$,
Lemma~\ref{app:prop:center}(a) shows that $\com(\bipartite)$
is one dimensional. 
Following Prop.~\ref{prop_obata}, 
we choose $S=\bigotimes_{v\in V_1} Z_v \bigotimes_{v'\in V_2} Y_{v'}$ such that $S H_j {+} H_j^t S = 0$ for every generator $iH_j$ of $\bipartite$.
$S\bar{S}=\pm \id_{2}\tn$ for $\abs{V_k}$ even ($+$) or odd ($-$). 
Proposition~\ref{prop_obata} proves the containment
$\bipartite \subseteq \so(2^n)$ or $\bipartite \subseteq \usp(2^n)$
depending on the parity of $\abs{V_k}$.
Counting the generators proves (i)-(iii) which implies (iv)-(v).
The generators of $\bipartiteW$ are contained in $\bipartite$ due to
Definitions~\ref{def:algebra:W}-\ref{def:so:sp}, hence
(vi).
\end{proof}

The \DLAs 
$\gdot$ or $\gdotdot$ from Corollary~\ref{coro:spinor}
are particular examples for $\bipartiteW$ which are contained in either 
$\so(2^n)$ or $\usp(2^n)$, except if $n$ is even for which
$\gdotdot\not\subseteq \so(2^n)$ and $\gdotdot\not\subseteq \usp(2^n)$ hold
as $W \subseteq V_k$ is not satisfied.
In particular, we have $\gdot\subseteq\usp(2^n)$ for $n \bmod 4 \in \{1,2\}$
and $\gdot\subseteq\so(2^n)$ for $n \bmod 4 \in \{3,0\}$
as well as $\gdotdot\subseteq\usp(2^n)$ for $n \bmod 4 = 1$
and $\gdotdot\subseteq\so(2^n)$ for $n \bmod 4 = 3$.
Clearly, these examples will not observe $\bipartiteW = \bipartite$,
except for the
two low-dimensional cases in Fig.~\ref{fig:sporadic}(b).

Assuming that the property $\bipartiteW = \bipartite$ holds for some specific connected bipartite graph $G$, we ask the following question: \textit{
can we extend this property with the same $W$
to a larger connected bipartite graphs $\tilde{G}$ that contains $G$ as a subgraph?} We provide
an induction argument for an induction step from $n$ to $n{+}1$
assuming that $\tilde{V}{-}1 =V=n\geq 4$:

\begin{lemma}[bipartite induction]\label{lem:induction}
Given a connected bipartite graph $G$,
we add the vertex $0$ and the edge $\{0,j\}$ with $j\in V$ and obtain the connected graph $\tilde{G}$
with
bipartition 
$\tilde{V}=\allowbreak{}\tilde{V}_1 {\uplus} \allowbreak{}\tilde{V}_2=\allowbreak{}V{\cup}\{0\}$.
This adds
$iX_0$ and $iZ_0Z_j$ and generates $\bipartiteWW {\subseteq} \bipartiteH$.
Let $n\geq 4$ and $\bipartiteW = \bipartite$.
(i)~Either (a)~$j \in \tilde{V}_k = V_k$,
$\tilde{V}_{\bar{k}} = V_{\bar{k}} {\cup} \{0\}$,
$Z_j\in \bipartiteW$, and $Z_0\not\in \bipartiteWW$
or (b)~$\tilde{V}_k = V_k {\cup} \{0\}$,
$j \in \tilde{V}_{\bar{k}} = V_{\bar{k}}$, and $Z_j\not\in \bipartiteW$ holds.
(ii)~We obtain that $\bipartiteWW = \bipartiteH$.
\end{lemma}

\begin{proof}
To prove (i), note that the edge $(0,j)$ needs to connect $\tilde{V}_{1}$ and $\tilde{V}_2$ 
as $\tilde{G}$
is bipartite. Thus $0$ and $j$ cannot be both in $\tilde{V}_{1}$ or $\tilde{V}_{2}$.
We apply Lemma~\ref{prop:so:sp}(iv) to decide whether $Z_j\in \bipartite = \bipartiteW$ or not.
In the case (a), $Z_0\in \bipartiteWW$
would imply $\bipartiteW \not\subset \bipartite$ which is impossible.
(ii)~$\bipartite$ is block-diagonally spanned by
$i I {\otimes} (\bigotimes_{v\in V} A_v)$ with $A_v \in \{I,X,Y,Z\}$ where the number of $A_v=X$ with 
$v \in V$ has the opposite parity of  the number of $A_w\in \{Y,Z\}$ with $w \in V_k$.
In case (a), $iZ_j\in \bipartiteW$ and adding $iZ_0Z_j$
results via Prop.~\ref{lem:max} in $\bipartite {\oplus} \bipartite$ which is maximal in $\so(2^{n+1})$
if $\bipartite \iso \so(2^{n})$ or in $\usp(2^{n+1})$ if $\bipartite \iso \usp(2^{n})$
[see Prop.~\ref{lem:max}].
Adding $iX_0$, we obtain that $\bipartiteWW = \bipartiteH$. In case (b),
$iZ_j\not\in \bipartiteW$ and adding $iZ_0Z_j$ results in $\su(2^{n})$
via Prop.~\ref{lem:max}. Its basis is spanned by all
$i I {\otimes} (\bigotimes_{v\in V} A_v)$ and $i Z {\otimes} B_j$ for 
Pauli strings $iB_j \not\in \bipartite$.
Assuming that $iZ_0 \in \bipartiteWW$, we generate $\su(2^{n}){\oplus}\uu(1)$
which is maximal in $\so(2^n)$ and $\usp(2^n)$ [see Prop.~\ref{lem:max}].
Adding $iX_0$ implies $\bipartiteWW = \bipartiteH$ where either
$\bipartite \iso \so(2^n)$ and $\bipartiteH \iso \usp(2^{n+1})$ or
$\bipartite \iso \usp(2^n)$ and $\bipartiteH \iso \so(2^{n+1})$.
Deciding if $iZ_0 \in \bipartiteWW$ for (b),
one of the two connected, four-vertex subgraphs of $\tilde{G}$ with the vertex $q=0$
in Fig.~\ref{fig:ygraph}(a) has to be a possibility and the generation of $iZ_0$
has been checked.
\end{proof}

To apply the induction step, we need to account for all base cases associated to connected bipartite graphs.
As $\bipartiteW \subseteq \bipartite$ [see Lemma~\ref{prop:so:sp}(vi)], $\bipartiteW = \bipartite$ holds for all $W\subseteq V_k$
if it is valid for $\abs{W}=1$.
We establish $\bipartiteW = \bipartite$ for path graphs with  $iZ_w$ where at least one vertex
$w\in W$ has degree two, i.e., is not at one end
[see Fig.~\ref{fig:ygraph}(b)]:

\begin{lemma}[Figure~\ref{fig:ygraph}(b)]\label{path-one}
For a path graph with
edges $E$ 
and bipartition 
$V=\allowbreak{}
V_1 {\uplus} V_2$ where
$V_1:=\{v {\in} V\allowbreak{} \text{ with}\allowbreak{} v \text{ odd}\}$
and $V_2:=\{v {\in} V\allowbreak{} \text{ with } \allowbreak{} v \text{ even}\}$,
$\kdot$ is generated by
$iX_u$ for $u\in V$, $iZ_uZ_v$ for $\{u,v\} \in E$,
$iZ_w$ for $w \in W\subseteq V_k$.
For cases not in Fig.~\ref{fig:sporadic}(b),
$W$ is assumed to contain a degree-two vertex.
We obtain $\kdot = \bipartite$ and Table~\ref{tab:path:one}.
\end{lemma}

\begin{figure}
\includegraphics{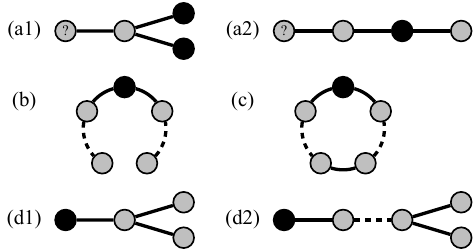}
\caption{\textbf{Lemmas~\ref{lem:induction}-\ref{lem:y:shape}.} 
Notation as in Fig.~\ref{fig:sporadic}.
(a)~$iZ_q$ is generated for $?$-marked
vertices $q$. 
(b)~Path graph with one black vertex not at the ends.
(c)~Cycle graph with one black vertex.
(d)~Y-shaped graphs with one black vertex at the long end
with general form (d2).
\label{fig:ygraph}}
\end{figure}

\begin{proof}
One can computationally verify the statements for $n\leq 4$.
The containment of $\kdot$ in $\so(2^n)$ or $\usp(2^n)$
follows from Lemma~\ref{prop:so:sp}. It is enough to prove the statements
for $\abs{W}=1$. Lemma~\ref{lem:induction} completes the proof.
\end{proof}

Moving from path graphs in Figure~\ref{fig:ygraph}(b) to  cycle graphs in
Figure~\ref{fig:ygraph}(c), we obtain $\bipartiteW = \bipartite$ for
cycle graphs:

\begin{lemma}[Figure~\ref{fig:ygraph}(c)]\label{odd-cycle-one}
The \DLA generated by
$iX_u$ for $1\leq u \leq n \geq 3$, $iZ_vZ_{v+1}$ for $1\leq v < n$, $iZ_1Z_{n}$,
and $iZ_1$ is isomorphic
to $\su(2^n)$ if $n$ is odd, 
$\so(2^n)$ if $n$ is divisible by four, and $\usp(2^n)$ otherwise.
\end{lemma}

\begin{proof}
Removing the edge $(2,3)$ and its generator $iZ_2Z_3$
so that vertex $1$ and $iZ_1$ are not at the end of the resulting path graph, we obtain either
$\so(2^n)$ or $\usp(2^n)$ via Lemma~\ref{path-one}. For $n$ even, 
$iZ_2Z_3$ is already contained in the generated \DLA and we stay with
$\so(2^n)$ or $\usp(2^n)$. For $n$ odd, $iZ_2Z_3$ is not contained in $\so(2^n)$ or $\usp(2^n)$
and Prop.~\ref{lem:max} shows that we generate $\su(2^n)$.
\end{proof}

One challenge arises from the path graphs in Corollary~\ref{coro:spinor} which 
cannot be used as base cases as either $n<4$ or they do not observe $\bipartiteW = \bipartite$.
To provide a set of base cases consisting of cycle graphs and 
path graphs extended 
with one additional vertex, the following Y-shaped graphs play a pivotal role:

\begin{lemma}[Figure~\ref{fig:ygraph}(d)]\label{lem:y:shape}
Given a graph 
with $n{\geq} 4$
and edges $E=\{\{1,3\};\allowbreak{} \{j,j{+}1\}\allowbreak{} \text{ for } \allowbreak{} 2{\leq} j {<}n\}$,
the \DLA
generated by
$iX_v$ for $v\in V$, $iZ_uZ_v$ for $\{u,v\} \in E$, and $iZ_n$
is isomorphic to $\usp(2^n)$ for $n \bmod 4 \in \{0,3\}$
and $\so(2^n)$ for $n \bmod 4 \in \{1,2\}$.
\end{lemma}

\begin{table}
\caption{Cases in Lemma~\ref{path-one} for $n=\abs{V}$ and $W\subseteq V_k$.
    \label{tab:path:one}}
    \begin{tabular}{@{\hspace{1mm}}l@{\hspace{3mm}}r@{\hspace{1mm}}}
    \hline\hline
    \\[-4mm]
    \begin{tabular}{@{\hspace{1mm}}l@{\hspace{4mm}}c@{\hspace{0mm}}c@{\hspace{0mm}}c@{\hspace{0mm}}c@{\hspace{0mm}}c@{\hspace{0mm}}c@{\hspace{0mm}}c@{\hspace{0mm}}c@{\hspace{1mm}}}
    $n \bmod 4$ & 0 & 0 & 1 & 1 & 2 & 2 & 3 & 3\\
    $k \bmod 2$ & 1 & 0 & 1 & 0 & 1 & 0 & 1 & 0
        \\[0.5mm] \hline
    \\[-4mm]
    $\kdot \iso \so(2^n)$ & $\times$ & $\times$ & & $\times$ & & & $\times$\\
    $\kdot \iso \usp(2^n)$ & & & $\times$ & & $\times$ & $\times$ & & $\times$\\[1mm]
    \end{tabular}
    &
    \begin{tabular}{@{\hspace{1mm}}l@{\hspace{4mm}}c@{\hspace{0mm}}c@{\hspace{0mm}}c@{\hspace{0mm}}c@{\hspace{0mm}}c@{\hspace{0mm}}c@{\hspace{0mm}}c@{\hspace{0mm}}c@{\hspace{1mm}}}
    $n \bmod 4$ & 0 & 0 & 1 & 1 & 2 & 2 & 3 & 3\\
    $k \bmod 2$ & 1 & 0 & 1 & 0 & 1 & 0 & 1 & 0
        \\[0.5mm] \hline
    \\[-4mm]
    $iZ_1 \in \kdot$ & $\times$ && $\times$ & \phantom{$\times$} & $\times$ && $\times$ & \phantom{$\times$}\\
    $iZ_n \in \kdot$ && $\times$ & $\times$ &&& $\times$ & $\times$\\[1mm]
    \end{tabular}
    \\[2mm]
    \hline\hline
    \end{tabular}
\end{table}

\begin{proof}
We first prove that we generate $iZ_3$ for $n$ odd and $iZ_4$ for $n$ even.
Using Corollary~\ref{coro:spinor}(ii), $iX_2,\ldots,\allowbreak{}iX_n$; $iZ_2Z_3,\ldots,\allowbreak{}iZ_{n-1}Z_n$; and $iZ_n$
generate the \DLA $\mathfrak{f}$ isomorphic to $\so(2n{-}1)$ and the corresponding basis elements have been
detailed in Section~\ref{appendix:sub:pathcycle}. We have 
\begin{align*}
& [f_1,[f_2,
[\tfrac{i}{2}Z_1Z_3,
[f_3,[f_4,
[\tfrac{i}{2}Z_1Z_3,\tfrac{i}{2}Z_2X_3\!\cdot\cdot{}X_{n-2d}]]]]]]\\
&=\tfrac{i}{2}Z_2X_3\!\cdot\cdot{}X_{n-2d-2},
\end{align*}
where  
$f_1:=iY_2X_3\!\allowbreak\cdot\cdot\allowbreak{}X_{n-2d-1}\allowbreak{}Y_{n-2d}/2$,
$f_2:=iY_3X_4\!\allowbreak\cdot\cdot\allowbreak{}X_{n-2d-1}\allowbreak{}Z_{n-2d}/2$,
$f_3:=iY_3X_4\!\allowbreak\cdot\cdot\allowbreak{}X_{n-2d-2}\allowbreak{}Y_{n-2d-1}/2$,
$f_4:=iY_2X_3\!\allowbreak\cdot\cdot\allowbreak{}X_{n-2d-2}\allowbreak{}Z_{n-2d-1}/2$
with $f_j\in\mathfrak{f}$ as well as
$0\leq 2d \leq n{-}3$ for $n$ odd and $0\leq 2d \leq n{-}4$ for $n$ even.
Thus we can generate $iZ_2X_3$ for $n$ odd and
$iZ_2X_3X_4$ for $n$ even starting from $iZ_2X_3\!\cdot\cdot{}X_{n}\in\mathfrak{f}$
and using $iZ_1Z_3$. As $iZ_2Y_3, iZ_2X_3Y_4 \in\mathfrak{f}$,
we
obtain $iZ_3/2=[iZ_2Y_3/2,iZ_2X_3/2]$ and $iZ_4/2=[iZ_2X_3Y_4/2,iZ_2X_3X_4/2]$
for $n$ odd and even, respectively. Thus
$iX_2,\ldots,\allowbreak{}iX_n$; $iZ_2Z_3,\ldots,\allowbreak{}iZ_{n-1}Z_n$; and 
$iZ_3$ for $n$ odd and $iZ_4$ for $n$ even are used to generate $\so(2^{n-1})$ and 
$\usp(2^{n-1})$ depending on $n$ [see Lemma~\ref{path-one}].
With $iZ_1Z_3$ and
$iX_1$,
Lemma~\ref{lem:induction} completes the proof.
\end{proof}

We formalize the notion of a path graph
extended with one additional vertex and the respective 
free-mixer \DLA augmented by one $iZ_w$ with $w\in V_k$
is determined:

\begin{lemma}\label{lem:extended:path}
Given a path graph $G$,
an \emph{extended path graph} $\tilde{G}$
is obtained by adding one vertex $0$ and one edge $(0,v)$ with $v \in V$ to $G$
such that $\tilde{G}$ is not a path graph. Then $\tilde{G}$
is bipartite with bipartition $\tilde{V}=\tilde{V}_1 {\uplus} \tilde{V}_2=V{\cup}\{0\}$
and $\tilde{V}_j \supseteq V_j$.
For $W=\{w\}\subseteq \tilde{V}_k$,
$\bipartiteWW = \bipartiteH$.
\end{lemma}

\begin{proof}
The result is verified for $\abs{V}\leq 3$ and let $\abs{V} \geq 4$.
Any extended path graph $\tilde{G}$ with $W=\{w\}\subseteq \tilde{V}_k$
is either obtained from Lemma~\ref{lem:y:shape} or
by adding vertices and edges to a graph $\bar{G}$
from Lemma~\ref{path-one} or \ref{lem:y:shape} such that 
the bipartition of $\bar{G}$ is given by $\bar{V}=\bar{V}_1 {\uplus} \bar{V}_2$
with $\bar{V}_j \subseteq V_j$ and
$W\subseteq \bar{V}_k$.
We then apply Lemma~\ref{lem:induction} multiple times.
\end{proof}

With the notion of an extended path graph, we can establish
this critical statement about $\gfree$
augmented with generators $iZ_w$ for $w\in W \subseteq V_k$:

\begin{proposition}\label{prop-bipartite-add-Z}
For a connected bipartite graph $G$
with $\abs{V}=n$,
$\bipartiteW$ for 
$\emptyset \neq W\subseteq V_k$
is generated by 
$iX_v$ for  $v\in V$, $iZ_uZ_v$ for edges  $\{u,v\}$,
and $iZ_w$ for $w\in W$.
For path graphs not in
Fig.~\ref{fig:sporadic}(b),
$W$ is assumed to contain a degree-two vertex.
(i)~$\bipartiteW = \bipartite$.
(ii)~$\bipartiteW \iso\allowbreak{} \so(2^n)$
for even $\abs{V_k}$ and
$\bipartiteW \iso \usp(2^n)$ for
odd $\abs{V_k}$. 
\end{proposition}

\begin{proof}
Clearly, (ii) follows from (i) and Lemma~\ref{prop:so:sp}.
We have checked all cases with $n\leq 4$, and suitable path and even-cycle graphs
are treated in Lemmas~\ref{path-one} and \ref{odd-cycle-one}
and extended path graphs are considered in Lemma~\ref{lem:extended:path}.
We proceed by induction from
a connected bipartite graph $\tilde{G}$ with one vertex $q$ (and its edges) removed from $G$. 
From now on, we assume
$n\geq 5$ and that $G$ is neither a path graph, a cycle graph,
nor an extended path graph. This also implies that $\tilde{G}$ cannot be a path graph.
By induction, we assume that Prop.~\ref{prop-bipartite-add-Z} holds for
all connected, bipartite graphs with less than $n$ vertices 
that are neither a path graph, a cycle graph,
nor an extended path graph.
We pick $q$ as one end vertex of a chosen minimal spanning $T$ tree of $G$,
then $\tilde{G}$ is connected (and clearly bipartite).
We divide the proof into three cases: (1)~$q\not\in W$ and $(V\setminus \{q\}) \cap W\neq \emptyset$,
(2)~$q\in W$ and $(V\setminus \{q\}) \cap W = \emptyset$, and
(3)~$q\in W$ and $(V\setminus \{q\}) \cap W \neq \emptyset$.
For (1) and (3), Prop.~\ref{prop-bipartite-add-Z} holds for $\tilde{G}$ 
as it is either a even-cycle graph [see Lemma~\ref{odd-cycle-one}]
or an extended path graph [see Lemma~\ref{lem:extended:path}]
or it holds for $\tilde{G}$ by induction.
The induction step follows by applying Lemma~\ref{lem:induction}.
For (2), we pick a different end vertex of the chosen
minimal spanning tree $T$ of $G$.
Then the conditions for (1) apply and we proceed as in (1).
\end{proof}

The free-mixer \DLA for connected bipartite graphs
can now be readily determined:

\begin{theorem}[bipartite graphs]\label{thm:free-mixer-bipartite}
Consider a connected bipartite graph $G$ 
which is neither
a path graph nor a cycle graph. Let $V=V_1 {\uplus} V_2$ 
be its vertex bipartition and
$\abs{V}=n\geq 4$.
The free-mixer \DLA for $G$ is given by 
(i)~$\gfree \iso \su(2^{n-1})$ if $n$ is odd (or, equivalently,
if $\abs{V_1}$ and $\abs{V_2}$ have opposite parities),
(ii)~$\gfree \iso \so(2^{n-1}){\oplus}\so(2^{n-1})$
if $\abs{V_1}$ and $\abs{V_2}$ are both even,
and (iii)~$\gfree \iso \usp(2^{n-1}){\oplus}\usp(2^{n-1})$
if $\abs{V_1}$ and $\abs{V_2}$ are both odd.
Note that $\dim[\su(2^{n-1})]=2^{2n-2}{-}1$,
$\dim[\so(2^{n-1}){\oplus}\so(2^{n-1})]=2^{2n-2}{-}2^{n-1}$, and
$\dim[\usp(2^{n-1}){\oplus}\usp(2^{n-1})]=2^{2n-2}{+}2^{n-1}$.
\end{theorem}

\begin{proof}
Let $\tilde{G}$ be
a connected bipartite graph with one vertex $q$ (and its edges) removed from $G$.
Its vertex bipartition $\tilde{V} = \tilde{V}_1 {\uplus} \tilde{V}_2 = V\setminus \{q\}$ 
observes $\tilde{V}_j \subseteq V_j$.
We recall from Def.~\ref{def:so:sp} the notation $k,\bar{k} \in \{1,2\}$ with $k\neq \bar{k}$.
We pick $q$ as one end vertex of a chosen minimal spanning tree of $G$,
then $\tilde{G}$ is connected (and clearly bipartite). We permute the numbers of the vertices so that $q$ becomes the first vertex.
We transform the generators $iX_v$ for $v \in V$ and  $iZ_uZ_v$ for edges $\{u,v\}$ of $G$ into the basis 
from Eq.~\eqref{generators_sym} and obtain
(for $u, v \in \tilde{V}$)
(a)~$iX_u$, (b)~$iZ_uZ_v$, (c)~$iZ_w$ for 
$w\in W \subseteq \tilde{V}_k \subseteq \tilde{V}$ where $W\neq \emptyset$ denotes the set of neighbors of $1$,
and
(d)~$iZ_1 \prod_{w\in \tilde{V}} X_w$.
As $G$ is neither a path graph nor a cycle graph,
$W$ contains at least one vertex of $\tilde{G}$ with degree two
if $\tilde{G}$ is a path graph. In the notation of Def.~\ref{def:so:sp},
it follows from Prop.~\ref{prop-bipartite-add-Z} that $\bipartiteWW$ is
equal to $\bipartiteH$ which is isomorphic to 
$\so(2^{n-1})$ if $|\tilde{V}_k|$ is even 
and to $\usp(2^{n-1})$ if $|\tilde{V}_k|$ is odd.
Here, $\bipartiteH$ is block-diagonally spanned by
$i I {\otimes} (\bigotimes_{v\in \tilde{V}} A_v)$ with $A_v \in \{I,X,Y,Z\}$ where the number of $A_v=X$ with 
$v \in \tilde{V}$ has the opposite parity of the number of $A_u\in \{Y,Z\}$ with $u \in \tilde{V}_k$.
We consider the cases (1)~$i\prod_{w\in \tilde{V}} X_w \in \bipartiteH$
[i.e.\ $n{-}1$ is odd due to
Lemma~\ref{prop:so:sp}(v)]
and (2)~$i\prod_{w\in \tilde{V}} X_w \not\in \bipartiteH$.
For (1), we add $iZ_1 \prod_{w\in \tilde{V}} X_w$
and  apply Prop.~\ref{lem:max} to obtain
$\bipartiteH{\oplus}\bipartiteH$.
For (2), adding the generator $iZ_1 \prod_{w\in \tilde{V}} X_w$
results in $\su(2^{n-1})$ via Prop.~\ref{lem:max}.
Its basis is spanned by all
$i I {\otimes} (\bigotimes_{v\in \tilde{V}} A_v)$ 
as defined before and $i Z {\otimes} B_j$ for 
Pauli strings $iB_j \not\in \bipartiteH$.
\end{proof}

\subsection{Non-bipartite graphs\label{appendix:sub:non:bipartite}}

Finally, we treat connected non-bipartite graphs also by adding a
generator of the form $iZ_j$ and by applying the basis change
leading to Eqs.~\eqref{generators_sym}-\eqref{generators_sym_add}:

\begin{proposition}\label{prop-non-bipartite}
Consider a connected, non-bipartite graph with $n$ vertices $V$ and edges $E$
and choose a vertex $j\in V$.
Then the generators $iX_v$ for $v\in V$, $iZ_uZ_v$ for $\{u,v\} \in E$, and $iZ_j$
always generate together the \DLA $\su(2^n)$.
\end{proposition}

\begin{proof}
For an odd-cycle graph, the result follows from Lemma~\ref{odd-cycle-one}. 
We can assume that $n\geq 4$
and that the shortest odd cycle in the non-bipartite graph is shorter than $n$.
We re-order the vertices such that one shortest odd cycle does not contain
the first vertex. Now the subgraph without the first vertex is not bipartite.
If this subgraph is not connected, 
pick a connected component not containing the chosen shortest odd cycle.
In the subgraph given by this connected component and the first vertex,
choose a minimal spanning tree and 
pick one of its end vertices different
from the first vertex as the new first vertex. Then the
subgraph not containing the new first vertex is connected.
Following the analysis leading to Eqs.~\eqref{generators_sym}-\eqref{generators_sym_add},
the generators
from the groups (a), (b), and (c)
in Eq.~\eqref{generators_sym}
are of tensor-product form $i(\otimes_{j} A_j)$ with $A_1 = I$ and they generate by induction 
the \DLA $\g_{a}\iso\su(2^{n-1})$ (Ex.~\ref{def_ga_gb}).
Using the generator $iZ_1X_2 \!\cdot\cdot X_n$
from group (d) in Eq.~\eqref{generators_sym},
Lemma~\ref{lem_ga_gb}(a) shows that we get $\g_{b}\iso \su(2^{n-1}) {\oplus} \su(2^{n-1})$
(Ex.~\ref{def_ga_gb}).
Depending on $j$, we add one of the generators $iX_1$ or $iX_1Z_j$
from the group~(e) in  Eq.~\eqref{generators_sym_add}.
Lemma~\ref{lem_ga_gb}(b) shows
that we first generate $iZ_1$ and thus 
$\g_c \iso \su(2^{n-1}) {\oplus} \su(2^{n-1}) {\oplus} \uu(1)$ (Ex.~\ref{def_ga_gb}),
which is maximal
in $\su(2^n)$ [see Prop.~\ref{lem:max}].
By the argument in Lemma~\ref{lem_ga_gb}(b3), we obtain the
\DLA $\su(2^n)$ by adding one generator from (e) in  Eq.~\eqref{generators_sym_add}.
\end{proof}

Determining the free-mixer \DLA in the non-bipartite case is now immediate:

\begin{theorem}[non-bipartite graphs]\label{thm:free-mixer-revisited}
The free-mixer \DLA $\gfree$ for a connected non-bipartite graph different from
a cycle graph with $n\geq 2$ vertices
is given by 
$\gfree\iso \su(2^{n-1}){\oplus}\su(2^{n-1})$ with $\dim(\gfree) = 2^{2n-1} - 2$.
\end{theorem}

\begin{proof}
Similarly as in the proof of Prop.~\ref{prop-non-bipartite},
we can assume $n\geq 4$ and we re-order the vertices such
that the subgraph without the first vertex is connected
and not bipartite.
We apply the analysis leading to Eq.~\eqref{generators_sym}.
It follows from Proposition~\ref{prop-non-bipartite} that the generators
from the groups (a), (b), and (c) in Eq.~\eqref{generators_sym} generate 
the \DLA $\g_{a}\iso\su(2^{n-1})$ (Ex.~\ref{def_ga_gb}).
Adding
$iZ_1X_2 \!\cdot\cdot X_n$
from group (d) in Eq.~\eqref{generators_sym},
Lemma~\ref{lem_ga_gb}(a) shows that we generate $\g_{b}\iso \su(2^{n-1}) {\oplus} \su(2^{n-1})$
(Ex.~\ref{def_ga_gb}).
\end{proof}

\subsection{Explicit representations: neither path nor cycle \label{appendix:explicit:reps:not:path:cycle}}
We complement the results in Sections~\ref{appendix:sub:bipartite}
and \ref{appendix:sub:non:bipartite} with the explicit
form of the corresponding representations. This determines
the embedding of $\gfree$
into $\su(2^n)$.
The standard representation
of a \DLA is denoted by $\kappa$ and $\bar{\kappa}$ identifies its dual;
$\epsilon$ is the trivial representation. 

\begin{proposition}[Free-mixer representations for connected graphs different from path and cycle]\label{prop:reps:not:path:cycle}
Consider
the free-mixer \DLAs $\gfree$ for the bipartite cases (i), (ii), 
and (iii) from Theorem~\ref{thm:free-mixer-bipartite}
as well as for the non-bipartite case from Theorem~\ref{thm:free-mixer-revisited}.
The respective representations are 
(up to an automorphism of $\gfree$) 
given by $\kappa {\oplus} \bar{\kappa}$ for the first case and by
$[\kappa {\otimes} \epsilon] {\oplus} [\epsilon {\otimes} \kappa]$
otherwise.
\end{proposition}

\begin{proof}
Lemma~\ref{app:prop:center} and the discussion in 
Appendix~\ref{app:free:preliminaries} imply that the commutant
$\com(\gfree)$ is two-dimensional and that the corresponding
representation $\xi_1 {\oplus} \xi_2$ of $\gfree$ splits into
two inequivalent irreducible representations $\xi_1$ and $\xi_2$ 
of degree $2^{n-1}$. The reductive decomposition
of $\gfree$ for the considered cases are known
from Theorems~\ref{thm:free-mixer-bipartite}
and \ref{thm:free-mixer-revisited}.
We assume $n\geq 5$ as
the cases with $n \leq 4$ can be verified by explicit computations.
For $n\geq 5$, all irreducible representations of 
degree $2^{n-1}$ for
$\so(2^{n-1})$, $\usp(2^{n-1})$, and $\su(2^{n-1})$
are respectively given 
(up to equivalence) by
the standard representation $\kappa$,
the standard representation $\kappa$, 
and the standard representation $\kappa$
and its dual $\bar{\kappa}$. Referring, e.g., to \cite{Bourbaki2008b},
the listed irreducible representations have the required
degree and they are the only possibilities up to equivalence,
which follows from the well-known
enumeration of the respective lowest-dimensional irreducible representations
(see, e.g., Chapter VIII, \S 13, Exercise 2 in \cite{Bourbaki2008b}).
This reduces the possibilities significantly and the results follow
as automorphisms of $\su(2^{n-1})$ switch between
$\kappa$ and $\bar{\kappa}$.
\end{proof}

\section{The free-mixer ansatz via Pauli strings\label{appendix:free-mixer:bases}}
This appendix complements the proof techniques in Appendix~\ref{appendix:free-mixer}
with an approach focused on Pauli-string bases.
For a particular graph, the original generators from Eq.~\eqref{freegenerators}
yield further Lie-algebra elements which in some instances can be
interpreted as adding edges to the graph. This enables us to highlight
and employ properties of the considered graph.
We start with general properties of bases of Pauli strings in Appendix~\ref{appendix:free-mixer:pauli}.
The complete and complete bipartite graphs are considered in Appendix~\ref{appendix:pauli:complete}.
Appendix~\ref{app:adding:edge} details the technique of adding edges to a graph.
This is then used to treat the bipartite and non-bipartite cases (different from path and cycle graphs).
In Appendix~\ref{app:pauli:bipartite:lie}, more work is needed to identify
the \DLAs for the bipartite graphs that are not path and cycle graphs.
We refer to Appendix~\ref{appendix:sub:pathcycle} for the cases of path and cycle graphs.

\subsection{Pauli-string bases\label{appendix:free-mixer:pauli}}

For a given Pauli string
\begin{equation}\label{tensorproductoperator}
P_j=\bigotimes_{u\in V} A_u \text{ with } A_u
\in \{X, Y, Z, I\},
\end{equation}
or its corresponding Lie-algebra element $i P_j$, let
$\#\mathrm{X}$, $\#\mathrm{Y}$, $\#\mathrm{Z}$,
and $\#\mathrm{I}$ denote the respective number of $\mathrm{X}$, $\mathrm{Y}$, $\mathrm{Z}$, and $\mathrm{I}$.
We also use the notation $\#\mathrm{X}(P_j)$ to identify a particular $P_j$.
And for instance,
$\#\mathrm{X}|_{V_1}$ indicates the restriction of $\#\mathrm{X}$ to a subset $V_1$ of $V$.
We introduce the parity $p(P_j):=\#\mathrm{Y}(P_j) + \#\mathrm{Z}(P_j)$ (or $p$ for short) of a Pauli string $P_j$
(or similar elements).

\begin{lemma}\label{app:prop:parity}
Given a connected graph and its free-mixer \DLA $\gfree$,
consider
$\sum_j r_j i P_j \in \gfree$ for Pauli strings $P_j$ as in Eq.~\eqref{tensorproductoperator} with
$0\neq r_j \in \R$. (i)~The parity $p(P_j):=\#\mathrm{Y}(P_j) + \#\mathrm{Z}(P_j)$
of each $P_j$ is even. (ii) For bipartite graphs with bipartition $V = V_1{\uplus}V_2$,
$\tilde{p}(P_j):=\#\mathrm{X}(P_j)+ \#\mathrm{Y}|_{V_1}(P_j) + \#\mathrm{Z}|_{V_1}(P_j)$ is odd.
\end{lemma}

\begin{proof}
We observe (i) and (ii)
for all generators from Eq.~\eqref{freegenerators}. Following Fig.~\ref{fig:commutator}, $p(P_j) \bmod 2$ 
and $\tilde{p}(P_j) \bmod 2$
are invariant
for each $P_j$ under commutators with generators from  Eq.~\eqref{freegenerators}.
This proves (i) and (ii) for Pauli strings $P_j$. Linear
combinations are similar.
\end{proof}

\begin{figure}
\includegraphics{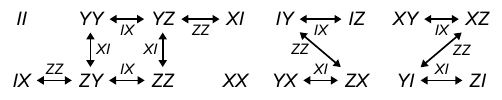}
\caption{\textbf{Commutator graph.}
A directed edge with label $C$ connects a vertex $A$ with a vertex $B$ iff
$[iC/2,iA/2]=(-1)^s iB/2$ for $s\in \{0,1\}$.
\label{fig:commutator}}
\end{figure}

Adding Lemma~\ref{app:prop:center}(ii) to the conditions in Lemma~\ref{app:prop:parity},
we obtain the following set of conditions
\begin{subequations}
\begin{align}
& \#\mathrm{I}(P_j)\neq n,\; \#\mathrm{X}(P_j)\neq n, \label{eq:cond:a}\\
& p(P_j)= \#\mathrm{Y}(P_j) + \#\mathrm{Z}(P_j) \text{ is even}, \text{ and} \label{eq:cond:b}\\
& \tilde{p}(P_j)=\#\mathrm{X}(P_j) + \#\mathrm{Y}|_{V_1}(P_j) + \#\mathrm{Z}|_{V_1}(P_j) \text{ is odd}\,, \label{eq:cond:c}
\end{align}
\end{subequations}
for Pauli strings $P_j$ where
$V=V_1{\uplus}V_2$ denotes the bipartition of a bipartite graph. 
We now count
the number of $P_j$ satisfying Eqs.~\eqref{eq:cond:a}-\eqref{eq:cond:b}
or Eqs.~\eqref{eq:cond:a}-\eqref{eq:cond:c} depending on the number $n$ of vertices
as well as the parity of $\abs{V_1}$ and $\abs{V_2}$:

\begin{lemma}[Counting Pauli strings]\label{lem:pauli}
(i)~The number of $P_j$ satisfying Eqs.~\eqref{eq:cond:a}-\eqref{eq:cond:b}
is equal to $2^{2n-1} {-} 2$.
(ii)~The number of $P_j$ observing Eqs.~\eqref{eq:cond:a}-\eqref{eq:cond:c} 
for a complete bipartite graph with bipartition $V=V_1{\uplus}V_2$
is (a)~$2^{2n-2} {-} 1$ for $n$ odd,
(b)~$2^{2n-2} {-} 2^{n-1}$ for $\abs{V_1}=\abs{V_2}$ even,
and 
(c)~$2^{2n-2} {+} 2^{n-1}$ for 
$\abs{V_1}=\abs{V_2}$ odd.
\end{lemma}

\begin{proof}
The number of $P_j$
that satisfy Eq.~\eqref{eq:cond:b} is with $2^{2n-1}$ exactly half
of all possible ones. Both $I^{\otimes n}$ and $X^{\otimes n}$
observe Eq.~\eqref{eq:cond:b} and (i) follows.
For $P_j=\otimes_{u\in V} A_u$ with
$A_u \in \{X, Y, Z, I\}$ and
$p(P_j)\neq n$, 
let us introduce a map $b$ from $P_j$ to 
$b(P_j):=\otimes_{u\in V} B_u$ with 
$B_u \in\{I,X\}\setminus \{A_u\}$ for the first vertex $u\in V$ with $A_u\in\{I,X\}$
and $B_u=A_u$ for $u\neq v$. The map $b$ induces a bijection on $P_j$
with $p(P_j)\neq n$ which switches the parity, i.e.\ $p(P_j)+p(b(P_j))$ is odd.
For $n$ odd, $p(P_j)\neq n$ and there is at least one $u\in V$ with $A_u \in \{I,X\}$.
Hence $b$ shows that half of $P_j$ satisfying Eq.~\eqref{eq:cond:b}
also observe Eq.~\eqref{eq:cond:c}. The cases $P_j \in \{I^{\otimes n}, X^{\otimes n}\}$
are ruled out by Lemma~\ref{app:prop:center}(ii), but
Eq.~\eqref{eq:cond:c} fails for $I^{\otimes n}$,
while Eqs.~\eqref{eq:cond:b} and \eqref{eq:cond:c} hold for $X^{\otimes n}$. This proves (iia).
If $n$ is even, there are $2^n$ operators $P_j$ with $p(P_j)=n$, i.e.\ $\#X(P_j)=0$.
For Eq.~\eqref{eq:cond:c} to be fulfilled for these $2^n$ operators, both $\#\mathrm{Y}|_{V_1}(P_j) + \#\mathrm{Z}|_{V_1}(P_j)$
and $\#\mathrm{Y}|_{V_2}(P_j) + \#\mathrm{Z}|_{V_2}(P_j)$ need to be odd, i.e., $\abs{V_1}=\abs{V_2}$ is odd.
The bijection $b$ is applied to $2^{2n-1} {-} 2^n$ operators,
half of which
satisfy Eq.~\eqref{eq:cond:c}.
Both $I^{\otimes n}$ and $X^{\otimes n}$ do not observe Eq.~\eqref{eq:cond:c}.
We obtain (iib) and (iic).
\end{proof}

We know from Lemma~\ref{app:prop:parity} that $\gfree$ is contained
in the span of the $iP_j$ with $P_j$
satisfying Eqs.~\eqref{eq:cond:a}-\eqref{eq:cond:b} for connected graphs.
For connected bipartite graphs, the additional condition in Eq.~\eqref{eq:cond:c}
has to be observed.
The number of $P_j$
agrees
with the respective dimensions in 
Thm.~\ref{thm:free-mixer-revisited} for non-bipartite graphs different from
a cycle graph or in 
Thm.~\ref{thm:free-mixer-bipartite}
for  bipartite graphs which are neither
a path graph nor a cycle graph.
Thus $\gfree$ is spanned in these cases 
by the respective sets of $iP_j$. Combined with 
Appendix~\ref{appendix:free-mixer}, this completes the derivation 
of Table~\ref{tab:free-mixer-full-table}.

\subsection{Complete and complete bipartite graphs\label{appendix:pauli:complete}}

\begin{lemma}\label{lem:free-mixer-complete}
Consider the free-mixer \DLA $\gfree$ for a complete graph.
(i)~$\gfree$ is spanned by $iP_j$ for Pauli strings $P_j$
satisfying Eqs.~\eqref{eq:cond:a}-\eqref{eq:cond:b}.
(ii)~$\dim(\gfree) = 2^{2n-1} {-} 2$.
(iii)~$\gfree \iso\su(2^{n-1}){\oplus}\su(2^{n-1})$.
\end{lemma}

\begin{proof}
Lemmas~\ref{app:prop:center} and \ref{app:prop:parity} show that the elements
of $\gfree$ have no support on odd-parity Pauli strings and that $\#\mathrm{I}\neq n$ and $\#\mathrm{X}\neq n$.
The results can readily established for $n=2$ and we now assume $n\geq 3$.
Applying Fig.~\ref{fig:commutator}, $iX_u$, $iX_v$, $iY_uY_v$, $iY_uZ_v$, and $iZ_uZ_v$
are contained in $\gfree$ for all vertices $u,v \in V$ with $u\neq v$. We also obtain $iX_u X_v \in \gfree$ as
$[[iY_uY_v,iZ_uZ_w],iZ_vZ_w]=4 iX_u X_v$ for three different vertices $u,v,w \in V$. We prove by induction
over $k:=n-\#\mathrm{I}$ that the stated basis elements are contained in $\gfree$, and we have just
verified the base cases of $k\in \{1,2\}$. We treat the two cases of
(a) $p  > 0$ even and (b) $p = 0$ separately. 
For (a), we have a
$iP_j$ such that $p(P_j)=p \geq 2$, $\#\mathrm{I}(P_j)=n{-}k$, and $\#\mathrm{X}(P_j)=k{-}p$.
If $k>p>0$, we can choose $v, w \in V$
such that $A_v = X$ and $A_w \in \{Y,Z\}$. Using Fig.~\ref{fig:commutator},
$X_vZ_w$ and $X_vY_w$ are reached from $Y_vI_w$ with a smaller $k$.
If $k=p>0$, we can choose $v_1,v_2,v_3 \in V$ such that $A_{v_j} \in \{Y,Z\}$.
Using Fig.~\ref{fig:commutator}, for example, $Y_{v_1}Y_{v_2}Y_{v_3}$ is reached from $Y_{v_1}I_{v_2}I_{v_3}$ with a smaller $k$.
For (b), $P_j$ has only
$X$ and $I$ in its tensor product with $\#\mathrm{X}(P_j)=k\geq 1$ and $\#\mathrm{I}(P_j)=n{-}k\geq 1$.
For $k\geq 2$, we choose $v_1,v_2,v_3 \in V$ such that $A_{v_1} = X$, $A_{v_2} = X$, and $A_{v_3} = I$.
Using Fig.~\ref{fig:commutator}, $X_{v_1}X_{v_2}I_{v_3}$ is reached from $X_{v_1}I_{v_3}I_{v_3}$ with a smaller $k$.
In all three cases, a smaller $k$ is obtained which shows (i) by induction.
Lemma~\ref{lem:pauli}(i) implies (ii). 
By Lemma~\ref{app:prop:structure}, $\gfree$ is isomorphic to a subalgebra of $\su(2^{n-1}){\oplus}\su(2^{n-1})$.
Hence (ii) implies (iii).
\end{proof}

\begin{lemma}\label{lem:bipartite:dimensions:combinatorics}
Consider the free-mixer \DLA $\gfree$ for a complete bipartite graph with $n\geq 4$
and its bipartition $V=V_1{\uplus}V_2$.
(i)~$\gfree$ is spanned by $iP_j$ for Pauli strings $P_j$
satisfying Eqs.~\eqref{eq:cond:a}-\eqref{eq:cond:c}.
The number of $P_j$ is
(ii)~$2^{2n-2} {-} 1$ for $n$ odd,
(iii)~$2^{2n-2} {-} 2^{n-1}$ for $\abs{V_1}=\abs{V_2}$ even,
and 
(iv)~$2^{2n-2} {+} 2^{n-1}$ for 
$\abs{V_1}=\abs{V_2}$ odd.
\end{lemma}

\begin{proof}
Lemmas~\ref{app:prop:center} and \ref{app:prop:parity} show that no other Pauli strings $iP_j$ can be generated.
The proof proceeds by induction on $k:=n-\#\mathrm{I}$. For $k=1$, only $iX_u$ for $u\in V$ are possible and these are contained 
in Eq.~\eqref{freegenerators}. For $k=2$, we obtain $iZ_uZ_v$, $iY_uZ_v$, $iZ_uY_v$, and $iY_uY_v$ for an edge $\{u,v\}$
following Fig.~\ref{fig:commutator} and all other possibilities in Fig.~\ref{fig:commutator} violate Eqs.~\eqref{eq:cond:a}-\eqref{eq:cond:c}.
We now assume $3\leq k \leq n$ and consider a Pauli string $P_j=\otimes_{u\in V} A_u$ with $A_u
\in \{X, Y, Z, I\}$ as in Eq.~\eqref{tensorproductoperator}.
We consider the cases such that either (a) $A_u \in \{Y,Z\}$ and $A_v\in\{X,Y,Z\}$ hold for $u\in V_1$ and $v\in V_2$
(and for $V_1$ and $V_2$ exchanged),
(b) $A_u \in \{Y,Z\}$ and $A_v=I$ hold for $u\in V_1$ and $v\in V_2$ (and for $V_1$ and $V_2$ exchanged), 
or (c) $p(P_j)= \#\mathrm{Y}(P_j) + \#\mathrm{Z}(P_j) =0$.
Let $B_{b_1}\ldots B_{b_s}\vert C_{c_1}\ldots C_{c_t}$ 
denote a Pauli string consisting of a subset of $P_j$,
where $b_j \in V_1$, $c_j \in V_2$, $B_{b_j}:=A_{b_j}$, $C_{c_j}:=A_{c_j}$
$s+t\leq n$, $s\leq \abs{V_1}$, and $t\leq \abs{V_2}$.
For (a), we have, e.g., the possibilities $Z\vert X$ which can be generated from $I\vert Y$ with a smaller $k$
and $Z\vert Z$ which can be generated from $X\vert I$ with a smaller $k$ (see Fig.~\ref{fig:commutator}) and all other cases
are similar. This resolves (a) by induction.
For (b), we have, e.g., the possibility $YY\vert I$ which can be generated from $XX\vert I$ and this case is treated in (c).
For (c), $\#X(P_j)$ is odd and thus $\#X(P_j)\geq 3$. The case (c) 
is further divided into the cases (c1) $\abs{V_1},\abs{V_2}\geq 2$ and (c2)  $\abs{V_1}=1$ or $\abs{V_2}=1$.
For (c1), we have $IX\vert XX$ which can be generated from $IY\vert IY$ with a smaller $k$.
For (c2), we have either $I\vert XXX$ or $X\vert IXX$ which can be both generated from $Y\vert YII$ with a smaller $k$.
We resolve (c) by induction, which proves (i).
Lemma~\ref{lem:pauli}(ii) implies (ii)-(iv).
\end{proof}

\subsection{Adding edges to a connected graph\label{app:adding:edge}}

Recall from Eqs.~\eqref{freegenerators}-\eqref{nonfreegens} that the problem
Hamiltonian $\sum_{\{u,v\} \in E} iZ_uZ_v$
can be split for the free-mixer ansatz 
into separate $iZ_uZ_v$ for each edge $\{u,v\}\in E$. Thus generating
$iZ_jZ_k$ for $\{j,k\}\not\in E$ is equivalent to adding $\{j,k\}$ to the edges without
changing the free-mixer \DLA. More and more edges can be added 
for graphs different from
path and cycle graphs
until either a complete or complete bipartite graphs is reached. 
This explains why only so few classes of free-mixer \DLAs occur.
Two particular configurations that allow us to add edges 
are Y-configurations and PAW-configurations
[see Figs.~\ref{fig:yconfiguration}(a)-(b)], 
which can be directly verified (see Fig.~\ref{fig:commutator}):

\begin{lemma}\label{lem:y:paw}
Consider a connected graph and either the (a) five- 
or (b) four-vertex subgraph
depicted in Fig.~\ref{fig:yconfiguration}(a)-(b) with the generators $iX_u$ for vertices $v$ and
$iZ_uZ_v$ for edges $\{u,v\}$. The respective elements (a) $iZ_1Z_5$ and $iZ_2Z_5$
or (b) $iZ_1Z_4$ and $iZ_2Z_4$ are generated.
\end{lemma}

For connected graphs that differ from a cycle graph but contain a cycle 
as depicted in Fig.~\ref{fig:yconfiguration}(c), we can generate an additional edge and
shorten the cycle by applying Lemma~\ref{lem:y:paw}(a). Repeating this argument 
leads to the following two propositions. 

\begin{proposition}\label{prop:triangle:paw}
Consider a connected non-bipartite graph different from a cycle graph.
Without changing its free-mixer
 \DLA, one can add edges
such that the resulting graph 
(i) has a PAW-configuration (and a triangle) as a subgraph
and (ii) is equal to the complete graph.
(iii)~$\gfree \iso \su(2^{n-1}){\oplus}\su(2^{n-1})$.
\end{proposition}

\begin{proof}
A non-bipartite graph contains an odd-cycle. As the graph is also connected and different
from a cycle graph, an
additional vertex exists that is connected to one of the vertices in the odd-cycle. As 
the vertices $1$ to $5$ in Fig.~\ref{fig:yconfiguration}(c) form a Y-configuration, we apply Lemma~\ref{lem:y:paw}(a)
to shorten the odd-cycle by adding the edge $\{2,5\}$. This proves (i) by induction.
Applying (i), we have a triangle subgraph which is formed by some vertices $1$, $2$, and $3$. 
Every vertex that is not directly connected to this triangle (say vertex $5$) is
indirectly connected to one vertex in this triangle (say vertex $3$) via some vertex
(say vertex $4$) as shown in Fig.~\ref{fig:yconfiguration}(d). We add further edges without changing $\gfree$
by repeatedly applying Lemma~\ref{lem:y:paw}(b) to PAW-configurations, starting
with the edges $\{1,4\}$ and $\{2,4\}$ in Fig.~\ref{fig:yconfiguration}(d). As $1$, $2$, and $4$ form a new triangle,
we can add the edges $\{1,5\}$ and $\{2,5\}$ by induction. 
This argument shows that any vertex $v$ is directly connected to two vertices in $\{1,2,3\}$, say the vertices $1$ and $2$. 
Using the PAW-configuration given by the edge $\{1,3\}$ and the triangle
$\{1, 2, v\}$, we add the edge $\{3,v\}$. By the previous arguments, any two vertices $v,w\in V\setminus \{1,2,3\}$ with $v\neq w$ are directly connected
to $1$, $2$, and $3$. Using the PAW-configuration given by the edge $\{2,w\}$ and the triangle $\{1, 2, v\}$, we add the edge $\{v,w\}$. This proves (ii). 
Lemma~\ref{lem:free-mixer-complete}(iii) implies (iii).
\end{proof}

\begin{figure}
\includegraphics{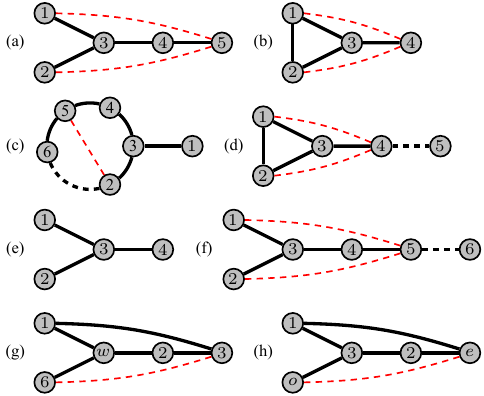}
\caption{\textbf{Add red edges, keep \DLA.}
(a)~Add the edges $\{1,5\}$ and $\{2,5\}$ to the Y-configuration.
(b)~Add the edges $\{1,4\}$ and $\{2,4\}$ to the PAW-configuration.
(c)-(h)~are similar.
Generators are $iX_u$ for all vertices $u$ and
$iZ_uZ_{v}$ for edges $\{u,v\}$ (see Fig.~\ref{fig:sporadic}).
\label{fig:yconfiguration}}
\end{figure}

\begin{proposition}\label{prop:pauli:bipartite}
Consider a connected bipartite graph different from a path or cycle graph.
Without changing its free-mixer
 \DLA, one can add edges
such that the resulting graph 
is equal to a complete bipartite graph.
\end{proposition}

\begin{proof}
Clearly, one vertex (say vertex $3$) has a degree of at least three and
it has three neighbors (say the vertices in $W=\{1, 2, 4\}$)
as shown in Fig.~\ref{fig:yconfiguration}(e). None of the neighbors of
$3$ can be directly connected as this would imply a triangle
in a bipartite graph which is impossible. If the neighbors
of $3$ have no additional neighbors except $3$ than the graph is
a the complete bipartite graph with bipartition $V=V_1{\uplus}V_2$ where
$\abs{V_1}=1$ and we have proven the result. So we now assume that
$\abs{V_1}, \abs{V_2}>1$ and one vertex in $W$ (say vertex $4$)
has a neighbor (say vertex $5$) different from $3$
as depicted in Fig.~\ref{fig:yconfiguration}(a).
Let $d(v)$ denote the (minimal) distance of a
vertex $v\in V$ from the vertex $3$ and the vertex $v$
is called even or odd depending whether $d(v)$ is even or odd
(which is well-defined as the graph is bipartite).
Our objective is to add all edges that
connect even vertices with odd ones
without changing the free-mixer \DLA (which is tacitly assumed in the following), 
which would prove the desired result.

First, we show that (*) one can add edges to connect all even vertices to the vertices in $W$ and all odd vertices
to $3$. We proceed by induction on $d(v)$. For $d(v)=0$, $v=3$ and it is already directly connected to the vertices in $W$.
For $d(v)=1$, $v$ is already a neighbor of $3$. For $d(v)=2$, there is a path from $v$ (say vertex $5$) to $3$
via one of its neighbors (say vertex $4$) as shown in Fig.~\ref{fig:yconfiguration}(a). 
Lemma~\ref{lem:y:paw}(a) implies
that we can add the edges $(1,5)$ and $(2,5)$ and the vertex $5$ is connected to all vertices in $W$.
For $d(v)\geq 3$, there is a path from $v$ (say vertex $6$) to $3$ via a neighbor (say vertex $5$)
of a vertex in $W$ (say vertex $4$) as shown in Fig.~\ref{fig:yconfiguration}(f).
We are free to assume that this path does not contain the other two vertices $1$ and $2$ in $W$
(as we could shorten the path).
If $d(6)\geq 3$ is even, we can apply Lemma~\ref{lem:y:paw}(a) by induction to add the edges $(1,6)$ and $(2,6)$.
Adding the edge $(1,6)$ to the star in Fig.~\ref{fig:yconfiguration}(e), Lemma~\ref{lem:y:paw}(a) implies
that edge $(4,6)$ can be added. If $d(6)\geq 3$ is odd, let $w$ be the neighbor of $6$ on the path from
$3$ to $6$ where $d(w)=d(6){-}1$ is even and $w \not\in W$. By induction, we can add the edges $\{w,1\}$,
$\{w,2\}$, and $\{w,4\}$. We obtain the situation in Fig.~\ref{fig:yconfiguration}(g) and can add $(3,6)$.
This proves (*) by combining all cases.
With the help of Fig.~\ref{fig:yconfiguration}(h),
it remains to show that any even vertex $e\neq 3$ can be directly connected to any odd vertex $o\not\in W$.
\end{proof}

\subsection{Identifying the \DLAs for bipartite graphs\label{app:pauli:bipartite:lie}}
We have already identified the free-mixer \DLA for non-bipartite graphs different from cycle graphs
as $\gfree \iso \su(2^{n-1}){\oplus}\su(2^{n-1})$ and we refer to the direct isomorphisms
in Section~\ref{appendix:sub:pathcycle}
for the cases of path and cycle graphs. It remains to tackle all other bipartite graphs (different from path and cycle graphs).
Lemma~\ref{lem:bipartite:dimensions:combinatorics} already states the respective dimension
for the relevant cases and Proposition~\ref{prop:pauli:bipartite} reduces the determination of the free-mixer \DLA
to the case of complete bipartite graphs. We first consider graphs with an even number of vertices while providing a different argument
as in Theorem~\ref{thm:free-mixer-bipartite}:

\begin{proposition}[bipartite graphs, $n$ even]\label{thm:free-mixer-bipartite-even-combinatorics}
Given a connected bipartite graph $G$ 
that is neither
a path graph nor a cycle graph, let $V=V_1 {\uplus} V_2$ 
denote its vertex bipartition and assume that
$\abs{V}=\abs{V_1}+\abs{V_2}=n\geq 4$ is even.
The free-mixer \DLA for $G$ is given by 
(i)~$\gfree \iso \so(2^{n-1}){\oplus}\so(2^{n-1})$
if $\abs{V_1}$ and $\abs{V_2}$ are both even,
and (ii)~$\gfree \iso \usp(2^{n-1}){\oplus}\usp(2^{n-1})$
if $\abs{V_1}$ and $\abs{V_2}$ are both odd.
\end{proposition}

\begin{proof}
Recall from Lemma~\ref{app:prop:center}(a) and the discussion in Section~\ref{app:free:preliminaries}
that the action of $\gfree$ splits into two irreducible and invariant subspaces $\HC_{+}$
and $\HC_{-}$ which are spanned by Hadamard basis states with respectively an even or odd number of minus
signs. Moreover, $\gfree=\gfree^{+}{\oplus}\gfree^{-}$ decomposes into two isomorphic, semisimple ideals with respect
to $\HC_{+}$
and $\HC_{-}$.
We will apply 
Proposition~\ref{prop_obata} to $\gfree^{+}$ by identifying 
a matrix $S^{+}$ such that (*)~$S^{+} H_j^{+} {+} (H_j^{+})^{t} S^{+} = 0$ for all generators $iH_j^{+}$
of $\gfree^{+}$ (and similarly for $\gfree^{-}$). We can assume that $V_1=\{1,\ldots,\abs{V_1}\}$ with
$\abs{V_1}>1$. Let $S:= Z_1\!\cdots\! Z_{\abs{V_1}} Y_{\abs{V_1}+1}\!\cdots\!Y_n$ 
where $S$ is symmetric if $\abs{V_2}$ is even. Switching to the basis of Eq.~\eqref{generators_sym},
$S$ is transformed to $I_1Z_2\!\cdots\! Z_{\abs{V_1}} Y_{\abs{V_1}+1}\!\cdots\!Y_n$. We obtain $S^{\pm}=
Z_2\!\cdots\! Z_{\abs{V_1}} Y_{\abs{V_1}+1}\!\cdots\!Y_n$ as a matrix on $n-1$ qubits acting on the respective invariant blocks,
where $S^{\pm}$ is symmetric if $\abs{V_2}$ is even.
Recalling that $n-1$ is odd,
we can directly verify (*) by referring to the explicit matrices $X_u$ for $u\geq 2$,
$Z_{v_1} Z_{v_2}$ for $v_1 \in V_1 \setminus \{1\}$ and $v_2 \in V_2$,
$Z_{v_2}$ for $v_2 \in V_2$, and $X_2 \!\cdot\cdot X_n$
from Eq.~\eqref{generators_sym} as acting on $n-1$ qubits
(projected to the respective invariant blocks).
Moreover, Proposition~\ref{prop_obata} implies the desired inclusions into the \DLAs
$\so(2^{n-1}){\oplus}\so(2^{n-1})$ or $\usp(2^{n-1}){\oplus}\usp(2^{n-1})$.
Finally, we apply Lemma~\ref{lem:bipartite:dimensions:combinatorics}(iii)-(iv) to complete the proof.

We also detail a basis-independent argument for the desired inclusions.
Note that $P_j^{t} = (-1)^{\#\mathrm{Y}(P_j)} P_j$ for any Pauli string $P_j$.
For the rest of the proof, let $P_j$ denote a Pauli string such that $iP_j \in \gfree$.
Recall from Lemma~\ref{lem:bipartite:dimensions:combinatorics}(i) that 
$\#\mathrm{Y}(P_j) + \#\mathrm{Z}(P_j)$  is even and
$\#\mathrm{X}(P_j) + \#\mathrm{Y}|_{V_1}(P_j) + \#\mathrm{Z}|_{V_1}(P_j)$ is odd.
We choose $S:= \prod_{v \in V_1} Z_v  \prod_{w \in V_2} Y_w$ and
we obtain $S^{t} = (-1)^{\abs{V_2}} S$ which implies
that $S$ is symmetric or skew-symmetric if $\abs{V_2}$ is even or odd, respectively. 
As $n$ is even, $S$ commutes with $X^{\otimes n}$ and  $\HC_{+}$
and $\HC_{-}$ are also invariant subspaces of $S$.
Moreover,
\begin{align*}
P_j S & = S P_j (-1)^{\#\mathrm{X}+\#\mathrm{Y}|_{V_1}+\#\mathrm{Z}|_{V_2}}\\
& = S P_j (-1)^{\#\mathrm{Y}|_{V_1}+\#\mathrm{Z}|_{V_1}+1 +\#\mathrm{Y}|_{V_1}+\#\mathrm{Z}|_{V_2}}\\
& = S P_j (-1)^{\#\mathrm{Z}+1} = S P_j (-1)^{\#\mathrm{Y}+1},
\end{align*}
which implies that
\begin{align*}
S P_j + P_j^{t} S & = S P_j + (-1)^{\#\mathrm{Y}} P_j S \\
&= S P_j + (-1)^{\#\mathrm{Y}} (-1)^{\#\mathrm{Y}+1} S P_j\\
&= S P_j - S P_j = 0.
\end{align*}
Let $S^{\pm}$ and $P_j^{\pm}$ denote the respective projections to
$\HC_{+}$
and $\HC_{-}$. Note that projection and matrix transposition commutes.
Clearly, also $S^{\pm} P_j^{\pm} + (P_j^{\pm})^{t} S^{\pm} = 0$ holds 
as terms such as $S^{+} P_j^{-}$ and $(P_j^{+})^t S^{-}$ in the expansion of $S P_j + P_j^{t} S=0$ are zero.
Note that $S^{\pm}$ is symmetric iff $S$ is symmetric.
We can apply Prop.~\ref{prop_obata}.
\end{proof}

The proof critically relies on the fact that $n$ is even, and a different free-mixer
\DLA appears if $n$ is odd:
\begin{proposition}[bipartite graph, $n$ odd]\label{thm:free-mixer-bipartite-odd-combinatorics}
Given a connected bipartite graph $G$ 
that is neither
a path graph nor a cycle graph, let $V=V_1 {\uplus} V_2$ 
denote its vertex bipartition and assume that
$\abs{V}=\abs{V_1}+\abs{V_2}=n\geq 4$ is odd.
The free-mixer \DLA for $G$ is given by 
$\gfree\iso\su(2^{n-1})$.
\end{proposition}

\begin{proof}
Without loss of generality, we let $\abs{V_1}$ and $\abs{V_2}$ be even and odd, respectively.
For Hadamard basis states $\ket{p_1\!\cdots p_n}$ with $p_j \in \{+,\! -\}$, let the parity $\pi_j=\pi_j(\ket{p_1\!\cdots p_n}) \in \{e,o\}$ for $j\in \{1, 2\}$
be equal to $e$ iff the number of $v \in V_j$
such that $p_v \in  \{-\}$ is even.
We define $V_{pq}$ for $p,q \in \{e,o\}$ to be the complex space spanned by Hadamard basis states
with $\pi_1=p$ and $\pi_2=q$. We have $\HC_{+} = V_{ee} {\oplus} V_{oo}$ and $\HC_{-} = V_{eo} {\oplus} V_{oe}$.
Let $\mathcal{B} = (\mathcal{B}_{ee}, \mathcal{B}_{oo}, \mathcal{B}_{eo}, \mathcal{B}_{oe})$ describe bases of the Hadamard basis states from
$V_{ee}$,  $V_{oo}$, and $V_{eo}$ and of the negatives of the Hadamard basis states from $V_{oe}$. Recall $X \ket{\pm} = \pm \ket{\pm}$ and $Z \ket{\pm} = \ket{\mp}$.
We define $\ket{\psi^*}:= \pm Z^{\otimes n} \ket{\psi}$ for Hadamard basis states $\ket{\psi}$ where the plus sign is chosen iff $\pi_j(\ket{\psi}) = e$. 
Our objective now is to show that any generator $iH_j$ of $\gfree$ has the block-diagonal form
\begin{equation}\label{eq:bipartite:block}
iH_j = i 
\begin{pmatrix}
        M_j & 0 \\
        0 & -M_j^{t}
\end{pmatrix}
\end{equation}
in the basis $\mathcal{B}$ for suitable $M_j$ depending on $H_j$. This would be imply that $\gfree$
acts on $\HC_{-}$ with the dual of the representation acting on $\HC_{+}$. It would follow that
$\gfree$ is isomorphic to a subalgebra of $\su(2^{n-1})$ and the dimension formula in Lemma~\ref{lem:bipartite:dimensions:combinatorics}(ii)
then shows $\gfree\iso\su(2^{n-1})$.

It remains to verify the form in Eq.~\eqref{eq:bipartite:block} for all generators $iX_v$ for $v\in V$ and
$iZ_vZ_w$ for edges $(v,w)$ with $v\in V_1$ and $w\in V_2$. Let $\ket{\psi}$ denote a Hadamard basis state in $\HC_{+}$.
Note that $X_v \ket{\psi} = f \ket{\psi}$ if and only if $X_v \ket{\psi^*} = - f \ket{\psi^*}$ which implies
that $\bra{\psi} X_v \ket{\psi} = - \bra{\psi^*} X_v \ket{\psi^*}$. For any Hadamard basis $\ket{\phi}$ in $\HC_{+}$ different from 
$\ket{\psi}$, we obtain $\bra{\phi} X_v \ket{\psi} = 0 = - \bra{\psi^*} X_v \ket{\phi^*}$. Hence the matrix $X_v$ in the subspace basis $(\mathcal{B}_{ee}, \mathcal{B}_{oo})$
is the negative transpose of its matrix in the subspace basis $(\mathcal{B}_{eo}, \mathcal{B}_{oe})$.

We continue with $Z_vZ_w$ and we introduce $\ket{\psi_{uv}}:= Z_uZ_v \ket{\psi}$ for any Hadamard basis state in $\HC_{+}$. It follows that $\ket{\psi_{uv}^*}= - Z_uZ_v \ket{\psi^*}$ and
we obtain $\bra{\psi_{uv}} Z_u Z_v \ket{\psi} = - \bra{\psi^*} X_v \ket{\psi_{uv}^*}$. For any Hadamard basis $\ket{\phi}$ in $\HC_{+}$ different from 
$\ket{\psi_{uv}}$, we observe $\bra{\phi} Z_u Z_v \ket{\psi} = 0 = - \bra{\psi^*} Z_u Z_v \ket{\phi^*}$ which verifies the form in Eq.~\eqref{eq:bipartite:block}.
\end{proof}

\section{Proof of Corollary~\ref{cor:variance} from Sec.~\ref{sec:free:implications}\label{proof:variance}}

We now prove Corollary~\ref{cor:variance} which is restated below after some preparations.
The QAOA cost function was specified in Eq.~\eqref{eq:cost} as 
$C(\thv)=\langle \psi(\thv)|H_p|\psi(\thv)\rangle$ 
for $|\psi(\thv)\rangle=U(\thv)|+\rangle^{\otimes n}$, where $|+\rangle^{\otimes n}$ is defined in
Eq.~\eqref{eq:fiduciary}. 
This relies on the unitaries 
$U(\thv)$ in the circuit from Eq.~\eqref{eq:circuit:general}
and $U(\thv)$
depends on 
the parameters
$\thv$ with real entries $\theta_{\vth}$.
For simplicity, the entries $\theta_{\vth}$
of $\thv$ are now indexed by
the numbers $\vth$.
Let $\partial_{\vth} C(\thv)$ denote the partial derivative of $C(\thv)$ with respect to the $\vth$-th
parameter $\theta_{\vth}$ in $\thv$. 

More concretely, the parameters
$\thv$ could be given by the real values $\theta_{\ell u},\theta_{\ell v w} \in [-\pi,\pi]$ 
in order to establish
the multi-angle QAOA unitaries
[as in Eqs.~\eqref{eq:circuit} or \eqref{eq:circuit:general}]
\begin{equation}\label{eq:multiangle:U}
U(\thv) = \prod_{\ell=1}^L \hspace{-0.5mm}
\Bigg[
\prod_{u\in V}
e^{-i \theta_{\ell u}  X_u } \hspace{-0.4mm} 
\Bigg] 
\hspace{-1.0mm} 
\Bigg[
\prod_{\{v,w\}\in E}
e^{-i \theta_{\ell v w} Z_v Z_w } \hspace{-0.4mm} 
\Bigg]
\end{equation}
for $L$ layers and a graph with vertices $V$ and edges $E$.
In addition, we could assume that each $\theta_{\ell u},\theta_{\ell v w}$ is sampled
independently and uniformly from $(-2\pi,2\pi)$.

Returning to a more general setting, 
the parameters $\thv$ are sampled
according to a given distribution $d\thv$ over a chosen parameter domain $\delta_{L}$. This
induces a distribution on the associated unitaries $U(\thv)$.
For a given distribution $\nu$ on a compact Lie group $\exp(\g)$ with \DLA $\g$, we recall
the second-order momentum operator 
\begin{align*}
M_{\nu}&:=\int_{U \in e^{\g}}d\nu(U)\;  U^{\otimes 2}{\otimes} \bar{U}^{\otimes 2}.
\intertext{In particular, we consider the momentum operators}
M_{e^{\g}}&:=\int_{U \in e^{\g}} d\mu_{e^{\g}}(U)\; U^{\otimes 2 }{\otimes} \bar{U}^{\otimes 2}\; \text{ and}\\
M_{L}&:=\int_{\thv\in \delta_L}d\thv\;  [U(\thv)]^{\otimes 2 }{\otimes} [\bar{U}(\thv)]^{\otimes 2 }
\end{align*}
for the Haar measure
$d\mu_{e^{\g}}$ 
on the (compact) Lie group $e^{\g}$
and a general distribution 
$d\thv$ on the parameters $\thv$
over a parameter domain $\delta_{L}$.
This leads to the following
\begin{definition}[Approximate unitary $2$-design]\label{def:approx:design}
A distribution $\nu$ on a compact Lie group $\exp(\g)$ with \DLA $\g$ is an $\varepsilon$-approximate unitary $2$-design
if
\begin{equation*}
\norm{M_\nu - M_{e^{\g}}}_{\infty} \leq \varepsilon.
\end{equation*}
\end{definition}
Here, $\norm{\cdot}_{\infty}$ denotes the Schatten $\infty$-norm (or operator norm) which is given by
the largest singular value of its argument. We are particularly interested under which assumptions
the condition $\norm{M_L - M_{e^{\g}}}_{\infty} \leq \varepsilon$ holds. 
The distribution of unitaries
in a multi-angle QAOA circuit has to be an approximate unitary $2$-design in the following 
result (which we restate from Sec.~\ref{sec:free:implications}):
\begin{repcorollary}{cor:variance}
For the free ansatz, consider
any archetypal graph from
Def.~\ref{def_other_graphs} with $n> 3$ vertices and $\abs{E}$ edges.
Recall the QAOA cost function 
$C(\thv)$
from Eq.~\eqref{eq:cost}
and its partial derivative
$\partial_{\vth} C(\thv)$ with respect to the $\vth$-th
parameter $\theta_{\vth}$ in $\thv$. Assume that the multi-angle QAOA circuit
has enough layers such that the distribution of unitaries is an $\varepsilon$-approximate unitary $2$-design.
Then, the expectation value
of the partial derivatives is $E_{\thv}[\partial_{\vth} C(\thv)]=0$ and their variance is given by (with $d=2^n$)
\begin{equation*}
    \Var_{\thv}[\partial_{\vth} C(\thv)] = {4d^2|E|}/[(d^2{-}4)(d{+}2)] \leq 4n^2/2^n.
\end{equation*}
\end{repcorollary}

We start the proof by first collecting relevant notation and results, while partially relying on Appendix~\ref{app:free:preliminaries}.
Recall that $X^{\otimes n}$ commutes with the free-mixer \DLA $\gfree$ and all its generators in $\GC_{\free}$
[see Lemma~\ref{app:prop:center}(a)] and that the vector space $\HC=(\mbb{C}^2)\tn$ of dimension $d=2^n$ splits
into the invariant subspaces $\HC=\HC_+\oplus\HC_-$ where
the $+1$ and $-1$ eigenspaces
$\HC_\pm = \{ \ket{\psi} \in \HC\,|\, X\tn \ket{\psi} = \pm \ket{\psi }\}$ of $X^{\otimes n}$
are spanned by all 
Hadamard basis states $\ket{b_1}\cdots\ket{b_n}$ with $b_j \in \{+, -\}$ and respectively an even or odd  number of minus
signs (i.e.\ $b_j = -$). Clearly, $\HC_{+}$ and $\HC_{-}$ have the same dimension $d_{+} = d_{-} = 2^{n-1}$
and the fiduciary state $\ket{+}^{\otimes n}$ [see Eq.~\eqref{eq:fiduciary}] is contained in $\HC_{+}$.

The corresponding $d\times d$ projection matrices $P_{\pm}= (I^{\otimes n} {\pm} X^{\otimes n})/2$ can be 
represented as $P_{\pm}=Q_\pm\ad Q_\pm$ using the rectangular reduction matrices $Q_\pm \in \C^{d_{\pm}\times d}$ where
the columns of $Q_\pm\ad \in \C^{d\times d_{\pm}}$ are given by the associated 
(orthonormalized) Hadamard basis states.
Note that $Q_\pm Q_\pm\ad = I^{\otimes (n-1)}$.
For a Pauli string (or any suitable matrix) $A$, we define the reduced operators
\begin{equation}\label{eq:red}
    A^{(\pm)}:=Q_\pm A Q_\pm\ad \in \C^{d_{\pm}\times d_{\pm}}.
\end{equation}

\begin{lemma}\label{lemma_stepping}
For a Pauli string $A$ on $n$ qubits
that is different from the identity and $X\tn$ and that
commutes with $X\tn$, the reduced operators $A^{(\pm)}$
are unitarily equivalent to a Pauli string on $n{-}1$ qubits.
In particular, $\Tr[A^{(\pm)}] = 0$ and
$\Tr[A^{(\pm)} A^{(\pm)}]=2^{n-1}$.
\end{lemma}

\begin{proof}
Clearly, $A$ and $A^{(\pm)}$ are hermitian. The proof proceeds in two steps verifying that (i)
the eigenvalues of $A^{(\pm)}$ are $\pm 1$ and that (ii) their multiplicity
is equal to $2^{n-2}$ (or, equivalently, $\Tr[A^{(\pm)}]=0$). We compute
\begin{align*}
    A^{(\pm)} A^{(\pm)} &= Q_{\pm} A Q_{\pm} \ad Q_{\pm} A Q_{\pm}\ad 
    = Q_{\pm} A P_{\pm} A Q_{\pm}\ad \\
    &= ( Q_{\pm} A A Q_{\pm}\ad \pm  Q_{\pm} A  A X\tn Q_{\pm}\ad)/2 \\
    &= ( Q_{\pm}Q_{\pm}\ad \pm  Q_{\pm} X\tn Q_{\pm}\ad)/2 \\
    &= ( I^{\otimes (n-1)} + I^{\otimes (n-1)})/2 = I^{\otimes (n-1)},
\end{align*}
where we have used that $A$ commutes with $X\tn$, that the square $A^2$ of the Pauli string $A$ is equal to $I^{\otimes n}$, 
that $Q_\pm Q_\pm\ad = I^{\otimes (n-1)}$, and that $Q_{\pm} X\tn Q_{\pm}\ad= \pm I^{\otimes (n-1)}$
as the columns of $Q_{\pm}\ad$ contain the eigenvectors of $X\tn$ to the eigenvalues $\pm 1$. Thus
$A^{(\pm)}$ has only $\pm 1$ as eigenvalues. 
In particular, $\Tr[A^{(\pm)} A^{(\pm)}]=2^{n-1}$.
We complete the proof by computing the trace (for $A \notin \{I\tn, X\tn\}$)
\begin{align*}
    \Tr[A^{(\pm)}] &= \Tr[ Q_{\pm} A Q_{\pm}\ad]
    =\Tr[A P_{\pm}]\nonumber \\
    &=(\Tr[A]+\Tr[A X\tn])/2 = 0.  \qedhere
\end{align*}
\end{proof}

The proof of Corollary~\ref{cor:variance} is based on a general result from \cite{larocca2021diagnosing,brown2010random}
which describes the variance of the cost function of a randomly initialized parametrized quantum circuit
that is deep enough and observes certain controllability conditions.
We recall a special case of this result:

\begin{fact}[special case of Thm.~2 in \cite{larocca2021diagnosing}]\label{thm_diagnosing}
Let the action of a \DLA $\g = \lie{H_1,\ldots,H_m}$ induce 
a Hilbert space  splitting $\HC=\HC_{1} \oplus \HC_{2}$ 
into invariant subspaces such that $\HC_{1}$ is irreducible and has
dimension $d_{1}=\dim(\HC_{1})$
and $P_{1} \g P_{1}=\su(d_{1})$ holds for the respective projector $P_{1}$.
We consider a variational quantum circuit [see Eq.~\eqref{eq:circuit:general}]
\begin{equation*}
U(\thv)=\prod_{\ell=1}^L \prod_{k=1}^{m} e^{-i \theta_{\ell k} H_{k}} = \prod_{\vth} e^{-i \theta_{\vth} H_{\vth}}
\end{equation*}
with $L$ layers and the notation $\vth = \vth(\ell,k) = (\ell{-}1)m + k$ and $H_{\vth(\ell,k)}= H_k$.
The corresponding cost function $\tilde{C}(\thv)=\Tr[U(\thv)\rho U(\thv)\ad O]$ depends on a 
hermitian measurement operator $O$ and a fiduciary mixed state $\rho$ with 
$\rho = P_{1} \rho P_{1}$.
Suppose that the number of layers $L$ in $U(\thv)$ is large enough  
so that the distribution of $U(\thv)$ is an $\varepsilon$-approximate $2$-design.
Then, the variance
of the partial derivative with respect to the parameter $\th_\mu$ of $\tilde{C}(\thv)$ is
\begin{equation}\label{eq_diagnosing}
\Var_{\thv}[\partial_{\vth} \tilde{C}(\thv)] =\frac{2d_1}{(d_1^2{-}1)^2}\, \Delta[H_{\vth}^{(1)}]\,\Delta[O^{(1)}]\,\Delta[\rho^{(1)}]   
\end{equation}
where the reduced operator $B^{(1)}$ of a matrix $B$ acting on $\HC$ is defined 
using the orthonormalized eigenvectors of $P_1$ to the eigenvalue one
as in and before Eq.~\eqref{eq:red}
and 
\begin{equation*}
\Delta(B)=\Tr[B^2]-\Tr[B]^2/\dim(\HC) = \dim(\HC) \Var[{\rm Eig}(B)],
\end{equation*}
where ${\rm Eig}(B)$ denotes the set of eigenvalues of $B$.
\end{fact}

We set $O=H_p$,
$\rho = \ketbra{\psi}{\psi}$ for $\ket{\psi}=\ket{+}^{\otimes n}$,
and $d_{1}=d_{+}=d/2$  in Fact~\ref{thm_diagnosing} and recover 
the setting of Corollary~\ref{cor:variance} with
$\tilde{C}(\thv) = \langle \psi(\thv)|H_p|\psi(\thv)\rangle = C(\thv)$.
We can now directly prove Corollary~\ref{cor:variance} by applying Eq.~\eqref{eq_diagnosing} and the following
Lemma~\ref{lem:delta}
as all the generators $H_{\vth}$ 
in the multi-angle ansatz
are Pauli strings that are different from $I\tn$ and $X\tn$
and that commute with $X\tn$.
\begin{lemma}\label{lem:delta}
In the setting of Corollary~\ref{cor:variance}, we have $\abs{E}$ edges and $d_{+}=d/2=2^{n-1}$ for $n$ qubits.
Let $A$ denote any Pauli string that is different from the identity and $X\tn$
and that commutes with $X\tn$.
Moreover, $H_p$ is the  parent Hamiltonian from
Eq.~\eqref{eq:prob-Ham} and $\rho =\ketbra{\psi}{\psi}$ is the projector for $\ket{\psi}=\ket{+}^{\otimes n}$ from Eq.~\eqref{eq:fiduciary}.
We obtain
\begin{align*}
\Delta[A^{(+)}] &= d_{+} = d/2,\\
\Delta[H_p^{(+)}] &= \abs{E}\, d_{+} = \abs{E}\, d/2,\\
\Delta[\rho^{(+)}] &= (d_{+}{-}1)/d_{+} = (d{-}2)/d.
\end{align*}
\end{lemma}

\begin{proof}
Applying Lemma~\ref{lemma_stepping}, we compute
\begin{equation*}
\Delta(A^{(+)}) =
\Tr[A^{(+)}A^{(+)}] -\Tr[A^{(+)}]^2/d_{+} = d_{+} = d/2.
\end{equation*}
Recall the notation $Q_{+}$ from Eq.~\eqref{eq:red} and we obtain 
\begin{align}
&\Tr[H_p^{(+)}] = \Tr[Q_{+} H_p Q_{+}^{\dagger}] = \Tr(H_p P_{+}) \nonumber \\
&= \tfrac{1}{2} \sum_{\{k,\ell\}\in E} \left[ \Tr( Z_k Z_{\ell} I^{\otimes n} ) +    \Tr( Z_k Z_{\ell} X^{\otimes n} )     \right] = 0
\label{eq:pauli:zero}
\end{align}
as the trace of a Pauli string (different from the identity) is always zero. 
Since $\Tr[H_p^{(+)}] = 0$, we similarly get
\begin{align*}
\Delta[H_p^{(+)}] & = \Tr[H_p^{(+)}H_p^{(+)}] = \Tr(H_p H_p P_{+}) \\
& =  \sum_{\{u,v\}\in E}\; \sum_{\{u',v'\}\in E} \Tr( Z_{u} Z_{v} Z_{u'} Z_{v'} P_{+} )
\intertext{where all terms not observing $u=u'$ and $c = v'$ do not contribute as in Eq.~\eqref{eq:pauli:zero}, and 
it follows that}
\Delta[H_p^{(+)}] & =  \sum_{\{u,v\}\in E} P_{+} = \abs{E}\, d_{+} = \abs{E}\, d/2.
\end{align*}
Note $\rho^2=\rho$ and $P_{+} \rho P_{+} = \rho$. We compute 
$\Tr[\rho^{(+)}] =  \Tr(\rho P_{+} ) =  \Tr(\rho P_{+} P_{+}) =  \Tr(P_{+} \rho P_{+} ) = \Tr(\rho) = 1$. Also,
$\Tr[\rho^{(+)} \rho^{(+)}] = \Tr(\rho P_{+} \rho P_{+} ) =\Tr(\rho^2 ) = \Tr(\rho )  =1$,
which implies the stated formula for $\Delta[\rho^{(+)}]$.
\end{proof}

\section{Proofs for Section~\ref{sec:standard-ansatz-partial-results}}
This appendix collects proofs for Section~\ref{sec:standard-ansatz-partial-results}.
In particular, Appendix~\ref{proof:lemma:std:path:2} provides the proof of 
Lemma~\ref{lemma:std:path:2}, Proposition~\ref{prop:std:path} is verified
in Appendix~\ref{proof:prop:std:path}, and Appendix~\ref{proof:prop:std:complete}
proves Proposition~\ref{prop:std:complete}.

\subsection{Proof of Lemma~\ref{lemma:std:path:2}\label{proof:lemma:std:path:2}}
In the following,
we assume $n>5$ as
one can directly verify that $\gnat\iso \uu(n)$ for all $n \in \{2,3,4,5\}$.
By analyzing suitable commutators, one concludes that the semisimple part $\mathfrak{s}_{\mathrm{nat}}=[\gnat,\gnat]$ 
of $\gnat$ has dimension $n^2{-}1$, while the center $\cent(\gnat)$
is one dimensional and it is
spanned by the single element
\begin{equation*}
z_{\mathrm{nat}} = 
\begin{cases}
       \phantom{ -i X_{\bar{n}+1} +}\hspace{2.5pt} \sum_{o=1}^{\bar{n}} iP_{oo}^{YY}{+} iP_{oo}^{ZZ} & \text{for $n$ even}, \\
        -i X_{\bar{n}+1} + \sum_{o=1}^{\bar{n}} iP_{oo}^{YY}{+}iP_{oo}^{ZZ} & \text{for $n$ odd}.
    \end{cases}
\end{equation*}
In particular, using the indices $o\in \{1,\ldots,\bar{n}\}$, $\tilde{o}\in \{1,\ldots,\bar{n}{-}1\}$, and
$p,q \in \{1,\ldots,n\}$
with $p < n+1- q$ and $p\neq q$,
a basis for $\mathfrak{s}_{\mathrm{nat}}$ is given by
\begin{subequations}
\label{eq:snat:path}
\begin{align}
&iX_{\bar{n}+1}{-}iP_{\bar{n}\bar{n}}^{YY} \;\text{ if $n$ is odd},\label{eq:snat:path:a}\\
&iP_{oo}^{YY}{-}iP_{oo}^{ZZ},\; iP_{\tilde{o}{+}1\, \tilde{o}{+}1}^{ZZ}{-}iP_{\tilde{o} \tilde{o}}^{YY},\label{eq:snat:path:b} \\
&iX_{o}{+}iX_{n+1-o},\;
iP_{oo}^{YZ}{+} iP_{oo}^{ZY}, \label{eq:snat:path:c}\\
&iP_{pq}^{YY}{+} iP_{qp}^{YY},\;
iP_{pq}^{ZZ}{+} iP_{qp}^{ZZ},\;
iP_{pq}^{YZ}{+} iP_{qp}^{ZY}. \label{eq:snat:path:d}
\end{align}
\end{subequations}

In order to prove $\gnat \iso \uu(n)$,
the induction hypothesis is that the elements in Eq.~\eqref{eq:gnat:path} with $o>1$, $p>1$, and $q>1$
generate the \DLA $\gk \iso \uu(n{-}2)$. Its center $\cent(\gk)$ is spanned by 
the element
\begin{equation*}
\tilde{z}_{\mathrm{nat}} = 
\begin{cases}
       \phantom{ -i X_{\bar{n}+1} +}\hspace{2.5pt} \sum_{o=2}^{\bar{n}} iP_{oo}^{YY}{+}iP_{oo}^{ZZ} & \text{for $n$ even}, \\
        -i X_{\bar{n}+1} + \sum_{o=2}^{\bar{n}} iP_{oo}^{YY}{+}iP_{oo}^{ZZ} & \text{for $n$ odd},
    \end{cases}
\end{equation*}
and we denote the semisimple part of $\gk$ by
$\tilde{\gk} \iso \su(d{-}2)$. Note that a basis
of $\tilde{\gk}$ is given by all the elements in Eq.~\eqref{eq:snat:path}
with $o > 1$, $\tilde{o}>1$, $p>1$, and $q>1$, and possibly the one from Eq.~\eqref{eq:snat:path:a}.

First, we prove that $\mathfrak{s}_{\mathrm{nat}}$ is simple. Recall that
$\mathfrak{s}_{\mathrm{nat}}$ is semisimple and that all its ideals are semisimple.
We compute the ideal $\mathfrak{i}$ in $\mathfrak{s}_{\mathrm{nat}}$
that contains the simple \DLA $\tilde{\gk}$ in order to show that $\mathfrak{i}$ is simple
and $\mathfrak{i}=\mathfrak{s}_{\mathrm{nat}}$.
To the contrary, we will assume that there exists a complementary ideal $\mathfrak{j}$
with $\mathfrak{i} \oplus \mathfrak{j} \subseteq \mathfrak{s}_{\mathrm{nat}}$.
But then $[\mathfrak{i}, \mathfrak{j}]=0$ holds. Thus $[g,\mathfrak{i}] =0$ holds
for all elements $g\in \mathfrak{s}_{\mathrm{nat}}$ outside of $\mathfrak{i}$.
We systematically verify that all elements in Eq.~\eqref{eq:snat:path}
with $o = 1$, $\tilde{o}=1$, $p =1$, or $q=1$ are also contained in $\mathfrak{i}$.
Each element in Eq.~\eqref{eq:snat:path:d} with $p=1$ or $q=1$
does not commute with the element
\begin{equation*}
\begin{aligned}
&iX_r{+}iX_{n+1-r} \\
&iP_{\bar{n}{+}1\, \bar{n}}^{ZZ}{+} iP_{\bar{n}\, \bar{n}{+}1}^{ZZ}
\end{aligned}
\qquad
\begin{aligned}
&\text{if $n$ is even or $r\neq \bar{n}{+}1$ and}\\
&\text{if $n$ is odd and $r= \bar{n}{+}1$}
\end{aligned}
\end{equation*}
from $\tilde{\gk}$ where $1\neq r \in \{p,q\}$.
The conditions in the previous equation are sufficient for the element to not commute
(but not necessary).
Thus all elements in Eq.~\eqref{eq:snat:path:d} need to be contained in $\mathfrak{i}$.
Moreover, the elements from Eqs.~\eqref{eq:snat:path:b}-\eqref{eq:snat:path:c} with $o=1$ and $\tilde{o}=1$ do not commute with $iP_{12}^{ZZ}{+} iP_{21}^{ZZ}$,
which implies that all elements from Eqs.~\eqref{eq:snat:path:b}-\eqref{eq:snat:path:c} are contained in $\mathfrak{i}$.
We have shown that $\mathfrak{i}=\mathfrak{s}_{\mathrm{nat}}$.
In each step, our proof technique also verifies that newly added elements 
do not commute with all elements already contained in $\mathfrak{i}$. 
While starting from the simple $\tilde{\gk}$, this implies that
$\mathfrak{i}$ never splits into two (or more) simple ideals.
Thus $\mathfrak{i}=\mathfrak{s}_{\mathrm{nat}}$ is simple.

An abelian subalgebra of $\mathfrak{s}_{\mathrm{nat}}$ is spanned by
the $n{-}1$ elements from Eqs.~\eqref{eq:snat:path:a}-\eqref{eq:snat:path:b}
and it is maximal abelian. Otherwise, the rank 
of $\mathfrak{s}_{\mathrm{nat}}$ would be at least $n$ 
while its dimension is equal to $n^2{-}1$. But this conflicts with
$\mathfrak{s}_{\mathrm{nat}}$ being simple.
Indeed, all compact simple \DLAs are
$\su(m{+}1)$, $\so(2m{+}1)$, $\usp(m)$, $\so(2m)$, $\mathfrak{g}_2$,
$\mathfrak{f}_4$, $\mathfrak{e}_6$, $\mathfrak{e}_7$, $\mathfrak{e}_8$
with respective ranks $m$, $m$, $m$, $m$, $2$, $4$, $6$, $7$, $8$
and dimensions $m^2{+}2m$, $2m^2{+}m$, $2m^2{+}m$, $2m^2{-}m$, $14$, $52$, $78$, $133$, $248$.
Clearly, a rank of $m\geq n> 5$ is not possible for $\mathfrak{s}_{\mathrm{nat}}$.
For $n>5$,
the rank of $m=n{-}1$ implies that $\mathfrak{s}_{\mathrm{nat}}\iso \su(n)$ and $\gnat \iso \uu(n)$.
This completes the proof by induction.

\subsection{Proof of Proposition~\ref{prop:std:path}\label{proof:prop:std:path}}

We again
can directly verify that the statement holds for $n\in\{2,3,4,5\}$ and
we assume now that $n>5$. We shortly recall the standard generators 
\begin{equation*}
g_p:=i\sum_{w=1}^{n-1} Z_{w} Z_{w+1} \;\text{ and }\;
g_m:=i\sum_{v=1}^{n} X_v
\end{equation*}
for the path graph and introduce the elements 
\begin{equation*}
\tilde{g}_p:=iZ_{1} Z_{2}{+}iZ_{n-1} Z_{n} \;\text{ and }\;
\tilde{g}_m:=iX_1{+}iX_n.
\end{equation*}
It is straightforward to verify that
\begin{align*}
\tilde{g}_m & = (-[g_p,[g_p,[g_p,[g_p,g_m]]]]-16 [g_p,[g_p,g_m]])/48,\\
\tilde{g}_p  & = - [\tilde{g}_m,[\tilde{g}_m,g_p]]/4.
\end{align*}
Consequently, $g_p-\tilde{g}_p$ and $g_m-\tilde{g}_m$ 
are contained in $\gstd$ and we can assume as induction hypothesis that
these two elements generate $\gk \iso \uu(n{-}2)$. Following Eq.~\eqref{eq:gnat:path},
clearly $iP_{22}^{YY} \in \gk$ and one obtains that
\begin{alignat*}{3}
g_1&:= iP_{11}^{ZZ} &&= [g_p,[g_p, iP_{22}^{YY}]]/8+ iP_{22}^{YY},\\
g_2&:= iP_{11}^{YZ} {+} iP_{11}^{ZY} &&= [\tilde{g}_m,iP_{11}^{ZZ}]/2,\\
g_3&:= iP_{11}^{YY} &&= [\tilde{g}_m, iP_{11}^{YZ} {+} iP_{11}^{ZY}]/4 + iP_{11}^{ZZ}.
\end{alignat*}
Consequently, $\tilde{g}_m$, $g_1$, $g_2$, and $g_3$ generate the \DLA $\g_2 \iso\uu(2)$
such that all elements of $\g_2$ commute with all elements of $\gk$.
We summarize
\begin{equation*}
\uu(2) \oplus \uu(n{-}2) \iso \g_2 \oplus \gk \subsetneq \gstd \subseteq \gnat \iso \uu(n),
\end{equation*}
but $\uu(2) {\oplus} \uu(n{-}2)$
is a maximal subalgebra of $\uu(n)$ \cite{borel1949,GG78}.
Thus
the induction step is complete.

\subsection{Proof of Proposition~\ref{prop:std:complete}\label{proof:prop:std:complete}}

For this proof,
let $n_I=n_I(S)$, $n_X =n_X(S)$, $n_Y= n_Y(S)$, and $n_Z=n_Z(S)$
denote respectively the number of $\mathrm{I}$, $\mathrm{X}$, $\mathrm{Y}$, and $\mathrm{Z}$
in a Pauli string $S$. Recall from Table~\ref{tab:free-basis-table} that
$n_Y + n_Z$ is even, $n_X \neq n$, and $n_I \neq n$ for
the Pauli-string basis of $\gfree$ for a complete graph.
In order to determine $\gnat$, we are counting the Pauli strings that
are also invariant under the action of the automorphism group 
$\text{Aut}(K_n)= \mathcal{S}_n$. We assume that the two
Pauli strings $S_1$ and $S_2$ contain the same number of 
$\mathrm{X}$, $\mathrm{Y}$, $\mathrm{Z}$,
and $\mathrm{I}$, i.e., $n_I(S_1)=n_I(S_2)$, $n_X(S_1)=n_X(S_2)$, $n_Y(S_1)=n_Y(S_2)$,
and $n_Z(S_1)=n_Z(S_2)$. Then there clearly exits an automorphism $\sigma\in\mathcal{S}_n$ of the complete graph
that maps $S_1$ to $S_2$, in the sense that $\zeta[\sigma] S_1 = S_2 \zeta[\sigma]$, where
$\zeta$ maps permutations to $n$-qubit matrices and has been defined in Sec.~\ref{SEC:STD}.
A weak $4$-composition of $n$ is an ordered quadruple of nonnegative integers $(n_I, n_X, n_Y, n_Z)$
with $n_I + n_X + n_Y + n_Z = n$. Thus a basis of $\gnat$ is given by all possible sums
of Pauli strings (with identity coefficients) such that 
the Pauli strings in each sum correspond to a fixed 
weak 4-composition $(n_I, n_X, n_Y, n_Z)$ of $n$ while observing the
additional conditions that $n_Y + n_Z$ is even and neither $n_I$ nor $n_X$ is equal to $n$.
Now all that remains is to count these allowed weak 4-compositions.
Ignoring all restrictions, we have $\binom{n+3}{3}$ 
weak 4-compositions
(as is easily verified by the stars-and-bars method \cite[pp.~17--18]{Stanley2012}).

If $n$ is odd, then $n_I + n_X$ and $n_Y + n_Z$ have opposite parity, so the map
$(n_I, n_X, n_Y, n_Z)\mapsto(n_Z, n_Y, n_X, n_I)$ establishes a bijection between weak 4-compositions
that satisfy the $\ZTWO$ symmetry and that do not. So, there are exactly
$\frac{1}{2}\binom{n+3}{3}$ weak 4-compositions that satisfy the parity condition.
Removing the weak 4-compositions $(n,0,0,0)$ and $(0,n,0,0)$, we obtain the final answer
of $\dim(\gnat) = \frac{1}{2}\binom{n+3}{3} - 2$ if $n$ is odd.

If $n$ is even, then the map $(n_I, n_X, n_Y, n_Z)\mapsto(n_I+1, n_X, n_Y, n_Z-1)$ establishes
a bijection from the set of weak 4-compositions with $n_Y + n_Z$ even and $n_Z\ge 1$ to the set of 
weak 4-compositions
with $n_Y + n_Z$ odd and $n_I\ge 1$. The only weak $4$-compositions that have not been accounted for are
the ones with $n_Y + n_Z$ even and $n_Z = 0$ and the ones with $n_Y + n_Z$ odd and
$n_I = 0$.

The former cases can be thought of as weak $3$-compositions of $n$ of the form
$(n_I, n_X, n_Y)$ with $n_Y$ even; all of these are allowed, as long as neither $n_I$ nor $n_X$
equals $n$. Letting $n_Y$ range over all even values from $0$ to $n$, we see that there are
$(n{+}1)+(n{-}1)+\cdots+3+1 = (n/2{+}1)^2$ possibilities.
The latter weak $4$-compositions can
be thought of as weak $3$-compositions of $n$ of the form $(n_X, n_Y, n_Z)$ with $n_X$ odd (since $n$ is even);
all of these are not allowed. Letting $n_X$ range over all odd values from $1$ to $n{-}1$, we see that
there are $n+(n{-}2)+\cdots+4+2 = n(n/2{+}1)/2$ possibilities.

The total of the
two quantities from the last paragraph is $\binom{n+2}{2}$, which is exactly the number of weak $4$-compositions
missing from the bijection
established for even $n$. 
Putting this all together, we conclude that the number of
allowed weak $4$-compositions is
\begin{align*}
    \dim(\gnat) &= [(\tfrac{n}{2}{+}1)^2 - 2] + \frac{1}{2}[\tbinom{n+3}{3} - \tbinom{n+2}{2}] \\
    &= \frac{1}{2}\tbinom{n+3}{3} + \frac{n}{4} - \frac{3}{2}
\end{align*}
if $n$ is even, which completes the proof.

\section{Proofs for Section~\ref{sec:character:computations}}
This appendix details the proofs for Section~\ref{sec:character:computations}.
In particular, Appendix~\ref{proof:lem:ZQ-inv-subspace} provides the proof of 
Lemma~\ref{lem:ZQ-inv-subspace}, Proposition~\ref{thm:n-odd-matching-subspaces} is shown
in Appendix~\ref{proof:thm:n-odd-matching-subspaces}, Appendix~\ref{proof:lem:character-formula}
verifies Lemma~\ref{lem:character-formula}, and finally the proof of Proposition~\ref{prop:trivial:multiplicity}
is given in Appendix~\ref{proof:prop:trivial:multiplicity}.

\subsection{Proof of Lemma~\ref{lem:ZQ-inv-subspace}\label{proof:lem:ZQ-inv-subspace}}
Notice that $X_v$ anticommutes with $Z^{\otimes n}$ for each qubit $v$, which means 
hat any mixer Hamiltonian created from terms with $X_v$ also anticommutes with $Z^{\otimes n}$.
However, $Z_wZ_{\tilde{w}}$ commutes with $Z^{\otimes n}$ for any two qubits $w$ and $\tilde{w}$ and
any Hamiltonian created by adding terms with $Z_wZ_{\tilde{w}}$
also commutes with $Z^{\otimes n}$. Consequently, any nested commutator of such Hamiltonians
commutes with $Z^{\otimes n}$ if it has an even number of mixer Hamiltonian terms or
anticommutes with $Z^{\otimes n}$ if it has an odd number of mixer Hamiltonian terms.
All of these nested commutators span $\g$. By definition, $g \ket{\psi}\in W$
for $\ket{\psi}\in W$ and a nested commutator $g \in \g$. Moreover, $Z^{\otimes n}\ket{\psi}$ is an arbitrary element
in $\left(Z^{\otimes n}\right)W$. Since either commutes or anticommutes with $Z^{\otimes n}$, it follows that
$g\left(Z^{\otimes n}\ket{\psi}\right) = \pm Z^{\otimes n}(g\ket{\psi})\in\left(Z^{\otimes n}\right)W$.
Hence $\left(Z^{\otimes n}\right)W$ is also an invariant subspace of $\g$.

\subsection{Proof of Proposition~\ref{thm:n-odd-matching-subspaces}\label{proof:thm:n-odd-matching-subspaces}}

As $Z\ket{+} = \ket{-}$ and $Z\ket{-} = \ket{+}$ and $n$ is odd,
$Z^{\otimes n}$ maps states in $\HC_{+}$ to states in $\HC_{-}$ and vice versa
via the preceding discussion. Alternatively, we notice that 
$Z^{\otimes n}$ anticommutes with $X^{\otimes n}$ if $n$ is odd, so it must map a 
$\pm 1$ eigenstate of $X^{\otimes n}$ to a $\mp 1$ eigenstate.

The symmetries of $\g$ induce a subspace decomposition of $\HC_{+}$ and $\HC_{-}$.
In particular, any invariant subspace $V_j$ within $\HC_{+}$ will lead to a corresponding invariant
subspace $\left(Z^{\otimes n}\right)V_j$ within $\HC_{-}$ due to Lemma \ref{lem:ZQ-inv-subspace}
and the preceding paragraph showing that $Z^{\otimes n}$ swaps $\HC_{+}$ and $\HC_{-}$
(as $n$ is odd). Conversely, any invariant subspace within $\HC_{-}$ will
lead to a corresponding invariant subspace within $\HC_{+}$ for the same reason. It follows
that $\HC_{+}$ and $\HC_{-}$ must have matching decompositions into irreducible, invariant subspaces,
with the correspondence is induced by multiplication with $Z^{\otimes n}$, exactly as desired.

\subsection{Proof of Lemma~\ref{lem:character-formula}\label{proof:lem:character-formula}}

An automorphism $\sg_2 \in \Aut \subseteq \mathcal{S}_n$
acts on a basis state $\ket{x}=\ket{x_1}\otimes\!\cdot\!\cdot\otimes\ket{x_n}$ via (see Section~\ref{SEC:STD})
\begin{equation*}
\sigma_2\cdot \ket{x}=\ket*{x_{\sg_2^{-1}(1)},\!..,x_{\sg_2^{-1}(n)}}
= \ket*{x_{\sg_2^{-1}(1)}}{\otimes}\!\cdot\!\cdot{\otimes}\ket*{x_{\sg_2^{-1}(n)}}.
\end{equation*}
This is described using the map
$\Upsilon(\sg_1,\sg_2)=\vartheta[\sg_1] \zeta[\sg_2]$
from the permutation group $\Gnatr=\mathcal{S}_2 \times \Aut$ to 
the group of natural symmetries
$\Gnat = \ZTWO\times \zeta[\Aut]$
where the form of
$\zeta[\sg_2]$ is detailed in Section~\ref{SEC:STD}
and $\vartheta[\sg_1]$ is given in Eq.~\eqref{eq:z2:mat}.
For the proof, we compute the character $\chi_{\text{nat}}(\sg_1,\sg_2)$ from Eq.~\eqref{eq:chi:hat}.
Based on the statement of this lemma,
we need to consider three cases: (i)~For $\sg_1 =\symone$, $\chi_{\text{nat}}(\sg_1,\sg_2) = 2^{c(\sg_2)}$.
(ii)~If $\sg_1 = (1,2)$ and every cycle in the decomposition in Eq.~\eqref{eq:cycle:decomp} for
$\sigma=\sg_2$ has even length, then $\chi_{\text{nat}}(\sg_1,\sg_2) = 2^{c(\sg_2)}$.
(iii)~If $\sg_1 = (1,2)$ and there exists a cycle of odd length in the decomposition of Eq.~\eqref{eq:cycle:decomp},
then $\chi_{\text{nat}}(\sg_1,\sg_2) =0$.

Recall that $I^{\tn}$ and $X^{\tn}$ from Eq.~\eqref{eq:z2:mat}
as well as $\zeta[\sg_2]$ act as permutation matrices on $\C^{d\times d}$.
Thus the same applies to $\Upsilon(\sg_1,\sg_2)$ and the number of fix points in the set of 
basis states determines the value of the trace as non-invariant basis states do not contribute.
Given a computational basis state $\ket{x}$, we observe that
\begin{subequations}
\begin{align}
    \Upsilon(\symone,\sg_2) \ket{x} &= \ket*{x_{\sg_2^{-1}(1)},\!..,x_{\sg_2^{-1}(n)}}, \label{eq:action:case:a}\\
    \Upsilon((1,2),\sg_2) \ket{x} &= \ket*{\neg x_{\sg_2^{-1}(1)},\!..,\neg x_{\sg_2^{-1}(n)}}, \label{eq:action:case:b}
\end{align}
\end{subequations}
where $\neg x_j$ again denotes the negation of $x_j$ by swapping the basis states $\ket{0}$ and $\ket{1}$.

\begin{figure*}[t]
\includegraphics{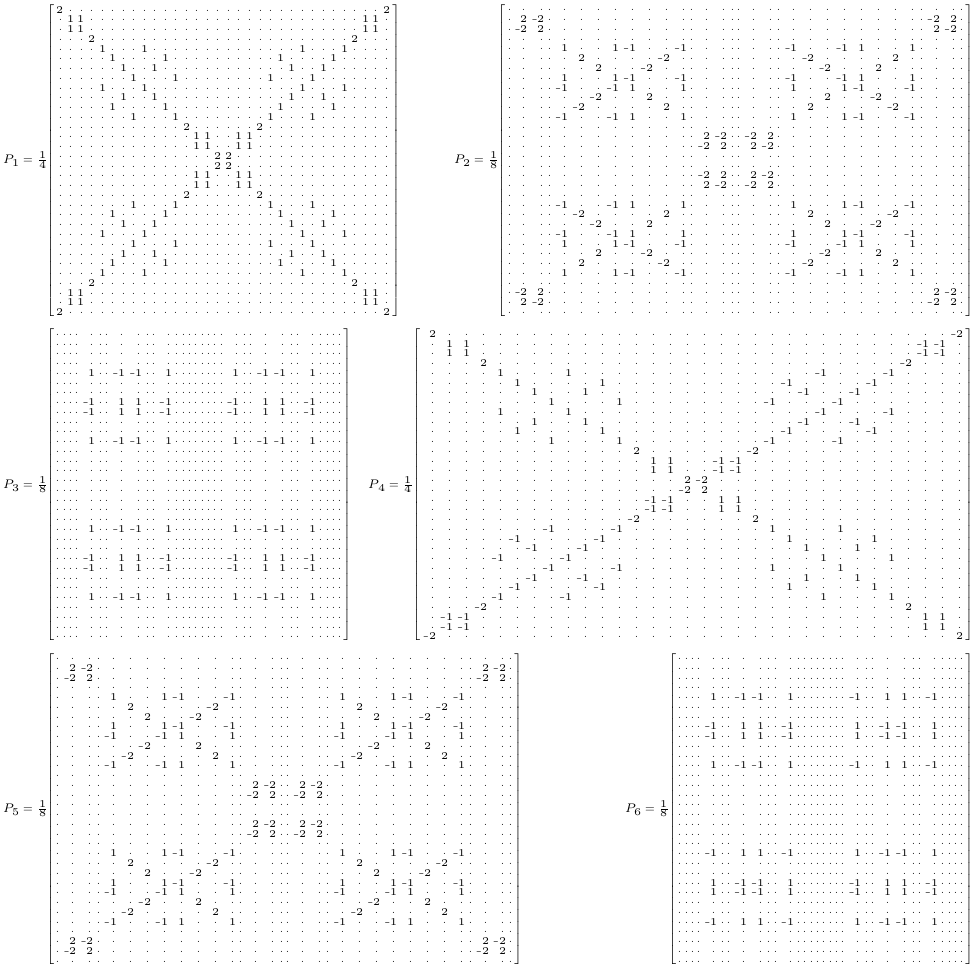}
\caption{\textbf{Explicit matrices for the projectors $P_j$ from Fig.~\ref{fig:decomposition:house}.}
Zeros are replaced by dots.
\label{fig:matrices:house}}
\end{figure*}

The case~(i) corresponds to Eq.~\eqref{eq:action:case:a} and $\ket{x}$ is fixed if and only if 
$x_j = x_{\sg_2^{-1}(j)}$ for each $1\leq j \leq n$. In order words, $x$ needs to be constant on each cycle 
in the decomposition of Eq.~\eqref{eq:cycle:decomp}. Consequently, the result $\chi_{\text{nat}}(\sg_1,\sg_2) = 2^{c(\sg_2)}$
follows for (i) as one has two possible values $0$ and $1$
for each cycle of $\sg_2$.

If $\sg_1=(1,2)$, Eq.~\eqref{eq:action:case:b} applies and it is necessary and sufficient for a fixed point
to observe $x_j = \neg x_{\sigma_2^{-1}(j)}$ for each $1\leq j \leq n$. Hence the fixed points are exactly the
basis states that alternate between zero and one on each cycle in the cycle decomposition of $\sg_2$.
But this is impossible for any cycle of odd length and we obtain $\chi_{\text{nat}}(\sg_1,\sg_2) =0$ in the case (iii).
For the case (ii), we again have $\chi_{\text{nat}}(\sg_1,\sg_2) = 2^{c(\sg_2)}$ as there are also two choices
for each cycle.

\subsection{Proof of Proposition~\ref{prop:trivial:multiplicity}\label{proof:prop:trivial:multiplicity}}

Lemma~\ref{lem:triv:mult} implies case~(b)
as we only need to sum over the identity element $\symone \in \Aut$.
As all terms in the sum in Lemma~\ref{lem:triv:mult} are positive,
case~(c) is obtained by restricting the sum to the identity element $\symone \in \Aut$.
For (a), we state two formulas
\begin{subequations}
\label{eq:proof:prop:trivial:multiplicity}
\begin{alignat}{5}
&\sum_{\sg  \in \mathcal{S}_n}\;&& 2^{c(\sg)} &&= (n{+}1)! \;\text{ and}
\label{eq:proof:prop:trivial:multiplicity:a}
\\
&\sum_{\sg  \in \mathcal{S}_{2m}} && 2^{c(\sg)} \, \pp(\sg) &&= (2m)!
\label{eq:proof:prop:trivial:multiplicity:b}
\end{alignat}
\end{subequations}
where Eq.~\eqref{eq:proof:prop:trivial:multiplicity:a} is a special case 
of Proposition~1.3.7 on p.~27 in \cite{Stanley2012}.
We will prove Eq.~\eqref{eq:proof:prop:trivial:multiplicity:b} below.
The proof of (a) now continues separately for 
odd and even $n$. If $n=2m{+}1$ is odd, then $\pp(\sg)=0$ for all $\sg \in \Aut$.
Thus Lemma~\ref{lem:triv:mult} and Eq.~\eqref{eq:proof:prop:trivial:multiplicity:a} imply
for $n=2m{+}1$ that 
\begin{align*}
\mathbf{m}_{(\triv,\triv)} &= (2m{+}2)! / [2 (2m{+}1)!] = m +1 = \floor{\tfrac{n}{2}}+1.
\intertext{For $n=2m$ even, combining Eqs.~\eqref{eq:proof:prop:trivial:multiplicity:a}
and \eqref{eq:proof:prop:trivial:multiplicity:b}
proves}
\mathbf{m}_{(\triv,\triv)} &= (2m{+}1)! / [2 (2m)!] + (2m)! / [2 (2m)!] \\
&= (m{+}\tfrac{1}{2}) + \tfrac{1}{2} = m +1 = \floor{\tfrac{n}{2}}+1.
\end{align*}

Thus we are left with verifying Eq.~\eqref{eq:proof:prop:trivial:multiplicity:b}.
Let $\PP(n,k)$ denote the set of all partitions $p$ of $n$ with $k$ parts \cite{Stanley2012},
i.e., $p=[p_1,\ldots,p_k]$ where each part is given by an integer $p_j >0$ such that 
$\sum_{j=1}^{k} p_j = n$ and
$p_j \geq p_{j+1}$. Similarly, let $\PP_e(n,k) \subseteq \PP(n,k)$ denote the set of 
partitions that have only even parts $p_j$.
Also, we describe a partition $p$ 
as in Eq.~\eqref{eq:cycle:type} using its cycle type
$(1^{b_1}, \ldots, n^{b_n})$ where the multiplicity $b_a\geq 0$ of 
$a$ in $p$ is given by $b_a = \sum_{j=1}^k \delta_{a p_j}$ with $1\leq a \leq n$.
The number of permutations $\sg \in S_n$
with cycle type $(1^{b_1}, \ldots, n^{b_n})$  as determined by a partition $p$ of $n$
is equal to (see Prop.~1.3.2 on p.~23 in \cite{Stanley2012})
\begin{equation*}
n! / \zz_{(n,p)}\; \text{ where }
\zz_{(n,p)} := \prod_{a=1}^n a^{b_a}\, b_a!.
\end{equation*}
Let $2p$ denote the partition obtained by multiplying each part of $p$ by $2$.
We now compute
\begin{subequations}
\allowdisplaybreaks
\begin{align}
& \sum_{\sg  \in \mathcal{S}_{2m}} 2^{c(\sg)} \, \pp(\sg) 
= \sum_{k=1}^m \sum_{p \in \PP_e(2m,k)} \hspace{-2mm} 2^k\, (2m)! / \zz_{(2m,p)}
\label{eq:proof:prop:a}
\\
&= (2m)! \,  \sum_{k=1}^m \sum_{p \in \PP(m,k)} 2^k / \zz_{(2m,2p)}
\label{eq:proof:prop:b}
\\
&= (2m)! \, \sum_{k=1}^m \sum_{p \in \PP(m,k)} \frac{2^k}{\zz_{(m,p)}\prod_{a=1}^m 2^{b_a}} 
\nonumber
\\
&= (2m)! \, \sum_{k=1}^m \sum_{p \in \PP(m,k)} \frac{2^k}{\zz_{(m,p)}\, 2^{\sum_{a=1}^m b_a}}
\nonumber
\\
&= (2m)! \, \sum_{k=1}^m \sum_{p \in \PP(m,k)} 1/\zz_{(m,p)} = (2m)!,
\label{eq:proof:prop:e}
\end{align}
\end{subequations}
where we can limit the summation over $k$ in Eq.~\eqref{eq:proof:prop:a} to values of up to $m$ as all partitions
have even parts $p_j\geq 2$. Equation~\eqref{eq:proof:prop:b} follows as all even partitions in $\PP_e(2m,k)$
are obtained from all partitions in $\PP(m,k)$ by multiplying each part by two
and all occurring parts are less than equal to $m$. Finally, 
we conclude by comparing
Eq.~\eqref{eq:proof:prop:e}
with $\sum_{k=1}^m \sum_{p \in \PP(m,k)} m! /\zz_{(m,p)} = \abs{\mathcal{S}_m} = m!$.

\section{Properties of centers for general compact \DLAs and subalgebras  of \texorpdfstring{$\gfree$}{gfree} \label{app:center}}

In this appendix, we shortly state some useful properties related to the centers of
certain \DLAs $\g$ such that $\g$ is a general compact \DLA  (i.e.\ $\g \subseteq \uu(m)$ for a suitable $m$),
$\g$ is contained in $\gfree$, or $\g$ is equal to $\gstd$.
Recall that any element $g$ contained in a compact \DLA 
$\g  \subseteq \uu(m)$ can be uniquely decomposed
into $g = s + c$ where $s \in [\g,\g]$ is contained in the semisimple
part $[\g,\g]$ of $\g$ and $c \in \cent(\g)$ is contained in the center $\cent(\g)$
of $\g$. This is a consequence of $\g$ being contained in a unitary \DLA 
and thus being reductive \cite{BourbakiLie1989,Bourbaki2008b}.
The number of elements in a generating set for a compact \DLA
constraints the dimension of its center:

\begin{lemma}
\label{lem:dim-center-at-most-num-generators}
Consider a \DLA $\lie{g_1,\ldots,g_k} = \g \subseteq \uu(m)$ that is contained in a unitary \DLA $\uu(m)$
and that is generated by $k$ elements $g_1, \ldots, g_k$.
Then the dimension $\abs{\cent(\g)}$  of the center $\cent(\g)$ of $\g$ is less than equal to
the number $k$ of generators, i.e., $\abs{\cent(\g)}\leq k$.
\end{lemma}

\begin{proof}
The generators $g_j = s_j + c_j$ can be uniquely decomposed
into $s_j \in [\g,\g]$ and $c_j \in \cent(\g)$.
Clearly, $[g_j,g_k] \subseteq [\g,\g]$ and all $s_j$ are obtained from the generators via linear combinations
of higher-order commutators. Thus, we obtain $c_j = g_j -s_j$ and $\cent(\g)$ is spanned by the elements
$c_j$, which completes the proof.
\end{proof}

Interestingly, we can interpret the center of a compact \DLA
as a subset of the center of its commutant:

\begin{figure}[b]
\includegraphics{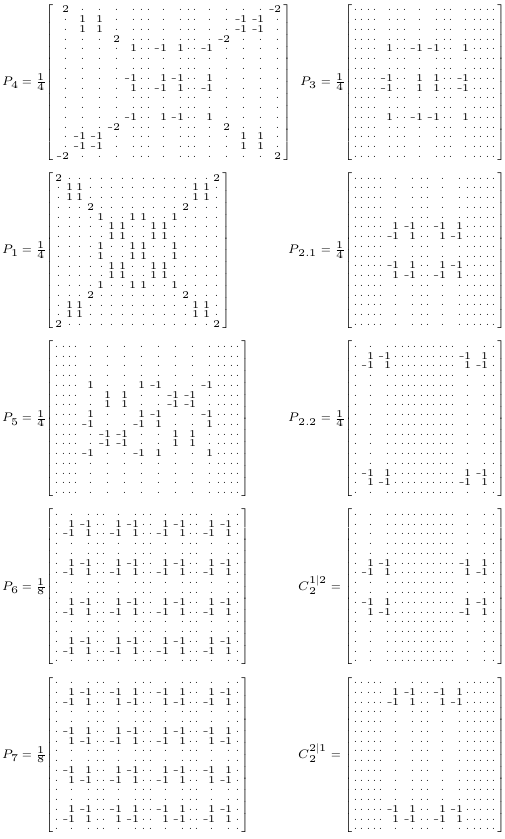}
\caption{\textbf{Explicit matrices for the projectors $P_j$ and cross terms  $C_2^{k|\ell}$ from Fig.~\ref{fig:decomposition:four}.}
Dots denote zeros.
\label{fig:matrices:four}}
\end{figure}

\begin{lemma}\label{lemma:cent:cent}
Given a compact \DLA $\g \subseteq \uu(m)$
and its commutant $\CC= \com(\g)$, we have
(i) $\cent(\g) = \CC \cap \g \subseteq \CC$ and (ii) $\cent(\g) \subseteq \cent(\CC)$.
\end{lemma}

\begin{proof}
The case (i) is obvious from the definitions.
Assume that there exists $g \in \cent(\g)$ and $M \in \CC$
such that $[g,M] \neq 0$. But this conflicts with the definition of the commutant
$\Cstd$ and the statement follows.
\end{proof}

This allows us to bound the dimension of the centers
for \DLAs that include the standard-ansatz and the free \DLAs $\gstd$ and $\gfree$:

\begin{lemma}
\label{lem:std:free:center}
Given a \DLA $\g \subseteq \gfree$, 
the dimension $\dim[\cent(\g)]$ of its center $\cent(\g)$ is bounded as
$\dim[\cent(\g)] \le \dim[\cent(\CC)]{-}2$
where $\CC = \CC(\g)$ is the commutant of $\g$.
\end{lemma}

\begin{proof}
Clearly.
$\g \subseteq \gfree$ and $\CC \supseteq \Cfree$.
Lemma~\ref{app:prop:center}(a) and (b) now imply
$iX\tn,iI\tn \notin \g$ and $iX\tn,iI\tn \in \CC$.
Clearly, $iX\tn,iI\tn \notin \cent(\g)$ follows.
This then verifies
$\dim[\cent(\g)] \le \dim(\CC){-}2$.

Lemma~\ref{lemma:cent:cent} implies $\cent(\g) \subseteq \cent(\CC)$.
The statement of the current lemma is immediate after we have shown
$iI\tn,iX\tn \in \cent(\CC)$.
We have $iI\tn \in \cent(\CC)$.
Moreover, 
$iX\tn \in \cent(\CC)$ follows, e.g., as the projectors $P_{\pm}=
(X\tn {\pm} I\n)/2$ are linear combinations of the isotypical projectors contained in $\cent(\CC)$
associated to the isotypical decomposition of $\CC$ (see Appendix~\ref{app:symmetry:theory}).
\end{proof}

We now bound the dimension of the center
$\cent(\gstd)$ for the standard-ansatz \DLA $\gstd$:
\begin{lemma}
\label{lem:std:center-dim-at-most-2}
Given the \DLA $\gstd$ of the standard ansatz,
the dimension $\dim[\cent(\gstd)]$ of its center $\cent(\gstd)$ is bounded as
$\dim[\cent(\gstd)] \le \min\{2,\dim[\cent(\Cstd)]{-}2\}$
where $\Cstd = \CC(\gstd)$ is the commutant of $\gstd$.
\end{lemma}

\begin{proof}
Lemma~\ref{lem:dim-center-at-most-num-generators} directly verifies
$\dim[\cent(\gstd)] \le 2$.
Lemma~\ref{lem:std:free:center} completes the proof.
\end{proof}

\section{Explicit projection matrices from Figures \ref{fig:decomposition:house} and \ref{fig:decomposition:four}\label{app:projections}}

The explicit form of the projection matrices from 
Figures~\ref{fig:decomposition:house} and \ref{fig:decomposition:four}
are shown in Figures~\ref{fig:matrices:house} and \ref{fig:matrices:four}, respectively.


%

\end{document}